\documentclass[11pt,a4paper]{amsart}
\usepackage[left=3cm,bottom=2.8cm,right=3.1cm,top=3.2cm]{geometry}
\usepackage{amsmath} % Umgebugen für abgesetzte Formeln
\usepackage{graphicx}
\usepackage{subfigure}
\usepackage[margin=10pt,font=small,labelfont=bf,labelsep=endash]{caption}
\usepackage[utf8]{inputenc}
\usepackage{amssymb}
\usepackage{mathrsfs}
\usepackage{datetime}
\usepackage{color}
\usepackage{float}
\usepackage{tikz}
\usetikzlibrary{decorations.pathmorphing,calc,intersections,backgrounds,patterns}
\newcommand{\bc}{\begin{cases}}
\newcommand{\ec}{\end{cases}}

\newcommand{\Ab}[1]{\|#1\|}
\newcommand{\ab}[1]{|#1|}

\newcommand{\La}{\Lambda}

\numberwithin{equation}{section}

\newtheorem{lemma}{Lemma}[section]
\newtheorem{theorem}[lemma]{Theorem}
\newtheorem{rem}[lemma]{Remark}

\begin{document}
\title[Static solns to the E-V system with non-vanishing cosmological constant]{Static solutions to the Einstein-Vlasov system with non-vanishing cosmological constant }
\author[H.~Andr\'easson, D.~Fajman, M.~Thaller]{H\r{a}kan Andr\'easson, David Fajman, Maximilian Thaller}
%\curraddr{$^{1}$University of Gothenburg, Chalmers University of Technology, Department of Mathematics, S-412 96 G\"oteborg, Sweden\\$^{2}$Gravitational Physics, Faculty of Physics\\
% University of Vienna\\
% Boltzmanngasse 5, 1090 Vienna \\
% Austria}
\date{\today}
%\email{hand@chalmers.se, David.Fajman@univie.ac.at}

\maketitle
\vspace{-1cm}
\begin{abstract}
We construct spherically symmetric, static solutions to the Einstein-Vlasov system with non-vanishing cosmological constant $\Lambda$. The results are divided as follows. For small $\Lambda>0$ we show existence of globally regular solutions which coincide with the Schwarzschild-deSitter solution in the exterior of the matter sources. For $\Lambda<0$ we show via an energy estimate the existence of globally regular solutions which coincide with the Schwarzschild-Anti-deSitter solution in the exterior vacuum region. We also construct solutions with a Schwarzschild singularity at the center regardless of the sign of $\Lambda$. For all solutions considered, the energy density and the pressure components have bounded support.
%The solutions constructed in this work describe Schwarzschild-deSitter and Schwarzschild-Anti-deSitter spacetimes with compactly supported matter around a regular center or with immersed matter shells and black holes.
Finally, we point out a straightforward method to obtain a large class of globally non-vacuum spacetimes with topologies $\mathbb R\times S^3$ and $\mathbb R\times S^2\times \mathbb R$ which arise from our solutions using the periodicity of the Schwarzschild-deSitter solution. A subclass of these solutions contains black holes of different masses.

\end{abstract}

\section{Introduction}
Schwarzschild's construction of a static explicit solution in 1915 was the first example of a solution to Einstein's field equations in general relativity~\cite{SC16}.
%Schwarzschild's construction of a static solution to Einstein's field equations has been the first approach to study the %theory of General Relativity via explicit solutions \cite{SC16}.
It has been found later that the class of static, spherically symmetric, asymptotically flat solutions to Einstein's equations in vacuum consists only of that element \cite{BUMA87} making it necessary to consider the non-vacuum field equations to construct further classes of spherically symmetric static spacetimes.
%In general the assumption of symmetry is not necessary to allow for static solutions as has been discovered recently for %the Einstein equations coupled to elastic matter by Andersson, Beig and Schmidt \cite{ANBESC08}. However, in this work %we restrict to the class of spherically symmetric spacetimes and focus on a specific matter model.
%
\subsection{Static solutions with Vlasov matter}
%The model which is considered in the present work is
In this work we consider matter described as a collisionless gas. In astrophysics this model is used to study galaxies and globular clusters where the stars, or the galaxies, are the particles of the gas and where collisions between these are sufficiently rare to be neglected. The particles interact by the gravitational field which the particle ensemble creates collectively. Within the framework of general relativity the particle system is described by the Einstein-Vlasov system. The mathematical investigation of this system was initiated by Rein and Rendall in 1992 \cite{rr92} in the context of the Cauchy problem and shortly thereafter the same authors provided the first study of static, spherically symmetric solutions to this system \cite{static}. Since then, the Einstein-Vlasov system has been successfully studied in several contexts and many global results have been obtained during the last two decades. We refer to~\cite{abriss} for a review of these results but let us in particular mention the recent monumental work on this system concerning the stability of the universe~\cite{R13}.

%such as the non-linear stability problem (
%cf.~\cite{a99,R13}) and the formation of singularities (cf.~\cite{rrs95,ak10,dr05,akr10}).
 The purpose of the present work is to extend the class of static solutions to the Einstein-Vlasov system to the case with a non-vanishing cosmological constant $\Lambda$. Several results on static and stationary solutions to this system have been obtained in the case when $\Lambda=0$. The first result of this kind was provided in~\cite{static}, where the authors construct spherically symmetric isotropic static solutions with compactly supported energy density and pressure. The solutions are asymptotically flat and thus serve as models for isolated, self-gravitating systems. Several generalizations of this result have since then been obtained, in particular solutions with non-isotropic pressure, and solutions with a Schwarzschild singularity at the center, have been established, cf.~\cite{rein94,rr98}. An approac by variational methods was developed by Wolansky \cite{w01}. The most difficult part in these proofs is to show that the matter has compact support. A neat and quite general method to treat this problem has recently been obtained by Ramming and Rein in~\cite{rr12}. However, this method does not straightforwardly apply to the situation we consider in this work. The cosmological constant changes the structure of the equations and this implies that inequality (1.23) in~\cite{rr12}, on which this method is based, does not hold when $\Lambda\neq 0$. Hence, we rely on a different method in this work. The results discussed above all concern the spherically symmetric case. Let us point out that results beyond spherical symmetry have been established. The existence of stationary axially symmetric solutions to the Einstein-Vlasov system has recently been shown, cf.~\cite{ANKURE11} and~\cite{ANKURE14} for the non-rotating and the rotating case respectively. In this context we also mention a result on static solutions for elastic matter which has been obtained without any symmetry assumption~\cite{ANBESC08}.

\subsection{Static solutions with non-vanishing cosmological constant}
A specific class of solutions has so far not been discussed which concerns the Einstein equations with a non-vanishing cosmological constant $\La$. The model solutions for the vacuum equations are the Schwarzschild-deSitter and Schwarzschild-Anti-deSitter (Schwarzschild-AdS) solution for $\La>0$ and $\La<0$, respectively. Einstein's equations with non-vanishing $\Lambda$ are of significant physical interest, where the case $\La>0$ applies to a universe with accelerated expansion \cite{R13}, while the case $\La<0$ is relevant in the context of AdS-CFT correspondence \cite{hr00}.
Concerning the Einstein-Vlasov system no existence results for the static Einstein equations with non-vanishing cosmological constant are known. The aim of the present paper is to prove existence of spherically symmetric static solutions to the Einstein-Vlasov system with small positive or arbitrary negative cosmological constant. The solutions we construct are in general anisotropic. The results provided in this work are as follows.
%We consider a general class of anisotropic distribution functions following the ansatz of \cite{rein94} and the spherically symmetric Einstein equations  and show existence for the following classes of solutions.
%
\subsubsection{Globally regular solutions for $0<\La\ll1$} We construct globally regular static solutions for small $\La>0$.
%A local existence theorem for small radii follows by standard arguments. A maximal existence interval of this solution is characterized by a continuation criterion using boundedness of coefficients of the evolution equation for the free parameter $y=\ln E_0-\mu$, where $\mu$ is the logarithm of the lapse function and $E_0$ the cutoff energy of the chosen static ansatz for the distribution function.
The fundamental difference to the case of vanishing cosmological constant is that for large radii the metric tends towards a cosmological horizon and it is thus necessary to show that the support of the matter quantities vanishes before the cosmological horizon is reached. We show that for small $\La>0$ the solutions we construct are close to the solutions corresponding to the $\La=0$ case for which the matter quantities have compact support, and in addition, the latter solutions obey a Buchdahl type inequality. These facts imply that the support of the matter quantities can be controlled also in the case when $\Lambda>0$. It is then possible to continue the solution from the vacuum region by a Schwarzschild-deSitter solution. This method yields a large class of globally regular solutions which coincide with a Schwarzschild-deSitter solution outside a compact set. The result is given in Theorem \ref{main-thm}.
\subsubsection{Globally regular solutions for $\La<0$}
The case of negative cosmological constant is a priori simpler since the cosmological term has a good sign which yields a monotonically decreasing behavior of the lapse function. An energy argument following the general idea of \cite{rein94} is used to establish global in $r$ existence yielding globally regular solutions for general $\La<0$.
The result is given in Theorem \ref{theo_glo_adS}.

\subsubsection{Solutions with a Schwarzschild singularity for $0<\La\ll1$}
To construct solutions with singularities in the center, we start with the vacuum equations which can be solved explicitly by the Schwarzschild-deSitter solution. This solution is considered until a radius which allows to continue the vacuum solution by one which at the same point satisfies the ansatz for the distribution function and eventually merges into a non-vacuum region. It is shown that the support of the matter quantities is compact and outside the matter region the solution can again be extended by a vacuum solution with mass parameter corresponding to the interior mass of black hole and matter. As in the non-singular case these constructions only work out for sufficiently small $\La>0$. The result is given in Theorem \ref{theo_la_pos_bh}.

\subsubsection{Solutions with a Schwarzschild singularity for $\La<0$}
This point is similar to the case $\La>0$ with Schwarzschild singularities. In particular
a smallness condition for $\ab{\La}$ is needed as well. The result is given in Theorem
\ref{theo_la_neg_bh}.
\subsubsection{Solutions with topologies $\mathbb R\times S^3$ and $\mathbb R\times S^2\times\mathbb R$}
A significant generalization of the results with $\La>0$ is presented in the final section. The periodic structure of the Schwarzschild-deSitter space \cite{GiHa77} allows us to consider solutions with regular massive center, and solutions with central black holes, and glue them to a periodic Schwarzschild-deSitter solution with a black hole region followed by another matter region - forming a static space-time with two non-vacuum ends and a black hole (or several) in between. The result is given in Theorem \ref{thm-top}.

\subsection{Outline of the paper}
This paper is organized as follows. In Section \ref{sec : prel} we introduce the notation and give a short review on the static Einstein-Vlasov system in spherical symmetry. We discuss the anisotropic ansatz for the distribution function, variations of which are used in this work. A Buchdahl type inequality, which applies to solutions of the Einstein-Vlasov system, is then reviewed shortly as it is used later in the existence proof for $\La>0$. The Einstein-Vlasov system in spherical symmetry with a specific ansatz for the distribution function reduces to an integro-differential equation given in (\ref{eqsimy}). This equation lies at the heart of the analysis in the paper. In Section \ref{sec : sagr} we prove existence of globally regular solutions for small $\La>0$. The proof is divided into several steps beginning with local in $r$ existence in \ref{ssec : sagr-1}, a continuation criterion in \ref{ssec : sagr-2}, existence for sufficiently large radii to reach a vacuum region in \ref{secsupport} and finally the proof of the existence Theorem in \ref{ssec : sagr-4}. In Section \ref{sec : sage neg} the existence of globally regular solutions for arbitrary $\La<0$ is proven along with a result (cf.~Theorem \ref{theo_glo_adS}) which states existence for such solutions outside a ball, which eventually is used to prove existence of solutions with Schwarzschild singularities in the center. Section \ref{sec : swss} begins with a generalization of the Buchdahl type inequality, mentioned above, for solutions with Schwarzschild singularities. This result is useful for the construction of solutions of this kind when $\La>0$. These solutions are obtained in Theorem \ref{theo_la_pos_bh}. Analogous solutions for the case of negative $\La$ are given in Theorem \ref{theo_la_neg_bh}. Finally, Section \ref{sec : glob sol} discusses the globally non-trivial generalizations of the constructed solutions for $\La>0$.

\begin{center}
\textbf{Acknowledgements}
\end{center}
D.F.~and M.T.~are grateful to Walter Simon and Bobby Beig for several discussions on static solutions. We thank Piotr Chrus\'ciel for the suggestion to study the case of negative cosmological constants. We are indebted to Greg Galloway for sharing his ideas and suggestions concerning the global solutions which we present in the last chapter of this work. We thank Christa \"Olz for helpful discussions. D.F.~thanks Amos Ori and Gershon Wolansky for interesting discussions.

%%%%%%%%%%%%%%%%%%
%%%%%%%%%%%%%%%%%%
%%%%%%%%%%%%%%%%%%

\section{Preliminaries}\label{sec : prel}
\subsection{Setup and notations}
We consider the Einstein-Vlasov system with cosmological constant $\Lambda\in\mathbb R$. For background on this system we refer to~\cite{abriss}. For the spherically symmetric, static Lorentzian metric $g$ we use the standard ansatz
\begin{equation}
\mathrm ds^2 = -e^{2\mu(r)} \mathrm dt^2 + e^{2\lambda(r)}\mathrm dr^2 + r^2\mathrm d\vartheta^2 + r^2\sin^2(\vartheta)\mathrm d\varphi^2.
\end{equation}
Assuming in addition the matter distribution function $f$ to be spherically symmetric and static one obtains the reduced system of equations
\begin{eqnarray} \label{eqvlasov}
\frac{v^a}{\sqrt{1+|v|^2}}\frac{\partial f}{\partial {x^a}}-\sqrt{1+|v|^2}\mu'\frac{x^a}{r} \frac{\partial f}{\partial {v^a}} &=& 0,\\
e^{-2\lambda}(2r\lambda'-1)+1-r^2\Lambda &=& \kappa r^2\varrho, \label{eeq1}\\
e^{-2\lambda}(2r\mu'+1)-1+r^2\Lambda &=& \kappa r^2 p, \label{eeq2}
\end{eqnarray}
where $\kappa = 8\pi$, $|v|=\sqrt{\delta_{ij}v^iv^j}$, $v_r=\frac{\delta_{ij}v^ix^j}{r}$ and the matter quantities read
\begin{eqnarray}
\varrho &=&\int_{\mathbb R^3} f(x,v)\sqrt{1+|v|^2}\; \mathrm dv^1\mathrm dv^2\mathrm dv^3,\label{rho} \\
p &=& \int_{\mathbb R^3}\frac{f(x,v)}{\sqrt{1+|v|^2}}v_r^2\;\mathrm dv^1\mathrm dv^2\mathrm dv^3.\label{p}
\end{eqnarray}
There is an additional Einstein equation
\begin{equation}
e^{-2\lambda}\left(\mu''\left(\mu+\frac 1 r\right)\left(\mu'-\lambda'\right)\right)=\kappa p_T, \label{eeq3}
\end{equation}
where
\begin{equation}\label{pt}
p_T = \frac 1 2 \int_{\mathbb R^3}\left|\frac{x \times v}{r}\right|^2 f(x,v)\frac{\mathrm dv}{\sqrt{1+|v|^2}}.
\end{equation}
The quantity $\varrho$ can be understood as energy density, $p$ as radial pressure and $p_T$ as tangential pressure.
%A detailed derivation of the system \ref{eqvlasov}-\ref{pt} in the $\Lambda=0$ case can be found in \cite{rr92}.
To ensure a regular center the following boundary condition is imposed
\begin{equation}\label{boundary1}
\lambda(0)=0.
\end{equation}
%and an asymptotically flat spacetime is obtained if
%\begin{equation}\label{boundary2}
%\lim_{r\to\infty}\lambda(r)=\lim_{r\to\infty}\mu(r)=0,
%\end{equation}
This condition will be used in the first part of this work but when we consider solutions with a Schwarzschild singularity at the center it will be dropped.
A detailed derivation of the system (\ref{eqvlasov})-(\ref{pt}) in the $\Lambda=0$ case can be found in \cite{rr92}.
It will be seen below that a solution of the reduced system (\ref{eqvlasov})-(\ref{p}) also solves the full system.
%For the discussion of the Einstein-Vlasov system with non vanishing cosmological constant a third matter quantity, the tangential pressure $p_T$, will be of importance. It is given by
Considering the characteristic curves of the Vlasov equation \eqref{eqvlasov} one can simplify the system of equations. Along these characteristic curves the quantities $E$ and $L$, given by
\begin{equation}
E = e^{\mu(r)}\sqrt{1+|v|^2} =: e^{\mu(r)}\varepsilon \quad \mathrm{and} \quad L=| x \times v |^2,
\end{equation}
are conserved (cf.~\cite{static}). Therefore any ansatz for the matter distribution $f$ of the form
\begin{equation}
f(x,v) = \Phi(E,L)
\end{equation}
solves the Vlasov equation \eqref{eqvlasov}, and this equation drops out from the system of equations.
\subsection{Relevant results}
In the following we discuss the known results for the Einstein-Vlasov system with vanishing cosmological constant, $\La=0$, which are relevant for the work presented in this paper. The existence of a unique solution $\mu(r)$, $\lambda(r)$ to given initial values $\mu(0)=\mu_0$ and $\lambda(0) = 0$ has been proved using the ansatz
\begin{equation} \label{ansfg}
f(x,v) = \Phi(E)[L-L_0]_+^\ell,
\end{equation}
where $E>0$, $L>0$, $L_0\geq 0$, $\ell>-\frac 12$, $\Phi\in L^\infty((0,\infty))$ for the matter distribution $f$, cf.~\cite{rein94}.
Furthermore, it can be shown that the support of the matter quantities is contained in an interval $[0,R_0]$, $0<R_0<\infty$, if one takes a so called {\em polytropic} ansatz for $f$. This ansatz has the form
\begin{equation} \label{ansfvs}
f(x,v) = \phi\left(1-\frac{E}{E_0}\right)L^\ell,
\end{equation}
where $\phi:\mathbb R\to[0,\infty)$ is measurable, $\phi(\eta)=0$ for $\eta < 0$, and $\phi > 0$ a.e.~on some interval $[0,\eta_1]$ with $\eta_1 > 0$ and $E_0$ is some prescribed cut-off energy, cf.~\cite{rr12}. Moreover, it is required that there exists $\gamma > -1$ such that for every compact set $K \subset \mathbb R$ there exists a constant $C>0$ such that
\begin{equation}
\phi(\eta) \leq C\eta^\gamma,\quad\eta\in K.
\end{equation} In \cite{rein98} this result is generalized to anisotropic matter distributions of the form
\begin{equation}
f(x,v)=c_0[E_0-E]_+^k[L-L_0]_+^\ell,
\end{equation}
where $k\geq0$, $\ell>-1/2$ fulfill the inequality $k < 3\ell+7/2$ and $c_0,E_0 > 0$, $L_0\geq 0$. It is shown in \cite{rein98} that for sufficiently small $L_0$ the support of $f$ is contained in an interval $[R_i,R_0]$ where $0\leq R_i < R_0 < \infty$ and $R_i > 0$ provided $L_0>0$.
\par
By direct calculation one shows that the matter quantities fulfill the generalized Tolman-Oppenheimer-Volkov equation (TOV equation)
\begin{equation} \label{toveq}
p'(r) = -\mu'(r) (p(r)+\varrho(r)) - \frac 2 r (p(r) - p_T(r)).
\end{equation}
Another result relevant for the proof presented here is a generalized Buchdahl inequality \cite{and08}, which is the content of the following lemma.
\begin{lemma}[Theorem 1 in \cite{and08}]\label{theobuchdahl}
Let $\lambda,\mu\in C^1([0,\infty))$ and let $\varrho, p, p_T\in C^0([0,\infty))$ be functions that satisfy the system of equations (\ref{eeq1})-(\ref{eeq3}), the condition (\ref{boundary1}) and such that $p+2p_T \leq \varrho$.
Then
\begin{equation}\label{AB}
\sup_{r>0}\frac{2m(r)}{r} \leq \frac89,
\end{equation}
%\begin{equation} \label{buch_cond_1}
%\left(m(r) + 4\pi r^3 p\right)e^{\mu+\lambda}=4\pi \int_0^r s^2e^{\mu+\lambda}(\varrho+p+2p_T)\mathrm ds
%\end{equation}
%a.e., where
where
\begin{equation}
m(r) = 4\pi\int_0^rs^2\varrho(s)\mathrm ds.
\end{equation}
%If there exists a number $\Omega > 0$ such that $p+2p_T \leq \Omega\varrho$ a.e.~then
%\sup_{r>0}\frac{2m(r)}{r} \leq \frac{(1+2\Omega)^2-1}{(1+2\Omega)^2}.
\end{lemma}

%\begin{rem}
%For any energy momentum tensor satisfying the dominant energy condition the assumptions of lemma \ref{theobuchdahl} are already fulfilled with $\Omega = 3$. In the case of Vlasov matter $p+2p_T \leq \varrho$ holds, thus
%\begin{equation}
%\sup_{r>0}\frac{2m(r)}{r} \leq \frac{8}{9}.
%\end{equation}
%Note also that the TOV equation \eqref{toveq} in combination with the Einstein equations \eqref{eeq1}, \eqref{eeq2} implies the condition \eqref{buch_cond_1}.
%\end{rem}

\begin{rem}
The inequality (\ref{AB}) holds for a more general class of functions, cf.~\cite{and08}. Moreover, the inequality is sharp, and the solutions which saturate the inequality are infinitely thin shell solutions, cf.~\cite{and08}. In~\cite{and07} it is shown that there exist regular, arbitrarily thin, shell solutions to the Einstein-Vlasov system such that the quantity $2m/r$ can be arbitrarily close to $8/9$. It should also be mentioned that Buchdahl type inequalities have been obtained in the case of non-vanishing cosmological constant, cf.~\cite{ab,abm}. These results {\em assume} the existence of static solutions to the Einstein-matter equations with a cosmological constant.
\end{rem}
To prove existence of solutions of the static Einstein-Vlasov system with non-vanishing $\Lambda $ we make use of the results discussed above. To simplify calculations we define $y := \ln(E_0)-\mu$ as in \cite{rr12} so that $e^\mu=E_0/e^y$. For the distribution function $f$ we choose the ansatz\footnote{To be precise any $\phi$ that is of the kind of the $\phi$ in \eqref{ansfvs} would meet the assumptions of the following lemmas and theorems.}
\begin{equation}\label{ouransf}
\begin{aligned}
f(x,v) &= \Phi(E,L) = c_0\phi\left(1-\frac{E}{E_0}\right)[L-L_0]_+^\ell\\
& = c_0\phi\left(1-\varepsilon e^{-y}\right)[L-L_0]_+^\ell, \\
 \phi(\eta) &=[\eta]_+^k,
\end{aligned}
\end{equation}
where $k\geq0$, $\ell>-1/2$ fulfill the inequality $k < 3\ell+7/2$ and $c_0,E_0 > 0$, $L_0\geq 0$.
%The cutoff energy $E_0$ cannot be chosen arbitrarily but is determined by the initial value $y_0$ of $y$ and the boundary conditions.
For the construction of globally regular solutions $L_0$ has to be sufficiently small to ensure finite support of the matter quantities \cite{rein98}. When considering solutions with a black hole at the center, there are positive lower bounds on $L_0$. The expressions for the matter quantities $\varrho$ and $p$ take the form
\begin{equation}
\varrho(r) = G_\phi(r,y(r)), \qquad p(r) = H_\phi(r,y(r)),
\end{equation}
where
\begin{align}
G_\phi(r,y) &= c_\ell c_0 r^{2\ell} \int_{\sqrt{1+L_0/r^2}}^\infty \phi\left(1-\varepsilon e^{-y}\right) \varepsilon^2 \left(\varepsilon^2-\left(1+\frac{L_0}{r^2}\right)\right)^{\ell+\frac{1}{2}}\mathrm d\varepsilon, \label{defgphi}\\
H_\phi(r,y) &= \frac{c_\ell c_0}{2\ell+3} r^{2\ell} \int_{\sqrt{1+L_0/r^2}}^\infty \phi\left(1-\varepsilon e^{-y}\right) \left(\varepsilon^2-\left(1+\frac{L_0}{r^2}\right)\right)^{\ell+\frac{3}{2}}\mathrm d\varepsilon, \label{defhphi}
\end{align}
given in \cite{rein94}.
The constant $c_\ell$ is given by
\begin{equation}
c_\ell=2\pi\int_0^1 \frac{s^\ell}{\sqrt{1-s}} \mathrm ds.
\end{equation}
\begin{lemma} \label{lemmagh}
The functions $G_\phi(r,y)$ and $H_\phi(r,y)$ defined in \eqref{defgphi} and \eqref{defhphi}, respectively, have the following properties.
\global\long\def\theenumi{\roman{enumi}}
\begin{enumerate}
\item $G_\phi(r,y)$ and $H_\phi(r,y)$ are continuously differentiable in $r$ and $y$. \label{ghdiff}
\item The functions $G_\phi(r,y)$ and $H_\phi(r,y)$ and the partial derivatives $\partial_y G_\phi(r,y)$ and $\partial_y H_\phi(r,y)$ are increasing both in $r$ and $y$. \label{ghinc} %Consult \cite{static} for details.
\item There is vacuum, i.e.~$f(r,\cdot)=p(r)=\varrho(r)=0$ if $e^{-y(r)}\sqrt{1+L_0/r^2} \geq 1$, in particular if $y(r) \leq 0$. \label{ghvac}
\end{enumerate}
\end{lemma}
\begin{proof}
By performing a change of variables in the integrals in \eqref{defgphi} and \eqref{defhphi} the differentiability follows, cf.~\cite{rein94}, Lemma 3.1. The monotonicity can be seen directly from the structure of $G_\phi$ and $H_\phi$. The last statement is obvious since $\phi(\eta)=0$ if $\eta\leq 0$.
%if one performs a change of variables in the integrals in \eqref{defgphi} and \eqref{defhphi}, cf.~\cite{rein94.
\end{proof}
\subsection{Main equation}
From the Einstein equations (\ref{eeq1}) and (\ref{eeq2}) one obtains the differential equation for $y$
\begin{equation} \label{eqsimy}
\begin{aligned}
y'(r) &= -\frac{\kappa/2}{1-\frac{\Lambda r^2}{3}-\frac{\kappa}{r}\int_0^r s^2 G_\phi(s,y(s))\mathrm ds}\\
&\qquad\qquad\qquad\qquad\qquad \times\left(r H_\phi(r,y(r))- \frac{2r\Lambda}{3\kappa}+\frac{1}{r^2}\int_0^r s^2 G_\phi(s,y(s))\mathrm ds\right).
\end{aligned}
\end{equation}
A solution to (\ref{eqsimy}) yields a solution to the system (\ref{eqvlasov}-\ref{p}). It should however be pointed out that in order to obtain an asymptotically flat solution one needs to redefine $E_0$ and $\mu$ as follows. Given an initial value $y_0$, a solution $y$ of equation (\ref{eqsimy}) is obtained having a limit $y(\infty)$. By letting $E_0:=1/y(\infty)$ and $e^{\mu}:=E_0y(r)$ we get a solution with the proper boundary condition at infinity. Furthermore it should be mentioned that a solution to the system (\ref{eqvlasov}-\ref{p}) provides a solution to all the Einstein equations. This is shown in Theorem 2.1 in \cite{rr92} in the case when $\La=0$. The proof is analogous in the case with non-vanishing $\La$. The equation (\ref{eqsimy}) is analyzed and solved in the remainder of this work.

\section{Static, anisotropic globally regular solutions for $\La>0$} \label{sec : sagr}
In this section we prove existence of globally regular static solutions with small $\La>0$.
\subsection{Local existence} \label{ssec : sagr-1}
The following local existence lemma corresponds to the first part of the proof of Theorem 2.2 in \cite{static} for the case $\La=0$.
\begin{lemma} \label{lemloc}
Let $\Phi:\mathbb R^2\to[0,\infty)$ be of the form \eqref{ouransf} and let $G_\phi$, $H_\phi$ be defined by equations (\ref{defgphi}) and (\ref{defhphi}), respectively. Then for every $y_0 \in \mathbb R$ and every $\Lambda >0$ there is a $\delta>0$ such that there exists a unique solution $y_\Lambda \in C^2([0,\delta])$ of equation \eqref{eqsimy} with initial value $y_\Lambda(0)=y_0$.
\end{lemma}

\begin{proof}
We consider the equation (\ref{eqsimy}) and integrate it using the initial condition $y_\Lambda(0)=y_0$. The following fixed point problem is obtained,
\begin{equation}
y_\Lambda(r) = (Ty_\Lambda)(r), \quad r\geq 0
\end{equation}
where the operator $T$ is given by
\begin{multline}
(Tu)(r) :=
y_0 - \int_0^r \frac{\kappa/2}{1-\frac{s^2\Lambda}{3}-\frac{\kappa}{s}\int_0^s\sigma^2 G_\phi(\sigma,u(\sigma))\mathrm d\sigma} \\ \times\left(s H_\phi(s,u(s))-\frac{2s\Lambda}{3\kappa}+\frac{1}{s^2}\int_0^s\sigma^2 G_\phi(\sigma,u(\sigma))\mathrm d\sigma\right)\mathrm ds. \label{eqdeft}
\end{multline}
This operator is considered on the set
\begin{multline}
M:=\Big\{u:[0,\delta]\to\mathbb R\;|\; u(0)=y_0,y_0-1\leq u(r)\leq y_0+1, \\
 \frac{r^2\Lambda}{3}+\frac{\kappa}{r}\int_0^rs^2 G_\phi(s,u(s))\mathrm ds \leq c < 1, r\in[0,\delta]\Big\}. \label{eqdefm}
\end{multline}
We note that $M$ is non-empty if $\delta>0$ is chosen sufficiently small.
As carried out in detail in the appendix, Section \ref{apptcontr}, it is shown that $T$ acts as a contraction on $M$. %, i.e.
%\global\long\def\theenumi{\roman{enumi}}
%\begin{enumerate}
%\item $u \equiv y_0\in M$,
%\item $u \in M \Rightarrow T u \in M$ and
%\item $\exists a \in \; (0,1) \;\forall u,v\in M:\;\Ab{Tu-T v}_{\infty,\delta} \leq a \Ab{u-v}_{\infty,\delta}$, where $\Ab{.}_{\infty,\delta}=\sup_{r\in[0,\delta]}(.)$.
%\end{enumerate}
This implies (by the Banach fixed-point theorem) that there exists $y_\Lambda\in M$ such that $T y_\Lambda = y_\Lambda$. Differentiability of $y_\Lambda$ follows from the structure of $T$. The differentiation with respect to $r$ yields that $y_\Lambda$ solves equation (\ref{eqsimy}) on the interval $[0,\delta]$. Away from the singularity $r=0$, standard existence and uniqueness results are applied to extend $y_\Lambda$ to a maximal solution on an interval $[0,R_c)$. Obviously, the boundary condition at $r=0$ is satisfied. The regularity of the functions $G_\phi$ and $H_\phi$ implies that $y_\Lambda\in C^2((0,R_c))$, (cf.~\cite{rein94}) and it can be shown that the second derivative continuously extends to $r=0$ and $y_\Lambda'(0)=0$.
\end{proof}

\subsection{Continuation criterion}\label{ssec : sagr-2}
%The characterization of the maximal existence interval of the local solutions constructed in lemma \ref{lemloc} is relevant to assure for the existence of solutions beyond the non-vacuum region proven in the next section. For this purpose a continuation criterion is proved in this section. Define $R_c$ as the maximal radius such that the unique local solution $y_\Lambda$ of equation (\ref{eqsimy}) can be continued to the interval $[0,R_c)$. If $R_c < \infty$ then either
%\begin{equation}
%\liminf_{r\to R_c} \left(1-\frac{r^2\Lambda}{3}-\frac{\kappa}{r}\int_0^rs^2 G_\phi(s,y_\Lambda(s))\mathrm ds\right) = 0
%\end{equation}
%or
%\begin{equation}
%\limsup_{r\to R_c} \left|y_\Lambda(r)\right| = \infty,
%\end{equation}
%which follows straightforward from equation (\ref{eqsimy}). In the case $\Lambda = 0$ it can be shown that $R_c=\infty$, cf.~\cite{rein94}. In the case of a positive cosmological constant, $R_c$ may be finite. However,
The solution $y_\Lambda$ exists at least as long as the denominator of the right hand side of equation (\ref{eqsimy}) is strictly larger than zero. The following lemma formulates this assertion.

\begin{lemma} \label{lemden}
Let $y_0\in\mathbb R$ and let $R_c > 0$ be the largest radius such that the unique local $C^2$-solution $y_\Lambda$ of equation \eqref{eqsimy} with $y_\Lambda(0)=y_0$ exists on the interval $[0,R_c)$. Then there exists $R_D \leq R_c$ such that
\begin{equation} \label{lemden1}
\liminf_{r\to R_D}\left(1- \frac{r^2\Lambda}{3}-\frac{\kappa}{r}\int_0^{r} s^2 G_\phi(s,y_\Lambda(s))\mathrm ds\right) = 0.
\end{equation}
\end{lemma}

\begin{rem}
Lemma \ref{lemden} implies that the denominator on the right hand side of equation (\ref{eqsimy}) becomes arbitrarily small on $[0,R_c)$, i.e.~the numerator has no singular behavior that would make the solution collapse as long as the denominator is larger than zero.
\end{rem}

\begin{rem}
We can a priori not exclude the case $R_c=\infty$ which would however not occur due to the $\La$ term.
\end{rem}

\begin{proof}
Assume
\begin{equation}
1-\frac{r^2\Lambda}{3}-\frac{\kappa}{r}\int_0^r s^2 G_\phi(s,y_\Lambda(s))\mathrm ds > 0
\end{equation}
for all $r\in[0,R_c)$. Otherwise $R_D < R_c$ (with $R_D$ characterized as above) occurs due to the continuity of $y_\Lambda$ and $G_\phi$ and the lemma follows. Assume now that the assertion of the lemma does not hold, i.e. there is a constant $a>0$ such that
\begin{equation} \label{defra}
1 - \frac{r^2\Lambda}{3} - \frac{\kappa}{r} \int_0^r s^2 G_\phi(s,y_\Lambda(s)) \mathrm ds \geq a
\end{equation}
for all $r\in[0,R_c)$. First we show that this implies the existence of a $C>0$ such that for all $r\in[0,R_c)$ we have $|y_\La'(r)| \leq C$. Therefore we consider
\begin{equation}\label{mod mu pr}
|y_\La'(r)| \leq \frac{4\pi}{a} \left(r H_\phi(r,y_\Lambda(r)) + \frac{2r\Lambda}{3\kappa} +\frac{1}{r^2} \int_0^r s^2 G_\phi(s,y_\Lambda(s))\mathrm ds\right).
\end{equation}
Here it is used that $H_\phi$ and $G_\phi$ are positive. It is obvious that the second term, $\frac{2r\Lambda}{3\kappa}$, is bounded on the interval $[0,R_c)$. We show that the right hand side of (\ref{mod mu pr}) is bounded on this interval. Assume the opposite,
\begin{equation}
\limsup_{r\to R_c} H_\phi(r,y_\Lambda(r))=\infty \quad\mathrm{or}\quad\limsup_{r\to R_c} \int_0^r s^2 G_\phi(s,y_\Lambda(s))\mathrm ds = \infty.
\end{equation}
 The second possibility implies $\limsup_{r\to R_c} G_\phi(r,y_\Lambda(r)) =\infty$. On the interval $[0,R_c)$ we have the upper bounds $H_\phi(r,y_\Lambda(r)) \leq H_\phi(R_c,y_\Lambda(r))$ and $G_\phi(r,y_\Lambda(r)) \leq G_\phi(R_c,y_\Lambda(r))$, cf.~Lemma \ref{lemmagh}, (\ref{ghinc}). And since $H_\phi(r,y)$ and $G_\phi(r,y)$ are increasing functions in $y$ (cf.~Lemma \ref{lemmagh}) this in turn implies
\begin{equation}
\limsup_{r\to R_c} y_\Lambda(r) = \infty.
\end{equation}
It follows that for all $\varepsilon > 0$ sufficiently small there exists $r\in (R_c-\varepsilon,R_c)$ such that $y_\Lambda'(r) > 0$ which on the other hand implies
\begin{equation}\label{eq-310}
r H_\phi(r,y_\Lambda(r))+\frac{1}{r^2}\int_0^r s^2 G_\phi(s,y_\Lambda(s))\mathrm ds < \frac{2r\Lambda}{3\kappa},
\end{equation}
by equation \eqref{eqsimy} for $y_\Lambda'$. This contradicts the assumption that either $H_\phi(r,y_\Lambda(r))$ or the integral $\int_0^r s^2 G_\phi(s,y_\Lambda(s))\mathrm ds$ diverge as the right hand side of (\ref{eq-310}) is bounded. Thus $|y_\Lambda'(r)|$ is bounded on $[0,R_c)$.\par
In the remainder of this proof it is shown that the solution can be continued beyond $R_c$ which yields the desired contradiction. To achieve this, similar methods as in the proof of Lemma \ref{lemloc} will be used. For $\delta,\varepsilon>0$, $\delta > \varepsilon$ define $y_{\varepsilon}=y_\Lambda(R_c-\varepsilon)$, the interval $I_\varepsilon$ containing $R_c$ by $I_\varepsilon=[R_c-\varepsilon,R_c-\varepsilon+\delta]$, and
\begin{equation} \label{defumu}
u_y(r) := \Big\{\begin{array}{ll} y_\Lambda(r);&\quad r \in [0,R_c-\varepsilon] \\ u(r);&\quad r > R_c-\varepsilon\end{array}.
\end{equation}
Consider the operator
\begin{equation}
\begin{aligned}
(T_\varepsilon u)(r)& = y_\varepsilon + \int_{R_c-\varepsilon}^{r} \frac{\kappa/2}{1-\frac{s^2\Lambda}{3}-\frac{\kappa}{s} \int_0^s\sigma^2 G_\phi(\sigma,u_y(\sigma))\mathrm d\sigma} \\
&\qquad\qquad\times \left(s H_\phi(s,u(s)) - \frac{2s\Lambda}{3\kappa} + \frac{1}{s^2}\int_0^s \sigma^2 G_\phi(\sigma,u_y(\sigma))\mathrm d\sigma\right) \mathrm ds
\end{aligned}
\end{equation}
acting on the set
\begin{equation}
\begin{aligned}
M_\varepsilon&=\Big\{u:I_\varepsilon\to\mathbb R\;|\; u(R_c-\varepsilon)=y_\varepsilon,y_\varepsilon-1\leq u(r)\leq y_\varepsilon+1, \\
&\qquad\qquad\qquad \frac{r^2\Lambda}{3}+\frac{\kappa}{r}\int_0^r s^2 G_\phi(s,u_y(s))\mathrm ds \leq c < 1, r\in I_\varepsilon\Big\}.
\end{aligned}
\end{equation}
Using (\ref{defra}) and $|y_\Lambda'(r)|<C$ on $[0,R_c)$ for a $C>0$ one can prove
%\begin{enumerate}
%\item $u(r) \equiv \mu_\varepsilon \in M_\varepsilon$,
%\item $\forall u \in M_\varepsilon: Tu\in M_\varepsilon$,
%\item $T_\varepsilon$ acts as a contraction on $M_\varepsilon$, i.e. $\forall u,v\in M_\varepsilon: \Ab{Tu-Tv}_{\infty,I_\varepsilon} \leq c\Ab{u-v}_{\infty,I_\varepsilon}, c\in(0,1)$, where $\Ab{.}_{\infty,I_\varepsilon}=\sup_{r\in I_\varepsilon}(.)$.
%\end{enumerate}
that $T_\varepsilon$ acts as a contraction on $M_\varepsilon$. In virtue of Banach's fixed point theorem the operator $T_\varepsilon$ has a fixed point $y_\varepsilon \in M_\varepsilon$ such that $\left(y_\varepsilon\right)_y$ defined by (\ref{defumu}) solves equation (\ref{eqsimy}) on the interval $(0,R_c-\varepsilon+\delta)$. But this contradicts the definition of $R_c$ and the lemma follows.
\end{proof}

\subsection{Existence beyond the non-vacuum region} \label{secsupport}

\begin{lemma} \label{bounded}
Let $\Phi:\mathbb R^2 \to [0,\infty)$ be of the form (\ref{ouransf}) and let $y$ be the unique global $C^1$-solution of equation (\ref{eqsimy}) in the case $\Lambda = 0$ where $y(0)=y_0>0$, cf.~\cite{rein94}. As proved in \cite{rein94}, $f$ has bounded spatial support $[0,R_0)$ where $y(R_0)=0$ defines $R_0$ uniquely. Let $y_\Lambda$ be the unique $C^2$-solution of equation (\ref{eqsimy}) with $\Lambda > 0$ and $y_\Lambda(0) = y(0)$ that according to Lemma \ref{lemloc} exists at least on an interval $[0,\delta]$ for a certain $\delta > 0$ and let $f_\La$ be the distribution function corresponding to $y_\La$. \par
Then $y_\Lambda$ exists at least on $[0,R_0+\Delta R]$ and the spatial support of $f_\Lambda$ is bounded by some $R_{0,\La}<R_0+\Delta R$ if $\La$ and $\Delta R>0$ are chosen such that
\begin{equation} \label{condlambda}
0<\Lambda < \min\left\{\frac{|y(R_0+\Delta R)|}{C_y(R_0+\Delta R)},\frac{\frac{1}{18}}{C_v(R_0+\Delta R)}\right\}
\end{equation}
holds. The constants $C_y(r)$ defined in equation (\ref{defcmu}) and $C_v(r)$ defined in equation (\ref{defcf}) are determined by the background solution $y$.
\end{lemma}

\begin{rem}
Note that the upper bound for $\Lambda$ in (\ref{condlambda}) is strictly larger than zero since $|y(R_0+\Delta R)|>0$. This holds because the globally existing background solution $y$ is strictly monotone and we have $y(R_0)=0$ by definition of $R_0$.
\end{rem}

\begin{proof}
We define
\begin{eqnarray}
m(r) &=& 4\pi\int_0^r s^2 \varrho(s)\mathrm ds, \qquad m_\Lambda(r) = 4\pi\int_0^r s^2 \varrho_\Lambda(s)\mathrm ds, \\
v(r) &=& 1 - \frac{2m(r)}{r}, \qquad \qquad \; v_\Lambda(r) = 1 - \frac{r^2\Lambda}{3} - \frac{2m_\Lambda(r)}{r}.
\end{eqnarray}
Consider the continuous function $v_\Lambda$. Note that $v_\Lambda(0)=1$. We define
%$r^* = r^*(\Lambda)$ as the smallest radius where $v_\Lambda(r) = \frac{1}{18}$. Lemma \ref{lemden} assures that $r^*<R_c$, i.e.~$r^*$ is well defined.
\begin{equation}
r^*:=\inf \{r\in[0,R_c)\,|\,v_\La(r)=1/18\},
\end{equation}
i.e., $r^* $ is the smallest radius where $v_\Lambda(r) = \frac{1}{18}$. Lemma \ref{lemden} assures that $r^*<R_c$, i.e., ~$r^*$ is well defined. Note that $v_\Lambda(r)$ is the quantity in Lemma \ref{lemden}.
In addition, we define
%$\tilde r$ as the radius until that $|y_\Lambda(r)-y(r)|\leq |y(R_0+\Delta R)|$. The right hand side of this inequality is given by the background solution $y$, which exists globally, i.e.
\begin{equation}
\tilde r:=\sup \{r\in [0,R_c]\,|\, |y_\Lambda(r)-y(r)|\leq |y(R_0+\Delta R)|\}.
\end{equation}
The right hand side of this inequality is given by the background solution $y$, which exists globally.
Note that $|y(R_0+\Delta R)|>0$ since $y$ is strictly monotone, and $y(0)=y_\Lambda(0)=y_0$, so $0 < \tilde r$ by continuity of $y$ and $y_\Lambda$. Let
\begin{equation}
\tilde r^* := \min\{r^*,\tilde r\}.
\end{equation}
Choosing $\Lambda$ s.t.~(\ref{condlambda}) holds, we will show that $\tilde r^* > R_0+\Delta R$. We assume the opposite, $\tilde r^* \leq R_0 + \Delta R$, and consider the sum $|\varrho_\Lambda(r) - \varrho(r)| + |p_\Lambda(r) - p(r)|$ on the interval $[0,\tilde r^*]$. By the mean value theorem we have
\begin{equation}
|\varrho_\Lambda(r) - \varrho(r)| + |p_\Lambda(r) - p(r)| = \left(\left|\partial_y G_\phi(r,y)\big|_{u_1}\right| + \left|\partial_y H_\phi(r,y)\big|_{u_2}\right|\right) |y_\Lambda(r)-y(r)|
\end{equation}
where $u_1,u_2\in[y(r),y_\Lambda(r)]$ are chosen appropriately. From the estimate (\ref{defcgh}) in Appendix B we have that for $r\leq \tilde r^*$
\begin{equation}
|\varrho_\Lambda(r) - \varrho(r)| + |p_\Lambda(r) - p(r)| \leq
%\left(|\hat G_\phi'| + |\hat H_\phi'|\right) \left| \int_0^r \mathrm (y_\Lambda'(s)-y'(s)) ds\right|
\Lambda C_{gh}(\tilde r^*),
\end{equation}
where $C_{gh}$ is defined in \eqref{defcgh}.
%and we used the abbreviations
%\begin{equation}
%\begin{aligned}
%|\hat G_\phi'| &= \sup\{|\partial_u G_\phi(\tilde r^*,u)|\;|\;u\in\{y(r):r\in[0,\tilde %r^*]\}\cup\{y_\Lambda(r):r\in[0,\tilde r^*]\}\}, \\
%|\hat H_\phi'| &= \sup\{|\partial_u H_\phi(\tilde r^*,u)|\;|\;u\in\{y(r):r\in[0,\tilde %r^*]\}\cup\{y_\Lambda(r):r\in[0,\tilde r^*]\}\}.
%\end{aligned}
%\end{equation}
%This implies in particular that $|\varrho_\Lambda(r) - \varrho(r)| \leq \Lambda C_{gh}(\tilde r^*)$.
Note that $C_{gh}(r)$ is increasing in $r$. Still on $[0,\tilde r^*]$ we compute
\begin{equation}\label{defcf}
\begin{aligned}
|v(r) - v_\Lambda(r)| &\leq \frac{r^2\Lambda}{3} + \frac{2}{r}|m_\Lambda(r)-m(r)| = \frac{r^2\Lambda}{3} + \frac{8\pi}{r}\int_0^r s^2|\varrho_\Lambda(s)-\varrho(s)|\mathrm ds  \\
&\leq \left(\frac{(\tilde r^*)^2}{3} + \frac{8\pi}{3} (\tilde r^*)^2C_{gh}(\tilde r^*)\right)\Lambda =: C_v(\tilde r^*) \Lambda
\end{aligned}
\end{equation}
Since we have $v(r) \geq \frac{1}{9}$ (Buchdahl inequality, cf.~Lemma \ref{theobuchdahl}) and $\Lambda < \frac{1/18}{C_v(R_0+\Delta R)}$ by choice of $\Lambda$ we can conclude
\begin{equation} \label{cont1}
v_\Lambda(r) \geq v(r) - \Lambda C_v(\tilde r^*) > \frac 1 9 - \frac{1/18}{C_v(R_0+\Delta R)} C_v(\tilde r^*) \geq \frac{1}{18}
\end{equation}
on $[0,\tilde r^*]$ since $C_v(\tilde r^*) < C_v(R_0+\Delta R)$ because $C_v(r)$ is increasing and $\tilde r^* \leq R_0+\Delta R$ by assumption.\par
We also consider the distance between $y$ and $y_\Lambda$ on $[0,\tilde r^*]$. Following the procedure depicted in Section \ref{apdetest} of the appendix one obtains
\begin{equation} \label{defcmu}
\begin{aligned}
|y_\Lambda(r) - y(r)| &\leq \Lambda \left(3r^2+29\pi r^4\left(H_\phi(r,y_0)+\frac 1 3 G_\phi(r,y_0)\right)\right) \\
&\qquad + 72\pi\left(r+24\pi r^2\left(H_\phi(r,y_0) + \frac 1 3 G_\phi(r,y_0)\right)\right) \int_0^r C_{gh}(s) \Lambda \mathrm ds  \\
& =: C_y(r) \Lambda \leq C_y(\tilde r^*) \Lambda.
\end{aligned}
\end{equation}
Since $C_y(\tilde r^*) \leq C_y(R_0+\Delta R)$ and  $\Lambda < \frac{|y(R_0+\Delta R)|}{C_y(R_0+\Delta R)}$ on $[0,\tilde r^*]$ by assumption, the relation
\begin{equation} \label{cont2}
|y_\Lambda(r) - y(r)| < |y(R_0+\Delta R)|
\end{equation}
 already holds. Equations (\ref{cont1}) and (\ref{cont2}) state that $v_\Lambda(\tilde r^*) > \frac{1}{18}$ and $|y_\Lambda(\tilde r^*)-y(\tilde r^*)| < |y(R_0+\Delta R)|$, respectively on the interval $[0,\tilde r^*]$, which is a contradiction to the definition of $\tilde r^*$. Thus we have $\tilde r^* > R_0+\Delta R$ as desired.\par
We have shown that $y_\Lambda$ exists at least on $[0,R_0+\Delta R]$ as the continuation criterion applies and from equation (\ref{cont2}) we already know that $y_\Lambda(R_0+\Delta R) < 0$. Since $y_\Lambda$ is continuous it has at least one zero at in the interval $(R_0,R_0+\Delta R)$. In particular there exists an interval $(R_{0\Lambda},R_0+\Delta R)$ where $y_\Lambda$ is strictly smaller than zero. $R_{0\Lambda}$ is the largest zero of $y_\Lambda$ in $(R_0,R_0+\Delta R)$. So the spatial support of $f_\Lambda$ is contained in the interval $[0,R_{0,\Lambda})$ and this implies the assertion.
\end{proof}

\subsection{Global regular solutions for $\La>0$}
\label{ssec : sagr-4}
In the last two sections we have seen that for suitably chosen $\Lambda$ there exists a unique solution $y_\Lambda$ to equation (\ref{eqsimy}) on the interval $[0,R_0+\Delta R]$ for some $\Delta R>0$. This solution uniquely induces a solution $\mu_\Lambda$, $\lambda_\Lambda$ of the equations \eqref{eeq1}, \eqref{eeq2} on $[0,R_0+\Delta R]$ whose distribution function $f_\Lambda$ is of bounded support in space. By gluing a  Schwarzschild-deSitter metric to this solution one can construct a global static solution to the Einstein-Vlasov system.

\begin{theorem}\label{main-thm}
Let $\Phi:\mathbb R^2 \to [0,\infty)$ be of the form (\ref{ouransf}). For every initial value $\mu_0 <0$ there exists a constant $C = C(\mu_0,\phi) > 0$ such that for every $0<\Lambda < C$ there exists a unique global solution $\mu_\Lambda, \lambda_\Lambda \in C^2([0,\infty)), f_\Lambda \in C^0([0,\infty)\times \mathbb R^3)$ of the static, spherically symmetric Einstein-Vlasov system \eqref{eqvlasov}-\eqref{p}
%\begin{eqnarray*}
%\frac{v^a}{\sqrt{1+|v|^2}}\frac{\partial f}{\partial x^a}-\sqrt{1+|v|^2}\mu'\frac{x^a}{r}\frac{\partial f}{\partial v^a} &=& 0, \\
%e^{-2\lambda}(2r\lambda'-1)+1-r^2\Lambda &=& \kappa r^2 \varrho, \\
%e^{-2\lambda}(2r\mu'+1)-1+r^2\Lambda &=& \kappa r^2p,
%\end{eqnarray*}
with $\mu_\Lambda(0)=\mu_0$, and $\lambda_\Lambda(0)=0$ such that the support of the distribution function is bounded. This solution coincides with the Schwarzschild-deSitter metric in the vacuum region.
\end{theorem}

\begin{proof}
According to Lemma \ref{lemloc} there exists a $C^2$-solution $y_\Lambda$ of equation (\ref{eqsimy}) on a small interval $[0,\delta]$. In the proof of Lemma \ref{bounded} we saw that this solution can be extended at least until $r = R_0+\Delta R$ for any $\Delta R$ if one chooses $\Lambda$ small enough. Beyond the support of $\varrho_\Lambda$ and $p_\Lambda$, thus for $r\in[R_{0,\Lambda},R_0+\Delta R]$, equation (\ref{eqsimy}) takes the form
\begin{equation} \label{eqmussds}
y_\Lambda'(r) = -\frac 1 2 \frac{\mathrm d}{\mathrm dr}\ln\left(1-\frac{r^2\Lambda}{3}-\frac{2M}{r}\right)
\end{equation}
where $M=m_\Lambda(R_{0,\Lambda})$. This equation is solved by the (shifted) Schwarzschild-deSitter metric, whose corresponding $y$-coefficient $y_S$ is given by
\begin{equation}
y_S(r) = -\frac 1 2\ln\left( 1-\frac{r^2\Lambda}{3}-\frac{2M}{r} \right) - \ln\left(e^{-\lambda(R_{0,\Lambda})}\right).
\end{equation}
The shift has been chosen such that $y_\Lambda$ can be extended by $y_S$ as a $C^2$-solution of equation \eqref{eqsimy} on $[0,\infty)$ using a modified ansatz for the matter distribution $f_\Lambda$. Namely, for $r > R_0+\Delta R$ we drop the original ansatz $\Phi$ for $f_\Lambda$ and continue $f_\Lambda$ by the constant zero function, i.e.
\begin{equation}
f_\Lambda(x,v)=\Big\{\begin{array}{ll}\left[1-\varepsilon e^{-y}\right]_+^k[L-L_0]_+^\ell,& r \in [0,R_0+\Delta R]\\0, & r > R_0+\Delta R\end{array}.
\end{equation}
Obviously $f_\Lambda$ is continuous since $f_\Lambda(r)=0$ already on $(R_{0,\Lambda},R_0+\Delta R)$ but $\frac{\mathrm d}{\mathrm dr}f_\Lambda(r,v)$ is not continuous in general. \par
Via $\mu_\Lambda=\ln(E_0)-y_\Lambda$ and
\begin{equation}
e^{-2\lambda_\Lambda}=1-\frac{r^2\Lambda}{3}-\frac{\kappa}{r}\int_0^r s^2 G_\phi(s,y_\Lambda(s))\mathrm ds
\end{equation}
 one can construct a local solution $\mu_\Lambda,\lambda_\Lambda\in C^2([0,R_c))$ of \eqref{eeq1}, \eqref{eeq2}, where $R_c>R_0+\Delta R$. This solution fulfills the boundary conditions $\lambda_\Lambda(0)=0$, $\mu_\Lambda(0)=\ln(E_0)-y_0$, $\lambda_\Lambda'(0)=\mu_\Lambda'(0)=0$.
We now see that $E_0 = e^{\mu(R_{0,\Lambda})}$ and continue $\mu_\Lambda$ and $\lambda_\Lambda$ with the Schwarzschild-deSitter coefficients $\mu_S$, $\lambda_S$ given by
\begin{equation}
e^{2\mu_S}=e^{-2\lambda_S}=1-\frac{r^2\Lambda}{3}-\frac{2M}{r}
\end{equation}
in a continuous way beyond $R_0+\Delta R$. From equation (\ref{eqmussds}) we deduce that also the derivatives of $\mu_\Lambda$ and $\mu_S$ can be glued together in a continuous way.  The functions $\mu_\Lambda$, $\lambda_\Lambda$, and $f_\Lambda$ solve the Einstein-Vlasov system \eqref{eqvlasov}, \eqref{eeq1}, \eqref{eeq2} globally.
\end{proof}

\begin{rem}
In the isotropic case, i.e.~$L_0=\ell=0$ in the ansatz \eqref{ouransf} for the distribution function $f$, the matter quantities $\varrho$ and $p$ are monotonically decreasing. This implies that their support in space is a ball. In the anisotropic case however, so called shell solutions occur, cf.~\cite{ar13}. The support of such matter shells is in general not connected.
\end{rem}

%%%%%%%%%%%%%%%%%%%%%%%%%%%%%%
%%%%%%%%%%%%%%%%%%%%%%%%%%%%%%
%%%%%%%%%%%%%%%%%%%%%%%%%%%%%%

\section{Static, anisotropic, globally regular solutions for $\La<0$} \label{sec : sage neg}
\subsection{Local existence} \label{seclocal}
In this section an existence lemma for $\Lambda < 0$ is stated for small radii. This lemma corresponds to the first part of the proof of Theorem 2.2 in \cite{static} for the case $\La=0$.
\begin{lemma} \label{lemloc_neg}
Let $\Phi:\mathbb R^2\to[0,\infty)$ be of the form \eqref{ouransf} and let $G_\phi$, $H_\phi$ be defined by equations (\ref{defgphi}) and (\ref{defhphi}), respectively. Then for every $y_0 \in \mathbb R$ and every $\Lambda < 0$ there exists a $\delta>0$ such that there exists a unique solution $y_\Lambda \in C^2([0,\delta])$ of equation \eqref{eqsimy} with initial value $y_\Lambda(0)=y_0$.
\end{lemma}

\begin{proof}
The proof works in an exact analogue way as in the case $\Lambda > 0$.
\end{proof}

\subsection{Globally regular solutions for $\La<0$}

For negative cosmological constants the global existence of solutions can be proved in an analogue way as done in \cite{rein94} for the case $\Lambda = 0$. After establishing the local existence of solutions analog to the $\La>0$ case, we show that the metric components stay bounded for all $r\in\mathbb R_+$ with an energy estimate. This will yield the global existence of solutions of the Einstein-Vlasov system with negative cosmological constant. In the next step we show by virtue of a suitable choice of an ansatz for the matter distribution $f$, that the matter quantities $\varrho$ and $p$ are of bounded support.\\
In the following theorem the existence on spatial intervals of the form $\mathbb R\setminus [0,r_0)$, for $r_0>0$ is included for the purpose of applying the same theorem to the construction of static spacetimes with Schwarzschild singularities in the center (cf.~Section \ref{bhAdS}). The solutions of interest here are those where the radius variable takes values in all of $\mathbb R$.
\begin{theorem} \label{theo_glo_adS}
Let $\La<0$ and let $\Phi:\mathbb R^2 \to [0,\infty)$ be of the form (\ref{ouransf}) and let $G_\phi$ and $H_\phi$ be defined by equations \eqref{defgphi} and \eqref{defhphi}. Then for every $r_0\geq 0$ and $\mu_0,\lambda_0\in \mathbb R$ there exists a unique solution $\lambda_\Lambda,\mu_\Lambda\in C^1([r_0,\infty))$ of the Einstein-Vlasov system  \eqref{eqvlasov}-\eqref{p} with $\mu_\Lambda(r_0)=\mu_0$ and $\lambda_\Lambda(r_0)=\lambda_0$. One has $\lambda_0=0$ if $r_0=0$.
\end{theorem}

\begin{proof} We use an energy argument similar to \cite{rein94}.
Let $y_\Lambda \in C^2([r_0,r_0+\delta])$ be the local solution of equation \eqref{eqsimy} with $y_\Lambda(r_0)=\ln(E_0)e^{-\mu_0}$. If $r_0=0$ the existence of this local solution is established by Lemma \ref{lemloc_neg} and in the case $r_0 > 0$ the existence of a local solution follows directly from the regularity of the right hand sides of \eqref{eeq1} and \eqref{eeq2}. Let $[r_0,R_c)$ be the maximal interval of existence of this solution. By $\mu_\Lambda = \ln(E_0)-y_\Lambda$ and
\begin{equation}
e^{-2\lambda_\Lambda}=1-\frac{\Lambda}{3}\left(r^2-\frac{r_0^3}{r}\right)-\frac{2}{r}\left(\frac{r_0}{2}\left(1-e^{-2\lambda_0}\right)+4\pi \int_{r_0}^r s^2 G_\phi(s,y_\Lambda(s))\mathrm ds\right)
\end{equation}
one constructs a local solution $\mu_\Lambda,\lambda_\Lambda \in C^2([r_0,R_c])$ of equations \eqref{eeq1} and \eqref{eeq2}. We define
\begin{equation}
w_\Lambda(r) = -\frac{\Lambda}{12\pi} + \frac{1}{r^3} \left(-\frac{r_0^3\Lambda}{24\pi} + \frac{r_0}{8\pi}\left(1-e^{-2\lambda_0}\right) + \int_{r_0}^r s^2\varrho_\Lambda(s)\mathrm ds\right).
\end{equation}
The Einstein equation \eqref{eeq1} implies
\begin{equation} \label{eqmu}
\mu_\Lambda'(r) = 4\pi r e^{2\lambda_\Lambda(r)}\left(p_\Lambda(r)+w_\Lambda(r)\right).
\end{equation}
By adding equations \eqref{eeq1} and \eqref{eeq2} we have
\begin{equation}
\left(\mu_\Lambda'(r) + \lambda_\Lambda'(r)\right) = 4\pi r e^{2\lambda_\Lambda(r)}(p_\Lambda(r) + \varrho_\Lambda(r)).
\end{equation}
We assume $R_c < \infty$ and consider the quantity $e^{\mu_\Lambda+\lambda_\Lambda}\left(p_\Lambda + w_\Lambda\right)$ on the interval $\left[\frac{R_c}{2},R_c\right)$. On this interval, in particular away from the origin, a differential inequality will be established that will allow us to deduce that both $\mu_\Lambda$ and $\lambda_\Lambda$ are bounded on $\left[\frac{R_c}{2},R_c\right)$. Using the TOV equation \eqref{toveq} we obtain for $r\in \left[\frac{R_c}{2},R_c\right)$
\begin{equation}
\begin{aligned}
\frac{\mathrm d}{\mathrm dr}\left(e^{\mu_\Lambda+\lambda_\Lambda}\left(p_\Lambda + w_\Lambda\right)\right) &= e^{\mu_\Lambda+\lambda_\Lambda}\left(-\frac{2p_\Lambda}{r} - \frac{3w_\Lambda}{r} - \frac{\Lambda}{4\pi r} + \frac{2p_{T\Lambda}}{r} + \frac{\varrho_\Lambda}{r}\right) \\
&\leq C_1 e^{\mu_\Lambda+\lambda_\Lambda} = \underbrace{\frac{C_1}{p_\Lambda + w_\Lambda}}_{=:C_2}\left(p_\Lambda+w_\Lambda \right) e^{\mu_\Lambda+\lambda_\Lambda}.
\end{aligned}
\end{equation}
In the course of this estimate we have used that $\frac{\Lambda}{4\pi r}$, $p_{T\Lambda}(r)/r$ and $\varrho_\Lambda(r)/r$ stay bounded for $r\in\left[\frac{R_c}{2},R_c\right)$. The constant $C_2$ is bounded since $w_\Lambda(r) > 0$ for negative $\Lambda$. It follows
\begin{equation} \label{lpmu-b}
\frac{\mathrm d}{\mathrm dr}\ln\left(e^{\mu_\Lambda+\lambda_\Lambda}\left(p_\Lambda + w_\Lambda\right)\right) \leq C_2 \quad \Rightarrow \quad \lambda_\Lambda+\mu_\Lambda < \infty.
\end{equation}
Equation \eqref{eqmu} implies that $\mu_\Lambda'(r)\geq 0$ and therefore $\mu_\Lambda(r) \geq \mu_0$. We also have
\begin{equation}
e^{-2\lambda_\Lambda} \leq 1 + \frac{r^2|\Lambda|}{3} \leq \frac{3+R_c^2|\Lambda|}{3} < \infty.
\end{equation}
This in turn implies $\lambda_\Lambda > -\infty$ and we deduce from equation \eqref{lpmu-b} that both $\mu_\Lambda$ and $\lambda_\Lambda$ are bounded on $\left[\frac{R_c}{2},R_c\right)$. This allows to continue $\mu_\Lambda$ and $\lambda_\Lambda$ as $C^2$-solutions of the Einstein equations beyond $R_c$ which contradicts its definition. So $R_c = \infty$.
\end{proof}
We prove in the following theorem that the distribution function in the previous theorem is compactly supported which yield physically reasonable solutions.
\begin{theorem} \label{theo_bou_adS}
Let $\Phi:\mathbb R^2 \to [0,\infty)$ be of the form (\ref{ouransf}) and let $\mu_0 \in \mathbb R$ , $r_0\geq 0$ and let $\lambda_\Lambda,\mu_\Lambda\in C^1([r_0,\infty))$, $f(x,v)=\Phi(E,L)$ be the unique global-in-$r$ solution of the Einstein-Vlasov system \eqref{eqvlasov} -- \eqref{p} with negative cosmological constant where $\mu_\Lambda(0)=\mu_0$ such that $y_0 = \ln(E_0)e^{-\mu_0} > 0$. Then there exists $R_0\in(r_0,\infty)$ such that the spatial support of $f_\Lambda$ is contained in the interval $[r_0,R_0)$.
\end{theorem}

\begin{proof}
Due to Lemma \ref{lemmagh}, (\ref{ghvac}) we have vacuum if $y_\Lambda(r) \leq 0$. By assumption we have $y_\Lambda(0) > 0$. In the following we show that $\lim_{r\to \infty} y_\Lambda(r) < 0$. Since $y_\Lambda$ is continuous and monotonically decreasing, this implies that $y_\Lambda$ possesses a single zero $R_0$ and the support of the matter quantities $\varrho_\Lambda$ and $p_\Lambda$ is contained in $[0,R_0)$. \par
We define $y_{\mathrm{vac},\Lambda}$ by
\begin{equation}
y_{\mathrm{vac},\Lambda}=y_0-\frac 1 2 \ln\left(1-\frac{r^2\Lambda}{3}\right).
\end{equation}
So we have
\begin{equation}
y_{\mathrm{vac},\Lambda}'(r)=-\frac{\kappa/2}{1-\frac{\Lambda r^2}{3}}\left(-\frac{2r\Lambda}{3\kappa}\right)
\end{equation}
and $y_{\mathrm{vac},\Lambda}(0)=y_\Lambda(0)=y_0$. Furthermore, since $y_{\Lambda}'(r)<y_{\mathrm{vac},\Lambda}'(r)$ which can be seen immediately by means of equation \eqref{eqsimy} we have
\begin{equation}
y_\Lambda(r) \leq y_{\mathrm{vac},\Lambda}(r) = y_0-\frac 1 2 \ln\left(1+\frac{r^2|\Lambda|}{3}\right)\stackrel{r\to\infty}{\longrightarrow}-\infty<0
\end{equation}
and the theorem follows.
\end{proof}

\begin{rem}
The solution coincides with Schwarzschild-AdS for $r \geq R_0$ if the continuity condition
\begin{equation}
\mu_\Lambda(R_0) = \ln(E_0) - y_\Lambda(R_0) = \frac 1 2 \ln\left(1-\frac{R_0^2\Lambda}{3}-\frac{2M}{R_0}\right)
\end{equation}
is fulfilled, where $M=4\pi\int_0^{R_0} s^2\varrho_\Lambda(s)\mathrm ds$. So if $y_0$ is given, the corresponding value of $E_0$ in the ansatz $\Phi$ for the matter distribution $f$ can be read off.
\end{rem}

% ---------------------------------- Schwarzschild singularity ----------------------------------

\section{Solutions with a Schwarzschild singularity at the center} \label{sec : swss}

In this section we construct spherically symmetric, static solutions of the Einstein-Vlasov system with non-vanishing cosmological constant that contain a Schwarzschild singularity at the center. We consider both the case with a positive and a negative cosmological constant. The construction for the case $\Lambda > 0$ makes use of the corresponding solutions with vanishing $\La$. In the following we will call this solution, where $\Lambda = 0$, a \emph{background solution}. The global existence of the background solution is proved in \cite{rein94}. The matter quantities belonging to this  background solution are of finite support.\par
%We distinguish between the notation for the Einstein-Vlasov system with $\Lambda = 0$ and the corresponding system with $\Lambda \neq 0$
%Every quantity of the system with $\Lambda \neq 0$ is denoted by an additional index $\Lambda$, e.g.~we write %$y_\Lambda(r)$, $\varrho_\Lambda(r)$ and correspondingly
%as done in the previous sections by using an index $\Lambda$ on the quantities in the latter case.

\subsection{Matter shells immersed in Schwarzschild-deSitter spacetime}\label{sec dS bh}
The construction of the solution with $\Lambda > 0$ can be outlined as follows. In the vacuum case, i.e.~when the right hand sides of the Einstein equations \eqref{eeq1} and \eqref{eeq2} are zero, the solutions are given by
\begin{equation} \label{ssdssol}
e^{2\mu(r)}=1-\frac{r^2\Lambda}{3}-\frac{2M_0}{r},\quad e^{2\lambda(r)}=\left(1-\frac{r^2\Lambda}{3}-\frac{2M_0}{r}\right)^{-1},\quad r> r_{B\Lambda}
\end{equation}
where $r_{B\Lambda}$ is defined to be the black hole event horizon, i.e.~the smallest positive zero of $1-r^2\Lambda/3-2M_0/r$. %Taking $E_0=1$ we have in virtue of lemma \ref{lemmagh}, (\ref{ghvac}) $G_\phi(r,u) = H_\phi(r,u) = 0$ if $e^u\sqrt{1+\frac{L_0}{r^2}} \geq 1$.
If one chooses $L_0$ and $M_0$ appropriately and $\Lambda$ sufficiently small the following configuration is on hand. For small $r>r_{B\Lambda}$ one sets $f(x,v)\equiv0$ and the metric is given by Schwarzschild-deSitter. Thus one has the coefficients \eqref{ssdssol}. Increasing the radius $r$ one reaches an interval $[r_{-\Lambda},r_{+\Lambda}]$ where also an ansatz $f(x,v)=\Phi(E,L)$ of the form \eqref{ouransf} yields vacuum, i.e.~$G_\phi(r,y(r))=H_\phi(r,y(r))=0$. In this interval it is possible to glue to the Schwarzschild-deSitter solution \eqref{ssdssol} a non vacuum solution solving the Einstein-Vlasov system. It will be shown that the matter quantities $\varrho_\Lambda$ and $p_\Lambda$ of this solution have finite support. Beyond the support of the matter quantities the solution will be continued again by Schwarzschild-deSitter. \par
For negative cosmological constant, globally defined solutions can be constructed as well. Like in the case above, the black hole is surrounded by a vacuum shell which is on its part surrounded by a shell containing matter. In the outer region, we again have vacuum. \par
Before we consider the system with $\Lambda \neq 0$ we establish a Buchdahl type inequality for solutions of the Einstein equations with a Schwarzschild singularity at the center. This inequality is relevant for the proof of existence of solutions of the Einstein-Vlasov system with $ \Lambda > 0$.
\begin{lemma} \label{buchdahlbh}
%Let $M_0>0$ and let $\lambda,\mu\in C^2((2M_0,\infty))$ be a solution of the spherically symmetric, static Einstein %equations
%\begin{eqnarray}
%e^{-2\lambda}(2r\lambda' -1)+1&=&\kappa r^2 \varrho, \label{2eeq1}\\
%e^{-2\lambda}(2r\mu' +1)-1&=&\kappa r^2 p, \label{2eeq2} \\
%e^{-2\lambda}\left(\mu''\left(\mu+\frac 1 r\right)\left(\mu'-\lambda'\right)\right)&=&\kappa p_T \label{2eeq3}
%\end{eqnarray}
Let $\lambda,\mu\in C^1([0,\infty))$ and let $\varrho, p, p_T\in C^0([0,\infty))$ be functions that satisfy the system of equations (\ref{eeq1}-\ref{eeq3}) with a Schwarzschild singularity with mass parameter $M_0>0$ at the center, and such that $p+2p_T \leq \varrho$.
%with a Schwarzschild singularity with mass parameter $M_0$ at the center. Let the matter quantities $\varrho,p,p_T\in %L_{loc}^1((2M_0,\infty))$ be zero on $\left({2M_0,\frac{9}{4}M_0}\right]$, fulfill the energy condition
%\begin{equation} \label{cb1}
%p + 2p_T \leq \varrho,
%\end{equation}
%and the generalized TOV equation
%\begin{equation} \label{cb2}
%p'(r) = -\mu'(\varrho+p)-\frac 2 r (p-p_T)
%\end{equation}
%a.e.~in their domain of definition.
Then the inequality
\begin{equation}\label{mb}
\frac{2(M_0 + m(r))}{r} \leq \frac 8 9
\end{equation}
holds for all $r\in\left[\frac{9}{4}M_0,\infty\right)$ where $m(r)$ is given by
\begin{equation}
m(r) ={4\pi\int_{2M_0}^r} s^2\varrho(s)\mathrm ds.
\end{equation}
\end{lemma}

\begin{proof}
For the prove of the lemma we apply techniques that are already used in \cite{ks08} to prove the Buchdahl inequality for globally regular solutions without Schwarzschild singularity. Only the steps that differ from the proof of Theorem 4.1 in \cite{ks08}, or Theorem 1 in \cite{a09} for the charged case, will be described in detail.\par
%At first we introduce the variables
%\begin{equation}
%x = \frac{2(M_0+m(r))}{r},\quad y = 8\pi r^2p(r).
%\end{equation}
%Note that $x < 1$ and $y \geq 0$. The first inequality must hold true since otherwise the lapse function $e^\mu$ would not stay bounded.
By integrating the Einstein equation \eqref{eeq1} over the interval $\left(\frac{9M_0}{4},r\right)$ we obtain
\begin{equation}
e^{-2\lambda} = 1 - \frac{9M_0}{4r}\left(1-e^{-2\lambda_0}\right) - \frac{8\pi}{r}\int_{\frac{9M_0}{4}}^r s^2\varrho(s) \mathrm ds,
\end{equation}
where $\lambda_0 = \lambda\left(\frac{9M_0}{4}\right)$. Since we have vacuum on $\left(2M_0,\frac{9M_0}{4}\right)$ on this interval the metric is given by the Schwarzschild metric and one can compute $\lambda_0$ explicitly. One finds that
\begin{equation}
e^{-2\lambda} = 1- \frac{2(M_0 + m(r))}{r}.
\end{equation}
We plug this into the other Einstein equation \eqref{eeq2} and obtain the differential equation
\begin{equation}
\mu'(r) = \frac{1}{1-\frac{2(M_0+m(r))}{r}}\left(4\pi r p + \frac{M_0+m(r)}{r^2}\right).
\end{equation}
We now introduce the variables
\begin{equation}
x = \frac{2(M_0+m(r))}{r},\quad y = 8\pi r^2p(r).
\end{equation}
Note that $x < 1$ and $y \geq 0$. The first inequality must hold true since otherwise the metric function $\lambda$ would not stay bounded.
Next we let $\beta = 2\ln(r)$ and consider the curve $\left(x\left(e^{\beta/2}\right),y\left(e^{\beta/2}\right)\right)$ parameterized by $\beta$ in $[0,1) \times [0,\infty)$. In the following a dot denotes the derivative with respect to $\beta$. %We calculate
%\begin{equation}
%\dot x = 4\pi r^2\varrho(r) - \frac{M_0+m(r)}{r}.
%\end{equation}
%Furthermore we have
%\begin{eqnarray}
%\dot y &=& 8\pi r^2\left(\frac{y}{8\pi r^2} - \frac{r}{2}\left(\mu'(p+\varrho)+\frac{2}{r}(p-p_T)\right)\right) \nonumber \\
%&=& y - 8\pi r^2 (p-p_T) - \frac{r}{2}(p+\varrho)8\pi r^2 \frac{4\pi rp + \frac{M_0+m(r)}{r^2}}{1-\frac{2(M_0+m(r))}{r}}.
%\end{eqnarray}
Using the Einstein equations and the generalized TOV equation (\ref{toveq})
one checks that $x$ and $y$ satisfy the equations
\begin{eqnarray}
8\pi r^2\varrho &=& 2\dot x + x, \label{xy_curve_1} \\
8\pi r^2 p &=& y,\\
8\pi r^2p_T &=& \frac{x+y}{2(1-x)}\dot x + \dot y + \frac{(x+y)^2}{4(1-x)}. \label{xy_curve_3}
\end{eqnarray}
By virtue of these equations \eqref{xy_curve_1} -- \eqref{xy_curve_3} the condition $p+2p_T \leq \varrho$ can be written in the form
\begin{equation}
(3x-2+y)\dot x + 2(1-x)\dot y \leq -\frac{\alpha(x,y)}{2},\quad \alpha = 3x^2-2x+y^2+2y.
\end{equation}
From now on the proof is analogue to the proof of Theorem 1 in \cite{a09} for the charged case. One defines the quantity
\begin{equation}
w(x,y) = \frac{(3(1-x)+1+y)^2}{1-x}
\end{equation}
and %calculates the derivative with respect to $\beta$,
%\begin{equation}
%\dot w = \frac{4-3x+y}{(1-x)^2}((3x-2+y)\dot x + 2(1-x)\cdot y) \leq -\frac{4-3x+y}{(1-x)^2} \alpha(x,y)
%\end{equation}
%depending on $x$ and $y$. One
shows that since $0 \leq x < 1$ and $y \leq 0$ this quantity is bounded by $16$ along the curve $(x,y)$ with an optimization procedure. The inequality $w \leq 16$ is already equivalent to
\begin{equation}
\frac{2(M_0+m(r))}{r} \leq \frac{8}{9}
\end{equation}
for all $r\in\left[\frac{9M_0}{4},\infty\right)$ and the proof is complete.
\end{proof}
\begin{rem}
In the case when $M_0=0$ it is known that the inequality is sharp, cf.~\cite{and08} and \cite{ks08}. For the purpose of this work the bound (\ref{mb}) is sufficient and we have not tried to show sharpness.
\end{rem}
In the course of the proof of Theorem \ref{theo_la_pos_bh} we will need a continuation criterion for the solution of the Einstein equations, namely the following statement.

\begin{lemma} \label{genlemden}
Let $\La>0$, $\mu_0\in\mathbb R$ and $M_0,r_0>0$. Let $G_\phi$ and $H_\phi$ defined by equations \eqref{defgphi} and \eqref{defhphi}. Then the equation
\begin{equation} \label{eq_gen_mula}
\begin{aligned}
\mu_\Lambda'&=\frac{1}{1-\frac{\Lambda}{3} \left(r^2-\frac{r_0^3}{r}\right) - \frac{2}{r}\left(M_0 + 4\pi\int_{r_0}^r s^2G_\phi(s,\mu_\Lambda(s))\mathrm ds\right)}\\
&\quad \times \left(4\pi r H_\phi(s,\mu_\Lambda(s))-\Lambda\left(\frac{r}{3} + \frac{r_0^3}{6r^2}\right) \frac{1}{r^2}\left(M_0+4\pi\int_{r_0}^r s^2G_\phi(s,\mu_\Lambda(s))\mathrm ds\right)\right)
\end{aligned}
\end{equation}
has a unique local $C^2$-solution $\mu_\Lambda$ with $\mu(r_0)=\mu_0$ with maximal interval of existence $[r_0,R_c)$, $R_c>0$. Moreover, there exists $R_D \leq R_c$ such that
\begin{equation} \label{genlemden1}
\liminf_{r\to R_D}\left(1-\frac{\Lambda}{3} \left(r^2-\frac{r_0^3}{r}\right) - \frac{2}{r}\left(M_0 + 4\pi\int_{r_0}^r s^2G_\phi(s,\mu_\Lambda(s))\mathrm ds\right)\right) = 0.
\end{equation}
\end{lemma}

\begin{proof}
The local existence of a $C^2$-solution of equation \eqref{eq_gen_mula} follows from the regularity of the right hand side. Basically one is in the situation of Lemma \ref{lemden}, i.e., the case with a regular center and $\Lambda>0$, except for the fact that there are additional terms containing $r_0$ and $M_0$. But on a finite interval $[r_0,R_c)$ these terms are bounded and well behaved, i.e.~the proof can be carried out in an analogue way.
\end{proof}

\begin{rem}
Lemma \ref{genlemden} implies that if
%there exists a solution $\mu_\Lambda$ of equation \eqref{eq_gen_mula} and
the denominator of the right hand side of equation \eqref{eq_gen_mula} is strictly larger than zero on an interval $[r_0,r)$, then $\mu_\Lambda$ can be extended beyond $r$ as a solution of \eqref{eq_gen_mula}.
\end{rem}
The following theorem states the existence of solutions for $\La>0$ with a Schwarzschild singularity at the center.
\begin{theorem} \label{theo_la_pos_bh}
Let $\Phi:\mathbb R^2 \to [0,\infty)$ be of the form (\ref{ouransf}) with $E_0=1$ and let $L_0,M_0 \geq 0$ such that $L_0 > 16M_0^2$. %Let
%\begin{equation}
%\Lambda < \min \left\{ \frac{1}{9M_0^2},\mathrm{for\;}r_{B\Lambda},\frac{3(L_0-2M_0\hat r)}{\hat r^2(\hat r^2 + L_0^2)}, 3\frac{r_{+\Lambda}^2-2M_0}{r_{+\Lambda}^4} \left( 1-e^{-\mu(R_0+\Delta R)} \right) \right\}.
%\end{equation}
Then there exists a unique solution $\mu_\Lambda, \lambda_\Lambda \in C^2((r_{B\Lambda},\infty))$, $f\in C^0((r_{B\Lambda},\infty)\times \mathbb R^3)$ of the Einstein-Vlasov system \eqref{eqvlasov} -- \eqref{p} for $\Lambda > 0$ sufficiently small. The spatial support of the distribution function $f_\Lambda$ is contained in a shell $\{r_{+\Lambda} < r < R_{0\Lambda}\}$. In the complement of this shell the solution of the Einstein equations is given by the Schwarzschild-deSitter metric.
\end{theorem}

\begin{rem}
In the course of the proof one will come across the fact that in one of the vacuum regions, either $r \leq r_{+\Lambda}$ or $r \geq R_{0\Lambda}$, the component $\mu_\mathrm{vac}$ given by $e^{2\mu_\mathrm{vac}}=1-\frac{r^2\Lambda}{3}-\frac{2M}{r}$ of the Schwarzschild-deSitter metric will be shifted by a constant. But this shift is just a reparametrization of the time $t$ \cite{rein94}. Thus the shell of Vlasov matter causes a redshift.
\end{rem}

\begin{proof}
In the first part of the proof we consider the black hole region and show that the chosen parameters lead to
%a convenient configuration with opportune properties,
the configuration depicted in~Figure \ref{fig_radii}. Then we make use of the existence of a background solution and construct the desired solution $\mu_\Lambda$. \par
\begin{figure}[h]
\begin{center}
\setlength{\unitlength}{0.14in}
\begin{picture}(38,18)
\put(0,1){\includegraphics[width=0.9\textwidth]{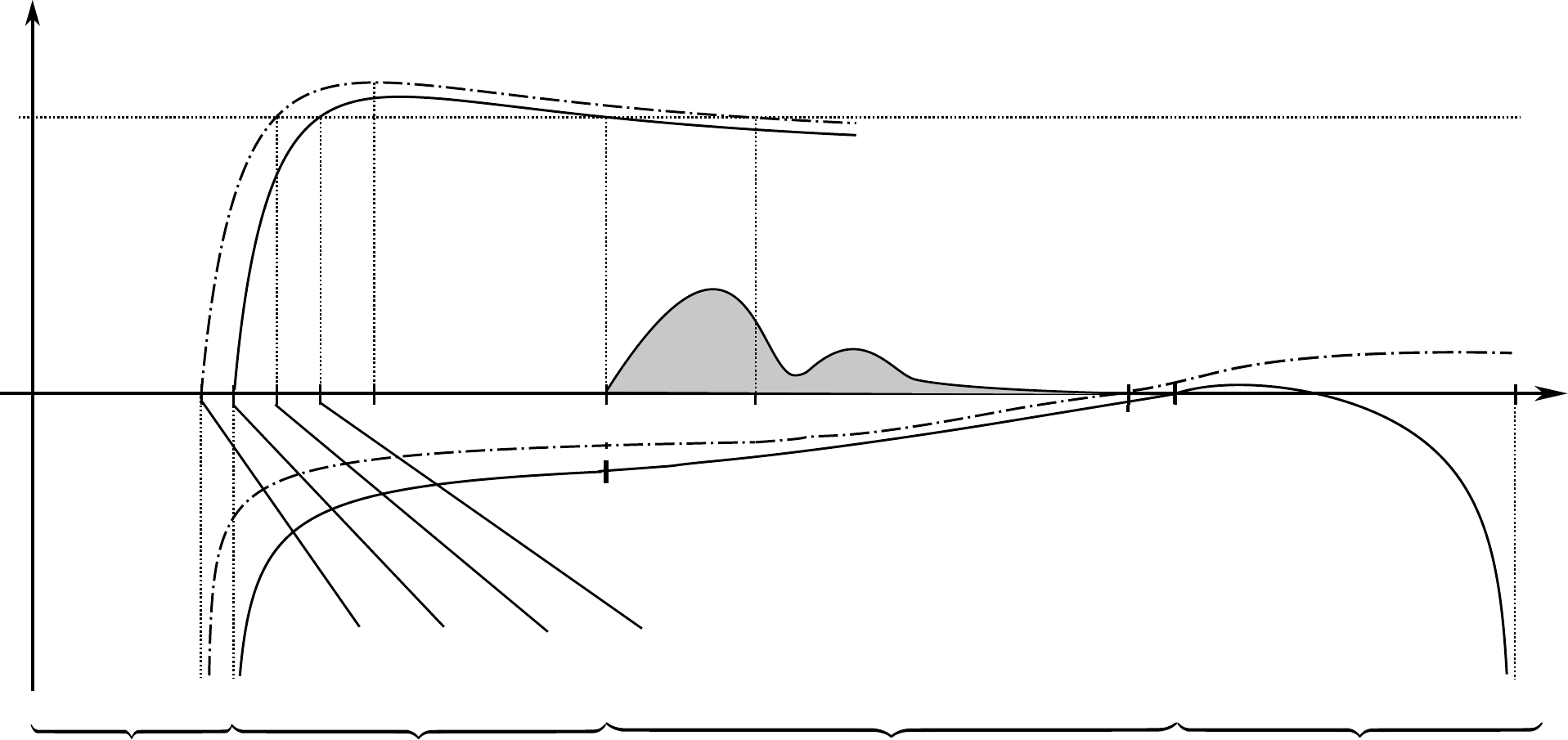}}
\put(-0.2,15.65){1}
\put(17.7,16.5){$a(r)$}
\put(19.3,14.5){$a_\Lambda(r)$}
\put(15.5,10){$\varrho_\Lambda(r)$}
\put(26.6,9.8){$R_0$}
\put(28,9.8){$R_{0\Lambda}$}
\put(9.05,8.3){$\hat r$}
\put(14.2,8.5){$r_{+\Lambda}$}
\put(18,8.5){$r_+$}
\put(38,8.5){$r$}
\put(8.3,2.7){$r_B$}
\put(10.2,2.7){$r_{B\Lambda}$}
\put(13.2,2.7){$r_-$}
\put(15.2,2.7){$r_{-\Lambda}$}
\put(13.3,6.2){gluing}
\put(27.1,8){gluing}
\put(34,11){$\mu(r)$}
\put(32.5,6){$\mu_\Lambda(r)$}
\put(0.7,0){black hole}
\put(8.3,0){vacuum}
\put(20,0){matter}
\put(31,0){vacuum}
\end{picture}
\caption{{Qualitative sketch of a black hole configuration surrounded by a shell of matter} \label{fig_radii}}
\end{center}
\end{figure}
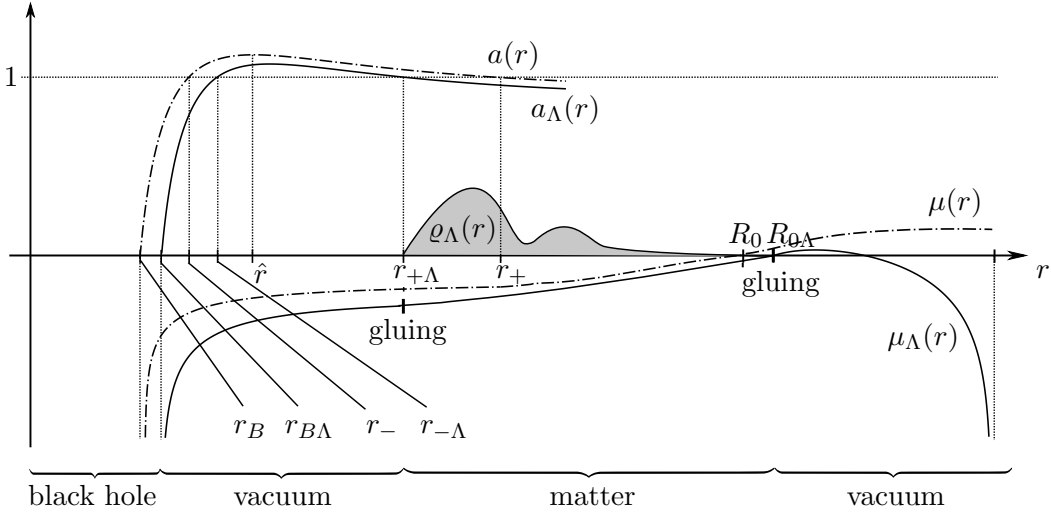
We define the functions
\begin{eqnarray}
a(r) &=& \sqrt{1-\frac{2M_0}{r}}\sqrt{1+\frac{L_0}{r^2}}, \\
a_\Lambda(r) &=& \sqrt{1-\frac{r^2\Lambda}{3}-\frac{2M_0}{r}}\sqrt{1+\frac{L_0}{r^2}}.
\end{eqnarray}
Moreover, we define $r_-$ and $r_+$ to be the first and second radius where $a(r)=1$, respectively, and $r_B:=2M_0$ to be the event horizon of the black hole. Since $L_0>16 M_0^2$ we have $r_B < r_- < r_+$ (cf.~\cite{rein94}). {Note also that $r_+ > 4M_0 > \frac{18}{5}M_0$.}\par
Since $9M_0^2\Lambda < 1$ by assumption ($\Lambda$ is chosen to be small), there exists a black hole horizon $r_{B\Lambda}$ of the Schwarzschild-deSitter metric with parameters $M_0$ and $\Lambda$. It can be calculated explicitly\footnote{To assure oneself of that one has chosen the right zero, {using $\frac{\mathrm d}{\mathrm dx}\arccos(x)=-\frac{1}{\sqrt{1-x^2}}$} one checks {$$\lim_{\Lambda \to 0} r_{B\Lambda} \stackrel{\mathrm{l'H\hat{o}spital}}{=} \lim_{\Lambda\to 0}\frac{2\sin\left(\frac 1 3 \arccos\left(-3M_0\sqrt \Lambda\right)+\frac \pi 3\right)\frac{1}{3\sqrt{1-(3M_0\sqrt \Lambda)^2}}\frac{3M_0}{2\sqrt \Lambda}}{1/\left(2\sqrt \Lambda\right)} = 2M_0.$$}} by
\begin{equation}
r_{B\Lambda} = -\frac{2}{\sqrt{\Lambda}} \cos\left(\frac{1}{3}\arccos\left(-3M_0\sqrt{\Lambda}\right)+\frac{\pi}{3}\right).
\end{equation}
Note that $r_B < r_{B\Lambda}$. We construct an upper bound to $r_{B\Lambda}$. Set $v(r) = 1-\frac{2M_0}{r}$.
\begin{equation}
\begin{aligned}
v(r_{B\Lambda}) &= \int_{r_B}^{r_{B\Lambda}} v'(s)\mathrm ds + \underbrace{v(r_B)}_{=0}\\
& \geq \int_{r_B}^{r_{B\Lambda}}\left(\inf_{s\in[r_B,r_{B\Lambda}]}v'(s)\right)\mathrm ds = (r_{B\Lambda}-r_B)v'(r_{B\Lambda}) \\
\Rightarrow \quad r_{B\Lambda} &\leq r_B + \frac{v(r_{B\Lambda})}{v'(r_{B\Lambda})}
\end{aligned}
\end{equation}
A short calculation yields $v(r_{B\Lambda})=\frac{r_{B\Lambda}^2\Lambda}{3}$ and $v'(r_{B\Lambda})=\frac{2M_0}{r_{B\Lambda}^2}$. One also checks by explicit calculation that $\frac{\mathrm d r_{B\Lambda}}{\mathrm d \Lambda} > 0$. So the distance
\begin{equation}
r_{B\Lambda} - r_B \leq \frac{r_{B\Lambda}^4 \Lambda}{6M_0}
\end{equation}
between the two horizons can be made arbitrarily small if $\Lambda$ is chosen to be sufficiently small. In particular we need $\Lambda$ to be small enough to assure $r_{B\Lambda} < r_-$. \par
Next we define $r_{-\Lambda}$ and $r_{+\Lambda}$ to be the first and second radius where $a_\Lambda(r)=1$. Note that $a(r) > a_\Lambda(r)$ for all $r\in(r_{B\Lambda},r_C)$, where $r_C$ is the cosmological horizon {of the vacuum solution}, thus the second positive zero of $1-r^2\Lambda/3-2M_0/r$. Between $r_-$ and $r_+$ the function $a(r)$ has a {unique} maximum at $r=\hat r$, given by
\begin{equation}
\hat r = \frac{L_0-\sqrt{L_0^2-12M_0^2L_0}}{2M_0}.
\end{equation}
We consider the distance between $a^2(r)$ and $a_\Lambda^2(r)$ at this radius $\hat r$:
\begin{equation}
|a^2(\hat r)-a_\Lambda^2(\hat r)| = \Lambda \frac{\hat r^2 + L_0}{3}.
\end{equation}
%Since $\Lambda < \frac{3(L_0-2M_0\hat r)}{\hat r^2(\hat r^2 + L_0^2)}$ by assumption we have
Choosing $\Lambda$ sufficiently small one can attain $|a^2(\hat r) - a_\Lambda^2(\hat r)| < a^2(\hat r) -1$. This implies that $a_\Lambda(r)-1$ has {exactly} two zeros in the interval $(r_-,r_+)$. This in turn yields the desired configuration
\begin{equation}
2M_0=r_B < r_{B\Lambda} < r_- < r_{-\Lambda} < \hat r < r_{+\Lambda} < r_+.
\end{equation} \par
In the vacuum region $[r_{-\Lambda},r_{+\Lambda}]$ the function $a_{\Lambda}(r)$ coincides with the expression $e^{-y_{\Lambda}(r)}\sqrt{1+\frac{L_0}{r^2}}$. Lemma \ref{lemmagh}, (\ref{ghvac}) implies that therefore for $r\in[r_{-\Lambda},r_{+\Lambda}]$ also the ansatz $\Phi$ for the distribution function $f$ yields $\varrho_\Lambda(r) = G_\phi(r,y_\Lambda(r))=0$ and $p_\Lambda(r)=H_\phi(r,y_\Lambda(r))=0$. So at $r=r_{+\Lambda}$ one can continue $f$ by the ansatz $\Phi$ in a continuous way and for $r\geq r_{+\Lambda}$ the Einstein equations lead to the differential equation
\begin{equation}\label{eqmulbh}
\begin{aligned}
\mu_\Lambda' &= \frac{1}{1 - \frac{\Lambda}{3} \left(r^2-\frac{r_{+\Lambda}^3}{r}\right) - \frac{2}{r} \left(\frac{r_{+\Lambda}}{2} \left(1-e^{-2\lambda_0}\right) + 4\pi\int_{r_{+\Lambda}}^r s^2\varrho_\Lambda(s)\mathrm ds\right)} \\
&\qquad\times
\left( 4\pi rp_{\Lambda} - \Lambda \left(\frac{r}{3}+\frac{r_{+\Lambda}^3}{6r^2} \right) + \frac{r_{+\Lambda}}{2r^2} \left(1-e^{-2\lambda_0}\right) + \frac{4\pi}{r^2} \int_{r_{+\Lambda}}^r s^2 \varrho_\Lambda(s) \mathrm ds \right)
\end{aligned}
\end{equation}
where $\lambda_0=\lambda(r_{+\Lambda})$. \par
%\begin{equation}
%\mu_\Lambda'(r)=\frac{1}{1-\frac{r^2\Lambda}{3}-\frac{2}{r}\left(M_0 + m_\Lambda(r)\right)}\left(4\pi rH_\phi(r,\mu(r))-\frac{r\Lambda}{3}+\frac{1}{r^2}\left(M_0 + m_\Lambda(r)\right)\right)
%\end{equation}
%where $m_\Lambda(r)=4\pi\int_0^r s^2G_\phi(s,\mu_\Lambda(s))\mathrm ds$.\par
There exists a background solution $\mu\in C^2((2M_0,\infty))$ to the Einstein equations with $\Lambda = 0$ (cf.~\cite{rein94}). For $r\in(2M_0, r_{+\Lambda}]$ this solution is given by the Schwarzschild metric and for $r>r_{+\Lambda}$ as a solution of equation \eqref{eqmulbh} with $\Lambda = 0$. The background solution is continuous at $r_{+\Lambda}$ if
\begin{equation}
\frac{r_{+\Lambda}}{2}\left(1-e^{-2\lambda_0}\right)=M_0.
\end{equation}
Furthermore, the background solution $\mu$ has the property that there exists $R_0>0$ such that $\mu(R_0)=0$ which implies that the support of matter quantities $\varrho$ and $p$ is contained in the interval $(r_+,R_0)$ (cf.~\cite{rein94}). In the remainder of the proof we show that using properties of this background solution $\mu$ one obtains a global solution $\mu_\Lambda$ of equation \eqref{eqmulbh}. We set
\begin{eqnarray}
\mu_{0\Lambda} &=& \frac 1 2 \ln\left(1-\frac{r_{+\Lambda}^2\Lambda}{3} - \frac{2M_0}{r_{+\Lambda}}\right),\label{mu0lbh} \\
\mu_0 &=& \mu(r_{+\Lambda}) = \frac 1 2 \ln\left(1-\frac{2M_0}{r_{+\Lambda}}\right).
\end{eqnarray}
In the following we seek for a solution $\mu_\Lambda$ of equation \eqref{eqmulbh} on on an interval beginning at $r=r_{+\Lambda}$ with the initial value $\mu_{0\Lambda}$ at given in \eqref{mu0lbh} that we can glue to the vacuum solution on $(r_{B\Lambda},r_{+\Lambda}]$. Note that $\mu_{0\Lambda} < 0$.
%since $\Lambda < 3\frac{r_{+\Lambda}^2-2M_0}{r_{+\Lambda}^4}\left(1-e^{-\mu(R_0+\Delta R)}\right)$ by assumption. \par
Since there are no issues with an irregular center the local existence of $\mu_\Lambda$ on an interval $(r_{+\Lambda},r_{+\Lambda}+\delta]$, $\delta > 0$ follows from the regularity of the right hand side of equation \eqref{eqmulbh}. So let $(2M_0,R_c)$ be the maximum interval of existence of $\mu_\Lambda$. We define
\begin{eqnarray}
v_{M_0}(r) &=& 1 - \frac{2}{r}\left(M_0+4\pi\int_{r_{+\Lambda}}^r s^2\varrho(s)\mathrm ds\right), \\
v_{M_0\Lambda}(r)&=&1-\frac{\Lambda}{3}\left(r^2-\frac{r_{+\Lambda}^3}{r}\right) - \frac{2}{r}\left(M_0+4\pi\int_{r_{+\Lambda}}^r s^2\varrho_\Lambda(s)\mathrm ds\right)
\end{eqnarray}
as the denominator of the right hand side of equation \eqref{eqmulbh}. We set \begin{equation}
\Delta v_0 := \frac{1}{18} v_{M_0\Lambda} (r_{+\Lambda})=\frac{1-\frac{2M_0}{r_{+\Lambda}}}{18} \leq \frac{1}{18}
\end{equation}
and define the radii

\begin{equation}
\begin{aligned}
r^* &= \inf\left\{r\in(r_{+\Lambda},R_c)\;\left|\;v_{M_0\Lambda}(r)= \Delta v_0 \right.\right\}, \\
\tilde r &= \sup\{r\in(r_{+\Lambda},R_c)\;|\;|\mu_\Lambda(r)-\mu(r)| \leq \mu(R_0+\Delta R)\},
\end{aligned}
\end{equation}
and set $\tilde r^* := \min\{\tilde r, r^*\}$. Note that $\mu(R_0+\Delta R)>0$ since $\mu(R_0)=0$ and $\mu$ is strictly increasing.
We assume that $r\leq \tilde r^*$ and calculate $|\mu(r) - \mu_\Lambda(r)|$. % = \left| \int_{r_{+\Lambda}}^r (\mu'(s)-\mu_\Lambda'(s))\mathrm ds + (\mu_0-\mu_{0\Lambda})\right|.
To make calculations more convenient, we extend $\varrho$ and $p$ on $[0,2M_0]$ as constant zero such that integrals of $\varrho$ and $p$ over $(r_{+},r)$ can be replaced by integrals over $(0,r)$. First we calculate
\begin{equation}
|\mu_0-\mu_{0\Lambda}| = \frac 1 2 \ln\left[1+\frac{r_{+\Lambda}^2 \Lambda}{3}\left(1-\frac{r_{+\Lambda}^2\Lambda}{3}-\frac{2M_0}{r_{+\Lambda}}\right)^{-1}\right] =:C_{0\Lambda}(r). %< \frac{\mu(R_0+\Delta R)}{2}
\end{equation}
We write
\begin{equation}
\begin{aligned}
&|\mu(r) - \mu_\Lambda(r)| \\
&\leq \int_{r_{+\Lambda}}^r \frac{1}{v_{M_0\Lambda}(s)} \bigg[4\pi s|p_\Lambda(s)-p(s)| - \Lambda \left(\frac{s}{3} + \frac{r_{+\Lambda}^3}{s^2}\right)\\
&\qquad\qquad\qquad\qquad + \frac{4\pi}{s^2} \int_0^s \sigma^2 |\varrho_\Lambda(\sigma)-\varrho(\sigma)| \mathrm d\sigma \bigg]\mathrm ds\\
&\quad+ \int_{r_{+\Lambda}}^r \left(4\pi sp(s) + \frac{4\pi}{s^2} \int_0^s \sigma^2 \varrho(\sigma)\mathrm d\sigma\right) \left|\frac{1}{v_{M_0\Lambda}(s)}-\frac{1}{v_{M_0}(s)}\right| \mathrm ds + C_{0\Lambda}(r)
\end{aligned}
\end{equation}
{We would like to apply the generalized Buchdahl inequality (Lemma \ref{buchdahlbh}) to the background solution $\mu$ on the interval $[r_{+\Lambda},\infty)$. We have that $r_{+\Lambda} > \hat r \geq 3M_0 > 9/4M_0 $. The crucial condition is the existence of a vacuum region on $\left(2M_0,\frac{9}{4}M_0\right]$. But this is ensured by virtue of the assumption $L_0 > 16M_0^2$ which implies $r_+>4M_0$.} So the difference $|\mu(r)-\mu_\Lambda(r)|$ can be further simplified and estimated. Using similar estimates as in Appendix B we obtain
an inequality of the form
\begin{equation}
|\mu(r) - \mu_\Lambda(r)| \leq C_\Lambda(r) + C(r) \int_0^r \left(|p(s)-p_\Lambda(s)|+|\varrho(s)-\varrho_\Lambda(s)|\right)\mathrm ds
\end{equation}
where $C(r)$ is increasing in $r$, $C_\Lambda(r)$ is increasing both in $\Lambda$ and $r$ and we have $C_\Lambda(r) = 0$ if $\Lambda = 0$. Note that the constants are fully determined by $M_0$, $L_0$, $\phi$ and $\mu$. \par
In virtue of the mean value theorem, the sum $|p_\Lambda - p|+|\varrho_\Lambda-\varrho|$ can be estimated as
\begin{equation}
|p_\Lambda(r) - p(r)|+|\varrho_\Lambda(r)-\varrho(r)| \leq C \cdot |\mu_\Lambda(r) - \mu(r)|,
\end{equation}
where the constant $C$ is determined by the derivatives of $G_\phi$ and $H_\phi$. A Gr\"onwall argument yields $|\mu_\Lambda(r)-\mu(r)| \leq C_{\mu\Lambda}(r)$ implying $|\varrho_\Lambda(r)-\varrho(r)| \leq C_{g\Lambda}(r)$ with certain constants $C_{g\Lambda}$ and $C_{\mu\Lambda}$.\par
One can choose $\Lambda$ small enough such that for all $r\in(r_{+\Lambda},R_0+\Delta R]$ we have
\begin{equation}
|\mu_\Lambda(r)-\mu(r)|<\mu(R_0+\Delta R).
\end{equation}
Moreover, we consider the difference
\begin{equation}
\left|v_{M_0}(r)-v_{M_0\Lambda}(r)\right| \leq \frac \Lambda 3 \left| r^2-\frac{r_{+\Lambda}^3}{r} \right| + \frac{8\pi r^2}{3} C_{g\Lambda}(r).
\end{equation}
Lemma \ref{buchdahlbh} implies $v_{M_0}(r)\geq \frac 1 9$ for all $r\in(r_{+\Lambda},\infty)$. Choosing $\Lambda$ sufficiently small, such that for all $r\in(r_{+\Lambda},R_0+\Delta R]$ we have $\left|v_{M_0}(r)-v_{M_0\Lambda}(r)\right| \leq \frac{1}{18}$ one obtains $v_{M_0\Lambda} \geq \frac{1}{18}$ on $(r_{+\Lambda},R_0+\Delta R]$.
\par
So altogether, one has deduced that $\tilde r^* \geq R_0 + \Delta R$ if $\Lambda$ is chosen sufficiently small. This implies that $\mu_\Lambda$ exists at least on $[0,R_0+\Delta R]$ by Lemma \ref{genlemden} and also that $\mu_\Lambda(R_0+\Delta R) > 0$. From the latter property one deduces that there exists a radius $R_{0\Lambda} > R_0$ such that for all $r\in[R_{0\Lambda},R_0+\Delta R]$ we have $\varrho_\Lambda(r) = p_\Lambda(r) = 0$. On this interval, we can glue an appropriately shifted Schwarzschild-de Sitter metric to $\mu_\Lambda$. This yields the desired solution defined on $(r_{B\Lambda},\infty)$.
\end{proof}

\begin{rem}
To see that the solutions constructed in Theorem \ref{theo_la_pos_bh} are non-vacuum, one checks {that for $r\geq r_{+\Lambda}$ one has
\begin{equation}
\frac{\mathrm d}{\mathrm dr}a_\Lambda(r) < 0 \quad \mathrm{and} \quad \frac{\mathrm d^2}{\mathrm dr^2}a_\Lambda(r) \leq 0.
\end{equation}
Since $a_\Lambda(r)$ corresponds to $e^{-y_\Lambda(r)}$,} this implies that for some $r > r_{+\Lambda}$ the quantity $e^{-y_\Lambda(r)}\sqrt{1+\frac{L_0}{r^2}} < 1$ which in turn implies by Lemma \ref{lemmagh}, (\ref{ghvac}) that $\varrho_\Lambda(r),p_\Lambda(r) > 0$ for some $r > r_{+\Lambda}$.
\end{rem}

\begin{rem} \label{rem_shift_bh}
In contrary to the metric without a singularity at the center, the metric with a Schwarzschild singularity does not coincide with the not shifted Schwarzschild-de Sitter solution for $r>R_{0\Lambda}$. This can be seen as follows. We have
\begin{equation}
\mu_\Lambda'(r)\geq \frac 1 2 \frac{\mathrm d}{\mathrm dr} \ln\left(1-\frac{r^2\Lambda}{3}-\frac{2M_0}{r}\right).
\end{equation}
%\begin{equation}
%\int_{r_{+\Lambda}}^{R_{0\Lambda}} \mu_\Lambda'(s) \mathrm ds > 0.
%\end{equation}
Certainly, the mass parameter $M$ of the vacuum solution, that is glued on in the outer region, is larger than $M_0$. This implies
\begin{equation}
1-\frac{r^2\Lambda}{3}-\frac{2M_0}{r} > 1-\frac{r^2\Lambda}{3}-\frac{2M}{r}
\end{equation}
for all $r\in(r_{B\Lambda},r_C)$. So there is no ansatz $\Phi$ for the matter distribution that yields a metric component $\mu_\Lambda$ that connects the two vacuum solutions without any shift. But by suitable choice of $\Phi$ and $E_0$ one can determine whether the inner or the outer Schwarzschild-deSitter metric is shifted. For the maximal $C^2$-extension of the metric constructed in Theorem \ref{theo_la_pos_bh} we will need the aolution to coincide with the not shifted Schwarzschild-deSitter metric for $r>R_{0\Lambda}$.
\end{rem}

\subsection{Matter shells immersed in Schwarzschild-AdS spacetimes}\label{bhAdS}
We construct solutions of the Einstein-Vlasov system with a Schwarzschild singularity at the center for the case $\Lambda < 0$. The result is given in the following theorem.

\begin{theorem}\label{theo_la_neg_bh}
Let $\Phi:\mathbb R^2 \to [0,\infty)$ be of the form (\ref{ouransf}) and let $L_0,M_0 \geq 0$ such that $L_0 < 16M_0^2$. %Choose $\Lambda < 0$ such that $|\Lambda|$ is sufficiently small.
Then there exists a unique solution $\mu_\Lambda, \lambda_\Lambda \in C^2((r_{B\Lambda},\infty))$, $f\in C^0((r_{B\Lambda},\infty)\times \mathbb R^3)$ of the Einstein-Vlasov system \eqref{eqvlasov} -- \eqref{p} for $\Lambda < 0$ and $|\Lambda|$ sufficiently small. The spatial support of the distribution function $f_\Lambda$ is contained in a shell, $\{r_{+\Lambda} < r < R_{0\Lambda}\}$. In {the complement} of this shell, the solution of the Einstein equations is given by the Schwarzschild-AdS metric.
\end{theorem}

\begin{proof}
We define $r_B:=2M_0$ to be the Schwarzschild black hole horizon of the background solution and $r_{B\Lambda}$ to be the black hole horizon for the Schwarzschild-AdS with $\Lambda < 0$, i.e.~the smallest positive zero of $1-r^2\Lambda/3-2M_0/r$. Define also the functions
\begin{eqnarray}
a(r) &=& \sqrt{1-\frac{2M_0}{r}}\sqrt{1+\frac{L_0}{r^2}}, \\
a_\Lambda(r) &=& \sqrt{1-\frac{r^2\Lambda}{3}-\frac{2M_0}{r}}\sqrt{1+\frac{L_0}{r^2}}.
\end{eqnarray}
Moreover we define $r_-$ and $r_+$ to be the first and second positive zero of $a(r)-1$, respectively, as well as $r_{-\Lambda}$ and $r_{+\Lambda}$ to be the first and second positive zero of $a_\Lambda(r)-1$. The assumption $L_0 < 16M_0^2$ assures that $r_B < r_- < r_+$ but a priori $r_{+\Lambda}=\infty$ and $r_{-\Lambda}=\infty$ are possible. However, we show that the configuration is
\begin{equation} \label{eqconfads}
r_{B\Lambda} < r_{-\Lambda} < r_{+\Lambda} < \infty.
\end{equation}
First, we observe that $a(r) < 1$ for all $r > r_+$ and also that $a_\Lambda(r) > a(r)$ for all $r\in\mathbb R_+$ since $\Lambda < 0$. So we have $r_{B_\Lambda} < r_{-\Lambda} < r_- < r_+ < r_{+\Lambda}$. It remains to show that $r_{+\Lambda} < \infty$. This is done by showing that for $|\Lambda|$ sufficiently small the functions $a$ and $a_\Lambda$ are sufficiently close at a radius $r_++\Delta r$, $\Delta r > 0$ such that $a_\Lambda(r_++\Delta r) < 1$. So we consider the difference $|a_\Lambda^2(r) - a^2(r)|$ at the radius $r_++\Delta r$:
\begin{equation}
|a_\Lambda^2(r_++\Delta r) - a^2(r_++\Delta r)| = |\Lambda|\frac{(r_++\Delta r)^2+L_0}{3}.
\end{equation}
Choosing $|\Lambda|$ small one attains this difference to be smaller than $a(r_++\Delta r)-1$ which implies $r_{+\Lambda} < r_++\Delta r < \infty$. \par
Given this configuration \eqref{eqconfads} we construct a global solution of the Einstein-Vlasov system in the following manner. For $r\in(r_{B\Lambda},r_{+\Lambda}]$ we set $f(x,v)\equiv 0$ and
\begin{equation} \label{eq_ads_mu0}
\mu_\Lambda(r) = \frac 1 2\ln\left(1- \frac{r^2\Lambda}{3} - \frac{2M_0}{r}\right).
\end{equation}
For $r \geq r_{+\Lambda}$ we set $f(x,v)=\Phi(E,L)$. Since also $\Phi(E,L) = 0$ on the interval $(r_{-\Lambda},r_{+\Lambda})$ the distribution function $f$ is continuous and the metric coefficient $\mu_\Lambda$ is given by the ODE \eqref{eqmulbh} with $\Lambda < 0$ for all $r\in(r_{-\Lambda},\infty)$. The initial value $\lambda_0$ of $\lambda_\Lambda$ is determined by the continuity criterion
\begin{equation} \label{eq_ads_l0}
\frac{r_{+\Lambda}}{2}\left(1-e^{-2\lambda_0}\right)=M_0.
\end{equation}
The last step of the proof is to assure for the existence of a solution with the desired properties of the Einstein-Vlasov system on $[r_{+\Lambda},\infty)$ with initial values $\lambda_0$ given by \eqref{eq_ads_l0} and $\mu_0 = \mu_\Lambda(r_{+\Lambda})$, given by equation \eqref{eq_ads_mu0}. But this is already implied by the Theorems \ref{theo_glo_adS} and \ref{theo_bou_adS}.
\end{proof}

\section{Solutions on $\mathbb R\times S^3$ and $\mathbb R\times S^2\times\mathbb R$}\label{sec : glob sol}
In Sections \ref{ssec : sagr-4} and \ref{sec dS bh} we constructed spherically symmetric, static solutions of the Einstein-Vlasov system with small positive cosmological constant $\Lambda$. For small radii the $\Lambda$-term plays only a minor role. This was crucial for the method of proof. However, the global structure of the constructed spacetime is substantially different when $\Lambda > 0$ and shows interesting properties. In particular, it allows for solutions with different global topologies. \par

The following theorem gives a class of new solutions to the non-vacuum field equations with non-trivial global topology. These solutions are constructed from pieces consisting of solutions constructed in Theorems \ref{main-thm} and \ref{theo_la_pos_bh}.

\begin{theorem}\label{thm-top}
Let $\Lambda >0$ be sufficiently small and let $\mathscr M_1=\mathbb R\times  S^3$ and $\mathscr M_2=\mathbb R\times S^2\times \mathbb R$. The following types of static metrics solving the Einstein-Vlasov system exist on these topologies.
\begin{enumerate}
\item There is a class of static metrics on $\mathscr M_1$, which is characterized in Figure \ref{penrose_simple}. In regions I and IV a metric in this class coincides with two a priori different solutions of the type constructed in Theorem \ref{main-thm} with identical total mass, but possibly different matter distributions and radii of the support of the matter quantities $R_1$ and $R_2$ and regular centers. The metric in regions II and III is vacuum. \label{class1}
\item There is a class of static metrics on $\mathscr M_1$, which is characterized in Figure \ref{penrose_diagram}. A metric in this class consists of two regular centers with finitely extended matter distribution around each of the centers of equal mass but possible different matter distributions and radii $R_1$, $R_2$ of the type constructed in Theorem \ref{main-thm}. These two regions are connected by a chain of black holes of identical masses (the diagram shows the minimal configuration with one black hole). \label{class2}
\item There is a class of metrics on $\mathscr M_2$, which is characterized in Figure \ref{penrose_periodic}. The spacetime consists of an infinite sequence of black holes, each surrounded by matter shells of possibly different radii and positions. In regions IV, VII, X and XIII these solutions coincide with those constructed in Theorem \ref{theo_la_pos_bh}. The necessary conditions on the masses are $M_{\varrho}^{A_1}=M_\varrho^{A_2}$, $M_{\varrho}^{B_1}=M_\varrho^{B_2}$ and $M_0^A+M_\varrho^{A_2}=M_{\varrho}^{B_1}+M_0^B$, where $M_0^i$, $i=A,B$, denote the mass parameter of the black holes and $M_\varrho^{i_j}$, $i=A,B$, $j=1,2$ denote the quasilocal mass of the matter shells defined in equation (\ref{eq_p3_quasilocal}). \label{class3}
\end{enumerate}
\end{theorem}

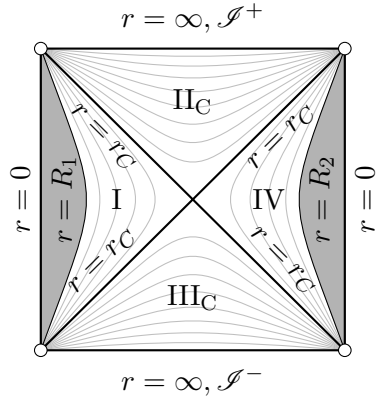
\begin{figure}[ht]
\begin{center}
\begin{tikzpicture}
\pgfdeclarelayer{bg}
\pgfsetlayers{bg,main}

\node (I1)    at (-1,0)  {I};
\node (IV1)   at (1,0)  {IV};
\node (IIC1)  at (0, 1.3)   {$\mathrm{II}_\mathrm{C}$};
\node (IIIC1) at (0,-1.3)   {$\mathrm{III}_\mathrm{C}$};

\path  % Three corners of left triangle
  (I1) +(-1,2)  coordinate[label=90:]  (I1top)
       +(-1,-2) coordinate[label=-90:] (I1bot)
       +(0:1) coordinate               (I1right)
       ;
\path % Three corners of the right triangle
   (IV1) +(1,2)  coordinate[label=90:]  (IV1top)
         +(1,-2) coordinate[label=-90:] (IV1bot)
         +(180:1)  coordinate  (IV1left)
         ;

\begin{scope}[on background layer]
% Lines of constant r
\foreach \a in {1,...,8}{
\pgfmathparse{1+(\a-1)/7}
\edef\w{\pgfmathresult}
\draw[lightgray] (I1top) to[out =-90+\a*5, in=90-\a*5, looseness=\w] (I1bot);
\draw[lightgray] (I1top) to[out=-\a*5, in=-180+\a*5, looseness=\w] (IV1top);
\draw[lightgray] (I1bot) to[out=\a*5, in=180-\a*5, looseness=\w] (IV1bot);
\draw[lightgray] (IV1top) to[out=-90-\a*5, in=90+\a*5, looseness=\w] (IV1bot);
}
\end{scope}

% left matter region
\draw[-, fill=black!30]
      (I1top) -- (I1bot) to[out=70, in=-70, looseness=1.5] node[midway, above, sloped] {$r=R_1$} (I1top) -- cycle;
\draw[thick] (I1bot) -- node[midway, above, sloped] {$r=0$} (I1top);

\draw[thick] (I1top) --
          node[midway, below, sloped] {$r=r_C$}
      (I1right) --
          node[midway, above, sloped] {$r=r_C$}
      (I1bot);

% Top,Bottom lines
\draw[thick] (I1top) -- (IV1top)
      node[midway, above, inner sep=2mm] {$r=\infty,\mathscr I^+$};
\draw[thick] (I1bot) -- (IV1bot)
      node[midway, below, inner sep=2mm] {$r=\infty,\mathscr I^-$};

% right matter region
\draw[-, fill=black!30]
      (IV1top) -- (IV1bot) to[out=110, in=-110, looseness=1.5] node[midway, below, sloped] {$r=R_2$} (IV1top) -- cycle;
\draw[thick] (IV1bot) -- node[midway, below, sloped] {$r=0$} (IV1top);

\draw[thick] (IV1bot) --
          node[midway, above, sloped] {$r=r_{C}$}
      (IV1left) --
          node[midway, below, sloped] {$r=r_{C}$}
      (IV1top);

% missing points
\path[fill=white,draw=black] (I1top) circle (0.5ex);
\path[fill=white,draw=black] (I1bot) circle (0.5ex);
\path[fill=white,draw=black] (IV1top) circle (0.5ex);
\path[fill=white,draw=black] (IV1bot) circle (0.5ex);
\end{tikzpicture}
\caption{Penrose diagram of the maximal $C^2$-extension of a metric constructed as spherically symmetric solution of the Einstein-Vlasov system. Region I corresponds to the region $0<r<r_C$. The metric is extended in an analogue way to the standard extension of the deSitter metric. The gray lines are surfaces of constant $r$. \label{penrose_simple}}
\end{center}
\end{figure}

%%%%%%
%%%%%%
%%%%%%
%%%%%%

\begin{figure}[ht]
\begin{center}
\begin{tikzpicture}
\pgfdeclarelayer{bg}
\pgfsetlayers{bg,main}

\node (I1)    at (-5,0)  {I};
\node (IV1)   at (-2,0)  {IV};
\node (I2)    at (2,0)   {VII}; % formals I
\node (IV2)   at (5,0)   {X}; % formals IV
\node (IIC1)  at (-4, 1.3)   {$\mathrm{II}_\mathrm{C}$};
\node (IIIC1) at (-4,-1.3)   {$\mathrm{III}_\mathrm{C}$};
\node (IIH1)  at (0, 1.3)    {$\mathrm{V}_\mathrm{BH}$}; % formals II_H
\node (IIIH1) at (0,-1.3)    {$\mathrm{VI}_\mathrm{WH}$}; % formals III_H
\node (IIC2)  at (4, 1.3)    {$\mathrm{VIII}_\mathrm{C}$}; % formals II_C
\node (IIIC2) at (4,-1.3)    {$\mathrm{IX}_\mathrm{C}$}; % formals III_C

\path  % Three corners of left triangle
  (I1) +(-1,2)  coordinate[label=90:]  (I1top)
       +(-1,-2) coordinate[label=-90:] (I1bot)
       +(0:1) coordinate                  (I1right)
       ;
\path  % Four corners of middle left diamond
  (IV1) +(90:2)  coordinate  (IV1top)
        +(-90:2) coordinate  (IV1bot)
        +(0:2)   coordinate  (IV1right)
        +(180:2) coordinate  (IV1left)
        ;
\path % Four corners of the middle right diamond (no labels this time)
   (I2) +(90:2)  coordinate (I2top)
        +(-90:2) coordinate (I2bot)
        +(180:2) coordinate (I2left)
        +(0:2)   coordinate (I2right)
        ;
\path % Three corners of the right triangle
   (IV2) +(1,2)  coordinate[label=90:]  (IV2top)
         +(1,-2) coordinate[label=-90:] (IV2bot)
         +(180:1)  coordinate  (IV2left)
         ;

\begin{scope}[on background layer]
% Lines of constant r
\foreach \a in {1,...,8}{
\pgfmathparse{1+(\a-1)/7}
\edef\w{\pgfmathresult}
\draw[lightgray] (I1top) to[out =-90+\a*5, in=90-\a*5, looseness=\w] (I1bot);
\draw[lightgray] (I1top) to[out=-\a*5, in=-180+\a*5, looseness=\w] (IV1top);
\draw[lightgray] (I1bot) to[out=\a*5, in=180-\a*5, looseness=\w] (IV1bot);
\draw[lightgray] (IV1top) to[out=-90-\a*5, in=90+\a*5, looseness=\w] (IV1bot);
\draw[lightgray] (IV1top) to[out=-90+\a*5, in=90-\a*5, looseness=\w] (IV1bot);
\draw[lightgray] (IV1top) to[out=-\a*5, in=-180+\a*5, looseness=\w] (I2top);
\draw[lightgray] (IV1bot) to[out=\a*5, in=180-\a*5, looseness=\w] (I2bot);
\draw[lightgray] (I2top) to[out=-90-\a*5, in=90+\a*5, looseness=\w] (I2bot);
\draw[lightgray] (I2top) to[out=-90+\a*5, in=90-\a*5, looseness=\w] (I2bot);
\draw[lightgray] (I2top) to[out=-\a*5, in=-180+\a*5, looseness=\w] (IV2top);
\draw[lightgray] (I2bot) to[out=\a*5, in=180-\a*5, looseness=\w] (IV2bot);
\draw[lightgray] (IV2top) to[out=-90-\a*5, in=90+\a*5, looseness=\w] (IV2bot);
}
\draw[lightgray] (IV1top) to[out=-90, in=90, looseness=0] (IV1bot);
\draw[lightgray] (I2top) to[out=-90, in=90, looseness=0] (I2bot);
\end{scope}

% left matter region
\draw[-, fill=black!30]
      (I1top) -- (I1bot) to[out=70, in=-70, looseness=1.5] node[midway, above, sloped] {$r=R_1$} (I1top) -- cycle;
\draw[thick] (I1bot) -- node[midway, above, sloped] {$r=0$} (I1top);

\draw[thick] (I1top) --
          node[midway, below, sloped] {$r=r_C$}
      (I1right) --
          node[midway, above, sloped] {$r=r_C$}
      (I1bot);

\draw[thick] (IV1left) --
          node[midway, below, sloped] {$r=r_C$}
      (IV1top) --
          node[midway, below, sloped] {$r=r_{B\Lambda}$}
      (IV1right) --
          node[midway, above, sloped] {$r=r_{B\Lambda}$}
      (IV1bot) --
          node[midway, above, sloped] {$r=r_C$}
      (IV1left) -- cycle;

% Top,Bottom lines
\draw[thick] (I1top) -- (IV1top)
      node[midway, above, inner sep=2mm] {$r=\infty,\mathscr I^+$};
\draw[thick] (I1bot) -- (IV1bot)
      node[midway, below, inner sep=2mm] {$r=\infty,\mathscr I^-$};

\draw[thick] (I2left) --
          node[midway, below, sloped] {$r=r_{B\Lambda}$}
      (I2top) --
          node[midway, below, sloped] {$r=r_{C}$}
      (I2right) --
          node[midway, above, sloped] {$r=r_{C}$}
      (I2bot) --
          node[midway, above, sloped] {$r=r_{B\Lambda}$}
      (I2left) -- cycle;

% right matter region
\draw[-, fill=black!30]
      (IV2top) -- (IV2bot) to[out=110, in=-110, looseness=1.5] node[midway, below, sloped] {$r=R_2$} (IV2top) -- cycle;
\draw[thick] (IV2bot) -- node[midway, below, sloped] {$r=0$} (IV2top);

\draw[thick] (IV2bot) --
          node[midway, above, sloped] {$r=r_{C}$}
      (IV2left) --
          node[midway, below, sloped] {$r=r_{C}$}
      (IV2top);

% Top,Bottom lines
\draw[thick] (I2top) -- (IV2top)
      node[midway, above, inner sep=2mm] {$r=\infty,\mathscr I^+$};
\draw[thick] (I2bot) -- (IV2bot)
      node[midway, below, inner sep=2mm] {$r=\infty,\mathscr I^-$};

% Squiggly lines
\draw[decorate,decoration={snake, segment length=3mm, amplitude=0.7mm},thick] (IV1top) -- (I2top)
      node[midway, above, inner sep=2mm] {$r=0$};

\draw[decorate,decoration={snake, segment length=3mm, amplitude=0.7mm},thick] (IV1bot) -- (I2bot)
      node[midway, below, inner sep=2mm] {$r=0$};

% missing points
\path[fill=white,draw=black] (I1top) circle (0.5ex);
\path[fill=white,draw=black] (I1bot) circle (0.5ex);
\path[fill=white,draw=black] (IV1top) circle (0.5ex);
\path[fill=white,draw=black] (IV1bot) circle (0.5ex);
\path[fill=white,draw=black] (I2top) circle (0.5ex);
\path[fill=white,draw=black] (I2bot) circle (0.5ex);
\path[fill=white,draw=black] (IV2top) circle (0.5ex);
\path[fill=white,draw=black] (IV2bot) circle (0.5ex);
\end{tikzpicture}
\caption{Penrose diagram of the maximal $C^2$-extension of a metric constructed as spherically symmetric solution of the Einstein-Vlasov system. Region I corresponds to the region $0<r<r_C$. In this region matter (represented by the shaded area) is present and the metric is regular. This metric is extended with the Schwarzschild-deSitter metric that leads to a periodic solution. The periodic course stops when a matter region appears again preventing the metric from being singular at $r=0$. The gray lines are surfaces of constant $r$. \label{penrose_diagram}}
\end{center}
\end{figure}
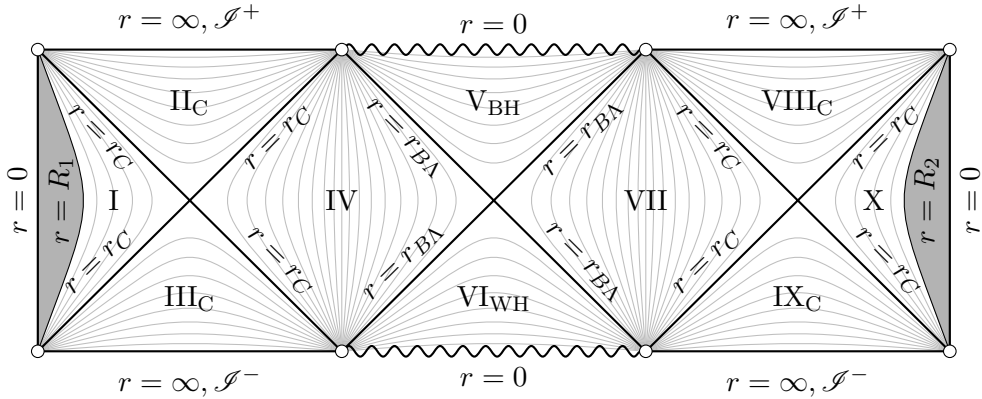

%%%%%%
%%%%%%
%%%%%%
%%%%%%

\begin{figure}[ht]
\begin{center}
\begin{tikzpicture}
\pgfdeclarelayer{bg}
\pgfsetlayers{bg,main}

\node (I)     at (-5,0)      {I};
\node (II)    at (-4,0.65)   {$\mathrm{II}_\mathrm{C}$};
\node (III)   at (-4,-0.65)  {$\mathrm{III}_\mathrm{C}$};
\node (IV)    at (-3,0)      {IV};
\node (V)     at (-2,0.65)   {$\mathrm{V}_\mathrm{B}$};
\node (VI)    at (-2,-0.65)  {$\mathrm{VI}_\mathrm{W}$};
\node (VII)   at (-1,0)      {VII};
\node (VIII)  at (0,0.65)    {$\mathrm{VIII}_\mathrm{C}$};
\node (IX)    at (0,-0.65)   {$\mathrm{IX}_\mathrm{C}$};
\node (X)     at (1,0)       {X};
\node (XI)    at (2,0.65)    {$\mathrm{XI}_\mathrm{B}$};
\node (XII)   at (2,-0.65)   {$\mathrm{XII}_\mathrm{W}$};
\node (XIII)  at (3,0)       {XIII};
\node (XIV)   at (4,0.65)    {$\mathrm{XIV}_\mathrm{C}$};
\node (XV)    at (4,-0.65)   {$\mathrm{XV}_\mathrm{C}$};
\node (XVI)   at (5,0)       {XVI};
\node (M_rho_a1) at (-3,-1.5) {$M_\varrho^{A_1}$};
\node (M_0_a) at (-2,-1.5) {$M_0^A$};
\node (M_rho_a2) at (-1,-1.5) {$M_\varrho^{A_2}$};
\node (M_rho_b1) at (1,-1.5) {$M_\varrho^{B_1}$};
\node (M_0_b) at (2,-1.5) {$M_0^B$};
\node (M_rho_b2) at (3,-1.5) {$M_\varrho^{B_2}$};

\path % left open diamond
(I) +(-1,1)  coordinate (Ilefttop)
    +(0,1)   coordinate (Itop)
    +(1,0)   coordinate (Iright)
    +(0,-1)  coordinate (Ibot)
    +(-1,-1) coordinate (Ileftbot)
;
\path  % first diamond
  (IV) +(-1,0)  coordinate  (IVleft)
  (IV) +(0,1)   coordinate  (IVtop)
  (IV) +(1,0)   coordinate  (IVright)
  (IV) +(0,-1)  coordinate  (IVbot)
;
\path  % first diamond
  (VII) +(-1,0)  coordinate  (VIIleft)
  (VII) +(0,1)   coordinate  (VIItop)
  (VII) +(1,0)   coordinate  (VIIright)
  (VII) +(0,-1)  coordinate  (VIIbot)
;
\path  % first diamond
  (X) +(-1,0)  coordinate  (Xleft)
  (X) +(0,1)   coordinate  (Xtop)
  (X) +(1,0)   coordinate  (Xright)
  (X) +(0,-1)  coordinate  (Xbot)
;
\path  % first diamond
  (XIII) +(-1,0)  coordinate  (XIIIleft)
  (XIII) +(0,1)   coordinate  (XIIItop)
  (XIII) +(1,0)   coordinate  (XIIIright)
  (XIII) +(0,-1)  coordinate  (XIIIbot)
;
\path % left open diamond
(XVI) +(1,1)  coordinate (XVIrt)
      +(0,1)   coordinate (XVItop)
      +(1,0)   coordinate (XVIleft)
      +(0,-1)  coordinate (XVIbot)
      +(1,-1) coordinate (XVIrb)
;
\begin{scope}[on background layer]
% Lines of constant r
\foreach \a in {2,4,6,8}{
\pgfmathparse{1+(\a-1)/7}
\edef\w{\pgfmathresult}
\draw[lightgray] (Itop)    to[out =-90+\a*5, in=90-\a*5, looseness=\w] (Ibot);
\draw[lightgray] (IVtop)   to[out =-90+\a*5, in=90-\a*5, looseness=\w] (IVbot);
\draw[lightgray] (VIItop)  to[out =-90+\a*5, in=90-\a*5, looseness=\w] (VIIbot);
\draw[lightgray] (Xtop)    to[out =-90+\a*5, in=90-\a*5, looseness=\w] (Xbot);
\draw[lightgray] (XIIItop) to[out =-90+\a*5, in=90-\a*5, looseness=\w] (XIIIbot);

\draw[lightgray] (IVtop)   to[out=-90-\a*5, in=90+\a*5, looseness=\w] (IVbot);
\draw[lightgray] (VIItop)  to[out=-90-\a*5, in=90+\a*5, looseness=\w] (VIIbot);
\draw[lightgray] (Xtop)    to[out=-90-\a*5, in=90+\a*5, looseness=\w] (Xbot);
\draw[lightgray] (XIIItop) to[out=-90-\a*5, in=90+\a*5, looseness=\w] (XIIIbot);
\draw[lightgray] (XVItop)  to[out=-90-\a*5, in=90+\a*5, looseness=\w] (XVIbot);

\draw[lightgray] (Itop)    to[out=-\a*5, in=-180+\a*5, looseness=\w] (IVtop);
\draw[lightgray] (IVtop)   to[out=-\a*5, in=-180+\a*5, looseness=\w] (VIItop);
\draw[lightgray] (VIItop)  to[out=-\a*5, in=-180+\a*5, looseness=\w] (Xtop);
\draw[lightgray] (Xtop)    to[out=-\a*5, in=-180+\a*5, looseness=\w] (XIIItop);
\draw[lightgray] (XIIItop) to[out=-\a*5, in=-180+\a*5, looseness=\w] (XVItop);

\draw[lightgray] (Ibot)    to[out=\a*5, in=180-\a*5, looseness=\w] (IVbot);
\draw[lightgray] (IVbot)   to[out=\a*5, in=180-\a*5, looseness=\w] (VIIbot);
\draw[lightgray] (VIIbot)  to[out=\a*5, in=180-\a*5, looseness=\w] (Xbot);
\draw[lightgray] (Xbot)    to[out=\a*5, in=180-\a*5, looseness=\w] (XIIIbot);
\draw[lightgray] (XIIIbot) to[out=\a*5, in=180-\a*5, looseness=\w] (XVIbot);
}
\draw[lightgray] (IVtop)   to[out=-90, in=90, looseness=0] (IVbot);
\draw[lightgray] (VIItop)  to[out=-90, in=90, looseness=0] (VIIbot);
\draw[lightgray] (Xtop)    to[out=-90, in=90, looseness=0] (Xbot);
\draw[lightgray] (XIIItop) to[out=-90, in=90, looseness=0] (XIIIbot);
\end{scope}

% left matter shells
\draw[-, fill=black!30]
      (IVtop) to[out=-70, in=70, looseness=1.5] (IVbot) to[out=55, in=-55, looseness=1.5] (IVtop) -- cycle;
\draw[-, fill=black!30]
      (VIItop) to[out=-125, in=125, looseness=1.5] (VIIbot) to[out=110, in=-110, looseness=1.5] (VIItop) -- cycle;

% right matter shells
\draw[-, fill=black!30]
      (Xtop) to[out=-70, in=70, looseness=1.5] (Xbot) to[out=55, in=-55, looseness=1.5] (Xtop) -- cycle;
\draw[-, fill=black!30]
      (XIIItop) to[out=-125, in=125, looseness=1.5] (XIIIbot) to[out=110, in=-110, looseness=1.5] (XIIItop) -- cycle;

% remaining lines
\draw[dashed] (Ilefttop) -- (Itop);
\draw (Itop) -- (IVtop);
\draw[decorate,decoration={snake, segment length=1.5mm, amplitude=0.4mm}] (IVtop) -- (VIItop);
\draw (VIItop) -- (Xtop);
\draw[decorate,decoration={snake, segment length=1.5mm, amplitude=0.4mm}] (Xtop) -- (XIIItop);
\draw (XIIItop) -- (XVItop);
\draw[dashed] (XVItop) -- (XVIrt);

\draw[dashed] (Ileftbot) -- (Ibot);
\draw (Ibot) -- (IVbot);
\draw[decorate,decoration={snake, segment length=1.5mm, amplitude=0.4mm}] (IVbot) -- (VIIbot);
\draw (VIIbot) -- (Xbot);
\draw[decorate,decoration={snake, segment length=1.5mm, amplitude=0.4mm}] (Xbot) -- (XIIIbot);
\draw (XIIIbot) -- (XVIbot);
\draw[dashed] (XVIbot) -- (XVIrb);

\draw (Itop) -- (IVbot);
\draw (IVtop) -- (VIIbot);
\draw (VIItop) -- (Xbot);
\draw (Xtop) -- (XIIIbot);
\draw (XIIItop) -- (XVIbot);

\draw (Ibot) -- (IVtop);
\draw (IVbot) -- (VIItop);
\draw (VIIbot) -- (Xtop);
\draw (Xbot) -- (XIIItop);
\draw (XIIIbot) -- (XVItop);

% missing points
\path[fill=white,draw=black] (Itop) circle (0.3ex);
\path[fill=white,draw=black] (IVtop) circle (0.3ex);
\path[fill=white,draw=black] (VIItop) circle (0.3ex);
\path[fill=white,draw=black] (Xtop) circle (0.3ex);
\path[fill=white,draw=black] (XIIItop) circle (0.3ex);
\path[fill=white,draw=black] (XVItop) circle (0.3ex);
\path[fill=white,draw=black] (Ibot) circle (0.3ex);
\path[fill=white,draw=black] (IVbot) circle (0.3ex);
\path[fill=white,draw=black] (VIIbot) circle (0.3ex);
\path[fill=white,draw=black] (Xbot) circle (0.3ex);
\path[fill=white,draw=black] (XIIIbot) circle (0.3ex);
\path[fill=white,draw=black] (XVIbot) circle (0.3ex);

% labels of the matter masses
\path (M_rho_a1)+(-0.1,0.3) coordinate (a1);
\path (M_rho_a1)+(0.3,1.1) coordinate (ma1);
\draw[->] (a1) to[out=135, in=-145, looseness=1.2] (ma1);
\path (M_rho_a2)+(0.1,0.3) coordinate (a2);
\path (M_rho_a2)+(-0.3,1.1) coordinate (ma2);
\draw[->] (a2) to[out=45, in=-35, looseness=1.2] (ma2);
\path (M_rho_b1)+(-0.1,0.3) coordinate (b1);
\path (M_rho_b1)+(0.3,1.1) coordinate (mb1);
\draw[->] (b1) to[out=135, in=-145, looseness=1.2] (mb1);
\path (M_rho_b2)+(0.1,0.3) coordinate (b2);
\path (M_rho_b2)+(-0.3,1.1) coordinate (mb2);
\draw[->] (b2) to[out=45, in=-35, looseness=1.2] (mb2);
\end{tikzpicture}
\caption{Penrose diagram of the maximal $C^2$-extension of a metric constructed as spherically symmetric solution of the Einstein-Vlasov system. The solution coincides with the Schwarzschild-deSitter spacetime in the vacuum regions and the black holes are surrounded by shells of Vlasov matter (gray shaded domains). Notably the black holes do not necessarily have the same mass. The grey lines are surfaces of constant $r$. \label{penrose_periodic}}
\end{center}
\end{figure}
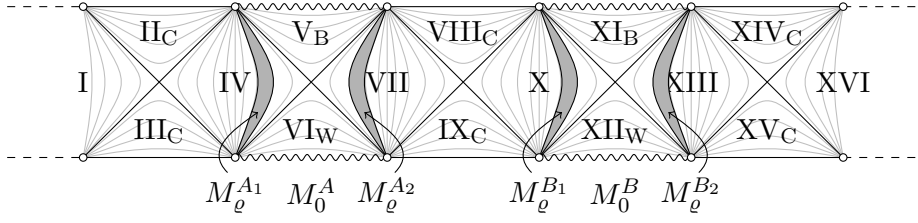

\global\long\def\theenumi{\alph{enumi}}
\begin{rem}\hfill
\begin{enumerate}
%\item The $r=0$ surfaces in the first an second class have the topology $\mathbb R$ and the topology of the $r=\infty$ surfaces is a cylinder, i.e.~$\mathbb R \times S^2$. The $r=0$ surfaces with singular metric in the second class have topology $\mathbb R\times S^2$.
\item The black hole masses in the third class of solutions in the previous theorem can be pairwise different. Only the total mass of black hole and matter shell have to agree pairwise, cf.~condition in (iii) above.
\item Combinations of the classes (ii) and (iii) yield similar metrics on $\mathscr M_3=\mathbb R\times \mathbb R^3$ with a regular center followed by an infinite sequence of black holes.
\item The second class of solutions could also be generalized by adding matter shells around the black holes. The mass parameters then have to be adjusted.
\item When crossing the cosmological horizon or the event horizon of a black or white hole the Killing vector $\partial_t$ changes from being timelike to spacelike. This means that the maximal extended spacetime contains both static and dynamic regions that are alternating. This holds for all constructed classes.
\end{enumerate}
\end{rem}

\begin{proof}
We outline now the construction of the spacetimes given in the previous theorem. For the first two classes of spacetimes we consider solutions of the Einstein-Vlasov system with a regular center. Let $(\mu_\Lambda, \lambda_\Lambda, f_\Lambda)$ be a static solution of the spherically symmetric Einstein-Vlasov system with positive cosmological constant $\Lambda$ defined for $r\in[0,r_C)$ such that the support of the matter quantities is bounded by a radius $0 < R_{0\Lambda} < r_C$. The radius $r_C$ denotes the cosmological horizon. On $[R_{0\Lambda},r_C)$ there is vacuum on hand and the metric is given by the Schwarzschild-deSitter metric (\ref{ssdsmetric}) with the ADM mass $M$ as mass parameter. The ADM mass $M$ is then given by
\begin{equation} \label{admmass}
M = 4\pi\int_0^{R_{0\Lambda}} s^2\varrho_\Lambda(s)\mathrm ds.
\end{equation}
If $9M^2\Lambda < 1$, the polynomial $r^3-\frac{3}{\Lambda}r+\frac{6M}{\Lambda}$ has one negative zero and two positive ones. The largest zero of this polynomial is defined to be the cosmological horizon $r_C$. Moreover, $r_n$ is the negative zero, and $r_{B\Lambda}$ the smaller positive one. In terms of the ADM mass $M$ and the cosmological constant $\Lambda$ these zeros can be calculated explicitly. Note that the Buchdahl inequality for solutions with $\Lambda \neq 0$ \cite{ab} implies $r_{B\Lambda} < R_{0\Lambda}$. \par
\vspace{10pt}
Case (i): Consider Figure \ref{penrose_simple}. This spacetime can be obtained in an analogue way to the standard procedure to compactify the deSitter space as described for example in \cite{GiHa77}. In the following, this procedure is carried out in detail. The metric is given as a non-vacuum solution of the Einstein-Vlasov system for $r\in[0,r_C)$, corresponding to region I in Figure \ref{penrose_simple}, as discussed in Theorem \ref{main-thm}. In this region we have
\begin{equation} \label{metric_mulambda}
\mathrm ds^2 = -e^{2\mu(r)}\mathrm dt^2 + e^{2\lambda(r)}\mathrm dr^2 + r^2\mathrm d\vartheta^2 + r^2\sin^2(\vartheta)\mathrm d\varphi^2.
\end{equation}
In the first step we introduce coordinates $U_\mathrm{I}$, $V_\mathrm{I}$ that transform the region $\mathbb R \times [0,r_C) \times S^2$ into the left triangle (region I) in Figure \ref{penrose_simple}. The coordinates usually used to compactify the vacuum deSitter metric as for example described in \cite{GiHa77} will do. They are given by
\begin{equation} \label{1_coord_1}
U_\mathrm{I} = \sqrt{\frac{r_C-r}{r_C+r}} e^{-\frac{t}{r_C}},\quad V_{\mathrm{I}} = -\sqrt{\frac{r_C-r}{r_C+r}}e^{\frac{t}{r_C}}
\end{equation}
and can be compactified via the transformations $p_\mathrm{I}=\arctan(U_\mathrm{I})$, $q_\mathrm{I}=\arctan(V_\mathrm{I})$. The left part of Figure \ref{cons_ps} shows the transformed region $\mathbb R \times [0,r_C)$ in the $p_\mathrm{I}$, $q_\mathrm{I}$-coordinates. \par
\begin{figure}[ht]
\begin{center}
\begin{tikzpicture}
\pgfdeclarelayer{bg}
\pgfsetlayers{bg,main}

\node (n1)  at (-4.5,0)    {I};
\node (n2)  at (-1,0)      {I};
\node (n3)  at (0,0.65)    {$\mathrm{II}_\mathrm{C}$};
\node (n4)  at (0,-0.65)   {$\mathrm{III}_\mathrm{C}$};
\node (n5)  at (1,0)       {IV};
\node (n6)  at (4.5,0)     {IV};
\node (pi)  at (-5,1.5)    {${p_\mathrm{I}}$};
\node (qi)  at (-2.75,1.5) {$q_\mathrm{I}$};
\node (pc)  at (-1,1.5)    {$p_C$};
\node (qc)  at (1.25,1.5)  {$q_C$};
\node (piv) at (3,1.5)     {$p_\mathrm{IV}$};
\node (qiv) at (5.25,1.5)  {$q_\mathrm{IV}$};
\node (r0l2) at (2,-1.3)   {$\scriptstyle \tilde r = R_{0\Lambda}$};
\node (r0l1) at (-2,-1.3)  {$\scriptstyle \tilde r = R_{0\Lambda}$};

\path % left chart
(n1) +(-0.5,1)      coordinate (n1lefttop)
     +(-0.75,1.25)  coordinate (n1lefttoptop)
     +(-0.5,-1)     coordinate (n1leftbot)
     +(-0.75,-1.25) coordinate (n1leftbotbot)
     +(0.5,0)       coordinate (n1center)
     +(1.5,1)       coordinate (n1righttop)
     +(1.75,1.25)   coordinate (n1righttoptop)
     +(1.5,-1)      coordinate (n1righttbot)
     +(1.75,-1.25)  coordinate (n1rightbotbot)
;

\path % middle chart
(n2) +(-1,0)        coordinate (n2left)
     +(0,1)         coordinate (n2lefttop)
     +(-0.25,1.25)  coordinate (n2lefttoptop)
     +(0,-1)        coordinate (n2leftbot)
     +(-0.25,-1.25) coordinate (n2leftbotbot)
     +(1,0)         coordinate (n2center)
     +(2,1)         coordinate (n2righttop)
     +(2.25,1.25)   coordinate (n2righttoptop)
     +(2,-1)        coordinate (n2rightbot)
     +(2.25,-1.25)  coordinate (n2rightbotbot)
     +(3,0)         coordinate (n2right)
;

\path % right chart
(n6) +(0.5,1)       coordinate (n6righttop)
     +(0.75,1.25)   coordinate (n6righttoptop)
     +(0.5,-1)      coordinate (n6rightbot)
     +(0.75,-1.25)  coordinate (n6rightbotbot)
     +(-0.5,0)      coordinate (n6center)
     +(-1.5,1)      coordinate (n6lefttop)
     +(-1.75,1.25)  coordinate (n6lefttoptop)
     +(-1.5,-1)     coordinate (n6lefttbot)
     +(-1.75,-1.25) coordinate (n6leftbotbot)
;

\begin{scope}[on background layer]
% Lines of constant r
\foreach \a in {2,4,6,8}{
\pgfmathparse{1+(\a-1)/7}
\edef\w{\pgfmathresult}
\draw[lightgray] (VIItop)  to[out =-90+\a*5, in=90-\a*5, looseness=\w] (VIIbot);
\draw[lightgray] (Xtop)    to[out=-90-\a*5, in=90+\a*5, looseness=\w] (Xbot);
\draw[lightgray] (VIItop)  to[out=-\a*5, in=-180+\a*5, looseness=\w] (Xtop);
\draw[lightgray] (VIIbot)  to[out=\a*5, in=180-\a*5, looseness=\w] (Xbot);
}
\foreach \a in {4,6,8}{
\pgfmathparse{1+(\a-1)/7}
\edef\w{\pgfmathresult}
\draw[lightgray] (n1lefttop)    to[out =-90+\a*5, in=90-\a*5, looseness=\w] (n1leftbot);
\draw[lightgray] (n6righttop)  to[out=-90-\a*5, in=90+\a*5, looseness=\w] (n6rightbot);
}
\foreach \a in {2,4}{
\pgfmathparse{1+(\a-1)/7}
\edef\w{\pgfmathresult}
\draw[lightgray] (Xtop)    to[out =-90+\a*5, in=90-\a*5, looseness=\w] (Xbot);
\draw[lightgray] (VIItop)  to[out=-90-\a*5, in=90+\a*5, looseness=\w] (VIIbot);
}
\draw[lightgray] (VIItop)  to[out=-90, in=90, looseness=0] (VIIbot);
\draw[lightgray] (Xtop)    to[out=-90, in=90, looseness=0] (Xbot);
\end{scope}

% matter shells
\draw[-, fill=black!30]
      (n1lefttop) to[out=-70, in=70, looseness=1.286] (n1leftbot) -- node[midway, above, sloped] {${\scriptstyle r=0}$} (n1lefttop) -- cycle;
\draw[-, fill=black!30]
      (n6righttop) to[out=-110, in=110, looseness=1.286] (n6rightbot) -- node[midway, below, sloped] {${\scriptstyle \tilde r=0}$} (n6righttop) -- cycle;

% axes
\draw[->] (n1leftbotbot) -- (n1righttoptop);
\draw[->] (n1rightbotbot) -- (n1lefttoptop);
\draw[->] (n2leftbotbot) -- (n2righttoptop);
\draw[->] (n2rightbotbot) -- (n2lefttoptop);
\draw[->] (n6leftbotbot) -- (n6righttoptop);
\draw[->] (n6rightbotbot) -- (n6lefttoptop);

\draw (n2lefttop) -- node[midway, above, sloped] {${\scriptstyle r=\infty, \mathscr I^+}$} (n2righttop);
\draw (n2leftbot) -- node[midway, below, sloped] {${\scriptstyle \tilde r=\infty, \mathscr I^-}$} (n2rightbot);

% r = 2m -- lines
\draw[dashed] (n1lefttop)  to[out=-80, in=80, looseness=1.143] (n1leftbot);
\draw[dashed] (n6righttop)    to[out=-100, in=100, looseness=1.143] (n6rightbot);
\draw[dashed] (n2lefttop) -- (n2left) -- (n2leftbot);
\draw[dashed] (n2righttop) -- (n2right) -- (n2rightbot);

% shaded areas
\draw[-, pattern=north west lines] (n1lefttop) to[out=-70, in=70, looseness=1.286] (n1leftbot) -- node[midway, below, sloped] {${\scriptstyle r=r_C}$} (n1center) -- node[midway, above, sloped] {${\scriptstyle r=r_C}$} (n1lefttop) -- cycle;
\draw[-, pattern=north west lines] (n2leftbot) to[out=120, in=-120, looseness=1.429] (n2lefttop) -- (n2center) -- (n2leftbot) -- cycle;
\draw[-, pattern=north east lines] (n2righttop) to[out=-60, in=60, looseness=1.429] (n2rightbot) -- (n2center) -- (n2righttop) -- cycle;
\draw[-, pattern=north east lines] (n6righttop) to[out=-110, in=110, looseness=1.286] (n6rightbot) -- node[midway, below, sloped] {${\scriptstyle \tilde r=r_C}$} (n6center) -- node[midway, above, sloped] {${\scriptstyle \tilde r=r_C}$} (n6righttop) -- cycle;

% labels
\path (r0l1)+(0.1,0.2) coordinate (r0l11);
\path (r0l1)+(0.57,1.3) coordinate (r0l12);
\draw[->] (r0l11) to[out=90, in=-135, looseness=1.2] (r0l12);
\path (r0l2)+(-0.1,0.2) coordinate (r0l21);
\path (r0l2)+(-0.57,1.3) coordinate (r0l22);
\draw[->] (r0l21) to[out=90, in=-45, looseness=1.2] (r0l22);

\end{tikzpicture}
\caption{Construction of the spacetime shown in Figure \ref{penrose_simple}. We use three coordinate charts to compactify the spacetime. Regions that are shaded in the same orientation are covered by two of the coordinate charts simultaneously, thus there coordinates can be changed. The gray areas are matter regions and the dashed lines correspond to $r=r_{B\Lambda}$. We distinguish between $r$ and $\tilde r$ to emphasize that there are different spacetime regions that cannot be covered by a single chart $(t,r,\vartheta,\varphi)$. All coordinates $p$ and $q$ take values in $\left[-\frac \pi 2,\frac \pi 2\right]$. \label{cons_ps}}
\end{center}
\end{figure}
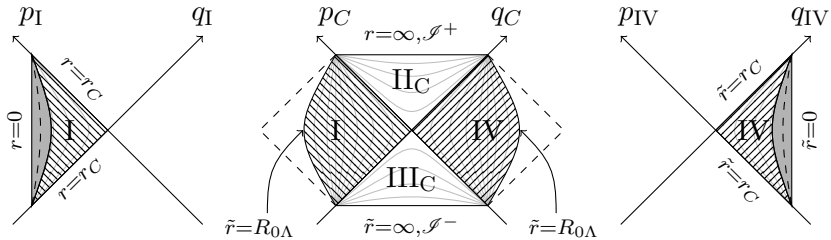
The support of the matter (i.e.~the matter distribution $f$) ends at a radius $R_{0\Lambda}$. For $r\geq R_{0\Lambda}$ the metric is merely given by the Schwarzschild-deSitter metric
\begin{equation} \label{ssdsmetric}
\mathrm ds^2 =-\left(1-\frac{r^2\Lambda}{3}-\frac{2M}{r}\right)\mathrm dt^2 + \frac{\mathrm dr^2}{1-\frac{r^2\Lambda}{3}-\frac{2M}{r}}+r^2\mathrm d\Omega^2,\quad R_{0\Lambda}\leq r < r_C.
\end{equation}
At $r=r_C$ there is a coordinate singularity of the metric that we want to pass. For this purpose we express the metric in other coordinates that do not have a singularity at $r=r_C$ being defined on the region where $r\in[R_{0\Lambda},r_c)$ (region I in the middle part of Figure \ref{cons_ps}). These coordinates are given by
\begin{equation} \label{1_coord_2}
U_C=\sqrt{\frac{(r_C-r)(r-r_n)^{\gamma-1}}{(r-r_{B\Lambda})^\gamma}}e^{-\frac{t}{2\delta_C}} >0, \quad V_C = -\sqrt{\frac{(r_C-r)(r-r_n)^{\gamma-1}}{(r-r_{B\Lambda})^\gamma}}e^{\frac{t}{2\delta_C}}<0,
\end{equation}
where $\delta_C = \frac{r_C}{\Lambda r_C^2-1} > 0$ and $\gamma = \frac{r_{B\Lambda}}{(1-\Lambda r_{B\Lambda}^2) \delta_C},$ $0<\gamma <1$\footnote{The signs of these expressions can be checked with the equality $1-\frac{r_C^2\Lambda}{3}-\frac{2M}{r_C}=0$}. They are used in the standard compactification procedure of the Schwarzschild-deSitter metric. For details, see \cite{bf86} or \cite{christina}. In the new coordinates the line element of the Schwarzschild-deSitter metric \eqref{ssdsmetric} reads
\begin{equation} \label{line_1_c_coords}
\mathrm ds^2 = -\frac{4\Lambda \delta_C^2}{3r}(r-r_n)^{2-\gamma}(r-r_{B\Lambda})^{1+\gamma}\mathrm dU_C \mathrm dV_C + r^2\mathrm d\vartheta^2 + r^2\sin^2(\vartheta)\mathrm d\varphi^2,\quad r\geq R_{0\Lambda}.
\end{equation}
Note that here $r$ is seen as a function of $U_C$ and $V_C$. The coordinates only take values in $\{(u,v)\in\mathbb R^2\,|\,u>0,v<0\}$. We extend them to $\mathbb R^2$. This extension gets beyond $r_C$. Again, the spacetime region covered by the coordinates $U_C$ and $V_C$ can be compactified using the transformation $p_C=\arctan(U_C)$, $q_C=\arctan(V_C)$. The middle part of Figure \ref{cons_ps} shows the region covered by $U_C$ and $V_C$, each taking values in $\mathbb R$, in the $p_C$, $q_C$-coordinates. The line element (\ref{line_1_c_coords}) can be extended to the whole area covered by $U_C$ and $V_C$ in an analytic way. In the region where $r\in[R_{0\Lambda},r_C)$ the coordinate charts (\ref{1_coord_1}) and (\ref{1_coord_2}) overlap and one can change coordinates (the shaded areas in the left and middle part of Figure \ref{cons_ps}). The transformation law is given by
\begin{equation}
\begin{aligned}
U_C(U_\mathrm{I}) &= \sqrt{\frac{(r_C+r)(r-r_n)^{\gamma-1}}{(r-r_{B\Lambda})^\gamma}}e^{\frac{3-2\Lambda r_C^2}{r_C}}U_I,\\
V_C(V_\mathrm{I}) &= \sqrt{\frac{(r_C+r)(r-r_n)^{\gamma-1}}{(r-r_{B\Lambda})^\gamma}}e^{-\frac{3-2\Lambda r_C^2}{r_C}}V_I.
\end{aligned}
\end{equation}
Region IV in Figure \ref{penrose_simple} corresponds to a second universe that also can be equipped with Schwarzschild coordinates $(\tilde t, \tilde r)$. We distinguish between $r$ and $\tilde r$ to emphasize that the charts $(t,r)$ and $(\tilde t, \tilde r)$ cover different regions of the spacetime. Geometrically these regions look equal. This will be different for the second class of spacetimes (\ref{class2}). In the region $\tilde r \in [R_{0\Lambda},r_C)$ (region IV in the middle part of Figure \ref{cons_ps}), in terms of the $\tilde t$, $\tilde r$-coordinates $U_C$ and $V_C$ are given by
\begin{equation} \label{c_coord_1}
U_C=-\sqrt{\frac{(r_C-\tilde r)(\tilde r-r_-)^{\gamma-1}}{(\tilde r-r_{B\Lambda})^\gamma}}e^{-\frac{\tilde t}{2\delta_C}} <0, \quad V_C = \sqrt{\frac{(r_C-\tilde r)(\tilde r-r_-)^{\gamma-1}}{(\tilde r-r_{B\Lambda})^\gamma}}e^{\frac{\tilde t}{2\delta_C}}>0.
\end{equation}
To get a compactification of the whole region IV, including $\tilde r<r_{B}$, we introduce coordinates similar to (\ref{1_coord_1}), namely
\begin{equation} \label{1_coord_3}
U_\mathrm{IV} = -\sqrt{\frac{r_C-\tilde r}{r_C+\tilde r}} e^{-\frac{\tilde t}{r_C}},\quad V_{\mathrm{IV}} = \sqrt{\frac{r_C-\tilde r}{r_C+\tilde r}}e^{\frac{\tilde t}{r_C}}
\end{equation}
covering the region characterized by $\tilde r \in[0,r_C)$. This region can again be compactified via $p=\arctan(U)$, $q=\arctan(V)$. This yields the right part of Figure \ref{cons_ps}. For $\tilde r \in [R_{0\Lambda},r_C)$ the coordinates can be changed using an analogue law as (\ref{c_coord_1}). On the spacetime region represented by the middle part of Figure \ref{cons_ps} the line element can be expressed by (\ref{line_1_c_coords}). Since both, in region I and IV the metric can be brought into the form (\ref{metric_mulambda}) via coordinate transformations also the energy densities are identical in these regions. This of course implies that in both regions the mass parameter is equal. \par
\vspace{10pt}
Case (ii): Now we come to the spacetimes characterized by Figure \ref{penrose_diagram}. For the construction of a $C^2$-extension of the metric (\ref{metric_mulambda}) at least five coordinate charts are necessary. Figure \ref{construction} illustrates this construction.
\begin{figure}[ht]
\begin{center}
\begin{tikzpicture}
\pgfdeclarelayer{bg}
\pgfsetlayers{bg,main}

\node (n1)  at (-6.5,0)    {I};
\node (n2)  at (-3,0)      {I};
\node (n3)  at (-2,0.65)   {$\mathrm{II}_\mathrm{C}$};
\node (n4)  at (-2,-0.65)  {$\mathrm{III}_\mathrm{C}$};
\node (n5)  at (-1,0)      {IV};
\node (n6)  at (2,0)       {IV};
\node (n7)  at (3,0.65)    {$\mathrm{V}_\mathrm{BH}$};
\node (n8)  at (3,-0.65)   {$\mathrm{VI}_\mathrm{WH}$};
\node (n9)  at (4,0)       {VII};
\node (pi)  at (-7,1.5)    {${p_\mathrm{I}}$};
\node (qi)  at (-4.75,1.5) {$q_\mathrm{I}$};
\node (pc)  at (-3,1.5)    {$p_C$};
\node (qc)  at (-0.75,1.5)  {$q_C$};
\node (piv) at (2,1.5)     {$p_B$};
\node (qiv) at (4.25,1.5)  {$q_B$};

\path % left chart
(n1) +(-0.5,1)      coordinate (n1lefttop)
     +(-0.75,1.25)  coordinate (n1lefttoptop)
     +(-0.5,-1)     coordinate (n1leftbot)
     +(-0.75,-1.25) coordinate (n1leftbotbot)
     +(0.5,0)       coordinate (n1center)
     +(1.5,1)       coordinate (n1righttop)
     +(1.75,1.25)   coordinate (n1righttoptop)
     +(1.5,-1)      coordinate (n1righttbot)
     +(1.75,-1.25)  coordinate (n1rightbotbot)
;

\path % middle chart
(n2) +(-1,0)        coordinate (n2left)
     +(0,1)         coordinate (n2lefttop)
     +(-0.25,1.25)  coordinate (n2lefttoptop)
     +(0,-1)        coordinate (n2leftbot)
     +(-0.25,-1.25) coordinate (n2leftbotbot)
     +(1,0)         coordinate (n2center)
     +(2,1)         coordinate (n2righttop)
     +(2.25,1.25)   coordinate (n2righttoptop)
     +(2,-1)        coordinate (n2rightbot)
     +(2.25,-1.25)  coordinate (n2rightbotbot)
     +(3,0)         coordinate (n2right)
;

\path % right chart
(n6) +(-1,0)        coordinate (n6left)
     +(0,1)         coordinate (n6lefttop)
     +(-0.25,1.25)  coordinate (n6lefttoptop)
     +(0,-1)        coordinate (n6leftbot)
     +(-0.25,-1.25) coordinate (n6leftbotbot)
     +(1,0)         coordinate (n6center)
     +(2,1)         coordinate (n6righttop)
     +(2.25,1.25)   coordinate (n6righttoptop)
     +(2,-1)        coordinate (n6rightbot)
     +(2.25,-1.25)  coordinate (n6rightbotbot)
     +(3,0)         coordinate (n6right)
;

\begin{scope}[on background layer]
% Lines of constant r
\foreach \a in {2,4,6,8}{
\pgfmathparse{1+(\a-1)/7}
\edef\w{\pgfmathresult}
\draw[lightgray] (n2lefttop)  to[out =-90+\a*5, in=90-\a*5, looseness=\w] (n2leftbot);
\draw[lightgray] (n2righttop)  to[out =-90+\a*5, in=90-\a*5, looseness=\w] (n2rightbot);
\draw[lightgray] (n6lefttop)  to[out =-90+\a*5, in=90-\a*5, looseness=\w] (n6leftbot);
\draw[lightgray] (n6righttop)  to[out =-90+\a*5, in=90-\a*5, looseness=\w] (n6rightbot);
\draw[lightgray] (n2righttop)    to[out=-90-\a*5, in=90+\a*5, looseness=\w] (n2rightbot);
\draw[lightgray] (n6lefttop)    to[out=-90-\a*5, in=90+\a*5, looseness=\w] (n6leftbot);
\draw[lightgray] (n6righttop)    to[out=-90-\a*5, in=90+\a*5, looseness=\w] (n6rightbot);
\draw[lightgray] (n2lefttop)  to[out=-\a*5, in=-180+\a*5, looseness=\w] (n2righttop);
\draw[lightgray] (n6lefttop)  to[out=-\a*5, in=-180+\a*5, looseness=\w] (n6righttop);
\draw[lightgray] (n2leftbot)  to[out=\a*5, in=180-\a*5, looseness=\w] (n2rightbot);
\draw[lightgray] (n6leftbot)  to[out=\a*5, in=180-\a*5, looseness=\w] (n6rightbot);
}
\foreach \a in {4,6,8}{
\pgfmathparse{1+(\a-1)/7}
\edef\w{\pgfmathresult}
\draw[lightgray] (n1lefttop)  to[out=-90+\a*5, in=90-\a*5, looseness=\w] (n1leftbot);
}
\foreach \a in {2,4}{
\pgfmathparse{1+(\a-1)/7}
\edef\w{\pgfmathresult}
\draw[lightgray] (n2lefttop)  to[out=-90-\a*5, in=90+\a*5, looseness=\w] (n2leftbot);
}
\draw[lightgray] (n2lefttop)  to[out=-90, in=90, looseness=0] (n2leftbot);
\draw[lightgray] (n2righttop)    to[out=-90, in=90, looseness=0] (n2rightbot);
\draw[lightgray] (n6lefttop)  to[out=-90, in=90, looseness=0] (n6leftbot);
\draw[lightgray] (n6righttop)    to[out=-90, in=90, looseness=0] (n6rightbot);
\end{scope}

% matter shell
\draw[-, fill=black!30]
      (n1lefttop) to[out=-70, in=70, looseness=1.286] (n1leftbot) -- node[midway, above, sloped] {${\scriptstyle r=0}$} (n1lefttop) -- cycle;

% axes
\draw[->] (n1leftbotbot) -- (n1righttoptop);
\draw[->] (n1rightbotbot) -- (n1lefttoptop);
\draw[->] (n2leftbotbot) -- (n2righttoptop);
\draw[->] (n2rightbotbot) -- (n2lefttoptop);
\draw[->] (n6leftbotbot) -- (n6righttoptop);
\draw[->] (n6rightbotbot) -- (n6lefttoptop);

\draw (n2lefttop) -- node[midway, above, sloped] {${\scriptstyle r=\infty, \mathscr I^+}$} (n2righttop);
\draw (n2leftbot) -- node[midway, below, sloped] {${\scriptstyle r=\infty, \mathscr I^-}$} (n2rightbot);
\draw[decorate,decoration={snake, segment length=1.5mm, amplitude=0.4mm}] (n6lefttop) -- (n6righttop)
      node[midway, above, inner sep=2mm] {${\scriptstyle r=0}$};
\draw[decorate,decoration={snake, segment length=1.5mm, amplitude=0.4mm}] (n6leftbot) -- (n6rightbot)
      node[midway, below, inner sep=2mm] {${\scriptstyle r=0}$};

% r = 2m -- lines
\draw[dashed] (n1lefttop)  to[out=-80, in=80, looseness=1.143] (n1leftbot);
\draw[dashed] (n2lefttop) -- (n2left) -- (n2leftbot);

% shaded areas
\draw[-, pattern=north west lines] (n1lefttop) to[out=-70, in=70, looseness=1.286] (n1leftbot) -- node[midway, below, sloped] {${\scriptstyle r=r_C}$} (n1center) -- node[midway, above, sloped] {${\scriptstyle r=r_C}$} (n1lefttop) -- cycle;
\draw[-, pattern=north west lines] (n2leftbot) to[out=120, in=-120, looseness=1.429] node[midway, above, sloped] {${\scriptstyle r=R_{0\Lambda}}$} (n2lefttop) -- (n2center) -- (n2leftbot) -- cycle;
\draw[-, pattern=north east lines] (n2righttop) -- node[midway, above, sloped] {${\scriptstyle r=r_{B\Lambda}}$} (n2right) -- node[midway, below, sloped] {${\scriptstyle r=r_{B\Lambda}}$} (n2rightbot) -- (n2center) -- (n2righttop) -- cycle;
\draw[-, pattern=north east lines] (n6lefttop) -- (n6center) -- (n6leftbot) -- node[midway, below, sloped] {${\scriptstyle r=r_C}$} (n6left) --node[midway, above, sloped] {${\scriptstyle r=r_C}$} (n6lefttop) -- cycle;
\draw[-, pattern=horizontal lines] (n6righttop) -- node[midway, above, sloped] {${\scriptstyle r=r_C}$} (n6right) -- node[midway, below, sloped] {${\scriptstyle r=r_C}$} (n6rightbot) -- (n6center) -- (n6righttop) -- cycle;

% dots
\path[fill=black,draw=black] (6,0) circle (0.2ex);
\path[fill=black,draw=black] (6.2,0) circle (0.2ex);
\path[fill=black,draw=black] (6.4,0) circle (0.2ex);
\end{tikzpicture}
\caption{Construction of the spacetime shown in Figure \ref{penrose_diagram}. On regions that are shaded in equal directions two coordinates are defined and one can change between them. All coordinates $p$, $q$ take values in $\left[-\frac \pi 2,\frac \pi 2\right]$. \label{construction}}
\end{center}
\end{figure}
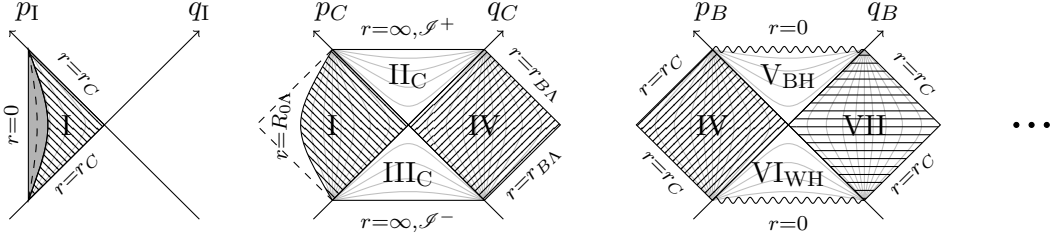
Again we begin with the region $r\in[0,r_C)$ where the metric is given by (\ref{metric_mulambda}). In the same way as described above one expresses the line element in other coordinates $p_C$, $q_C$ that avoid the singularity at $r=r_C$ and cover the region $r_{0\Lambda} < r < r_C$. The line element as given by (\ref{line_1_c_coords}) can be analytically\footnote{In matter regions the regularity of the metric is $C^2$ as provided by Theorem \ref{main-thm}, in vacuum regions the metric is analytic.} extended onto the regions I -- IV in Figure \ref{construction}. From now on the procedure differs from the one above. Regions I and IV are not supposed to be geometrically identical but region IV shall be a vacuum region thus the metric will be given by Schwarzschild-deSitter everywhere. Certainly, the line element (\ref{line_1_c_coords}) of the Schwarzschild-deSitter metric  being given in terms of the coordinates $U_C$, $V_C$ now shows a singularity at $r=r_{B\Lambda}$\footnote{By abuse of notation we use $r$ for the radius coordinate in every region of the spacetime $\mathscr M_1$.}. This coordinate singularity can be overcome by virtue of the coordinates
\begin{equation} \label{2_coord_bh}
U_B = \sqrt{\frac{(r-r_{B\Lambda})(r-r_n)^{\beta-1}}{(r_C-r)^\beta}}e^{\frac{t}{2\delta_B}}, \quad V_B = -\sqrt{\frac{(r-r_{B\Lambda})(r-r_n)^{\beta-1}}{(r_C-r)^\beta}}e^{-\frac{t}{2\delta_B}},
\end{equation}
where $\delta_B = \frac{r_{B\Lambda}}{1-\Lambda r_{B\Lambda}^2} > 0$ and $\beta = \frac{r_C}{(\Lambda r_C^2-1)\delta_B} > 1$. The coordinates are defined on the middle part of Figure \ref{construction}. This is part of the standard compactification procedure of the Schwarzschild-deSitter metric, cf.~\cite{bf86} or \cite{christina}. Alternating the coordinate charts $(U_C, V_C)$ and $(U_B, V_B)$ this procedure can be continued an arbitrary amount of times extending the spacetime to additional black hole and cosmological regions. This periodic extension stops if for $r<r_C$ the metric is not given by a vacuum solution of the Einstein equations but again by the solution (\ref{metric_mulambda}) of the Einstein-Vlasov system. There is no coordinate singularity at $r=r_{B\Lambda}$ and also a regular center at $r=0$. So a regular expression of the line element by the coordinates (\ref{1_coord_1}) is possible again, leading to region X in Figure \ref{penrose_diagram}. This region now is geometrically identical to region I in Figure \ref{penrose_diagram} (and also in Figure \ref{construction}). In the extension procedure just described the expressions for the coordinates (\ref{c_coord_1}) and (\ref{2_coord_bh}) used to pass the coordinate singularities at $r=r_{B\Lambda}$ and $r=r_C$ in the vacuum regions of the spacetime $\mathscr M_1$ depend on $\Lambda$ and $M$. So the identification of corresponding regions in the different coordinate charts, e.g.~I or IV in Figure \ref{construction}, is only possible if the parameters $\Lambda$ and $M$ are equal in all regions of $\mathscr M_1$. In terms of the notation of Figure \ref{penrose_diagram} this implies $M_1=M_2$. \par
\vspace{10pt}
Case (iii): A maximal extension of a solution to the Einstein-Vlasov system on the manifold $\mathscr M_2$ as characterized by Figure \ref{penrose_periodic}, i.e.~spacetimes in class (\ref{class3}), can be obtained in a similar way. Starting point is the region $r_{B \Lambda} < r < r_C$. On this interval the existence of a unique solution to a given ansatz for $f$ is established by Theorem \ref{theo_la_pos_bh}. The solution on hand can be understood as a Schwarzschild-deSitter spacetime with an immersed shell of Vlasov matter supported on an interval $(r_{+\Lambda},R_{0\Lambda})$. Two mass quantities are important. On the one hand one has the mass parameter $M_0$ of the black hole at the center. On the other hand $M$ that is defined to be
\begin{equation} \label{eq_p3_quasilocal}
M=M_0 + M_\varrho,\quad M_\varrho = 4\pi \int_{r_{+\Lambda}}^{R_{0\Lambda}} s^2\varrho_\Lambda(s)\mathrm ds.
\end{equation}
This quantity represents the sum of the mass of the black hole and the shell of Vlasov matter. As constructed in Theorem \ref{theo_la_pos_bh}, for $r_{B\Lambda} < r \leq r_{+\Lambda}$ the metric is given by a shifted Schwarschild-deSitter metric
\begin{equation}
\mathrm ds^2 = -C\left(1-\frac{r^2\Lambda}{3}-\frac{2M_0}{r}\right)\mathrm dt^2 + \frac{\mathrm dr^2}{C\left(1-\frac{r^2\Lambda}{3}-\frac{2M_0}{r}\right)} + r^2\mathrm d\Omega^2,\quad r_{B\Lambda} < r \leq r_{+\Lambda}
\end{equation}
with the mass $M_0$ of the black hole as mass parameter and the shift $C>0$. For $R_{0\Lambda} \leq r < r_C$ the metric is given by the Schwarzschild-deSitter metric (\ref{ssdsmetric}) with mass parameter $M$.\par
The two critical horizons, $r_{B\Lambda}$ and $r_C$ can be given explicitly as zeros of the expression $1-\frac{r^2\Lambda}{3}-\frac{2m(r)}{r}$. But it is important to note that the mass parameter $m(r)$ does not stay constant throughout the whole interval $(r_{B\Lambda},R_{0\Lambda})$. The black hole horizon $r_{B\Lambda}$ is characterized by $M_0$ and the cosmological horizon $r_C$ by $M$. This has to be kept in mind when choosing coordinates to construct an extension of the metric on $\mathscr M_2$ as illustrated in Figure \ref{constr_p}.
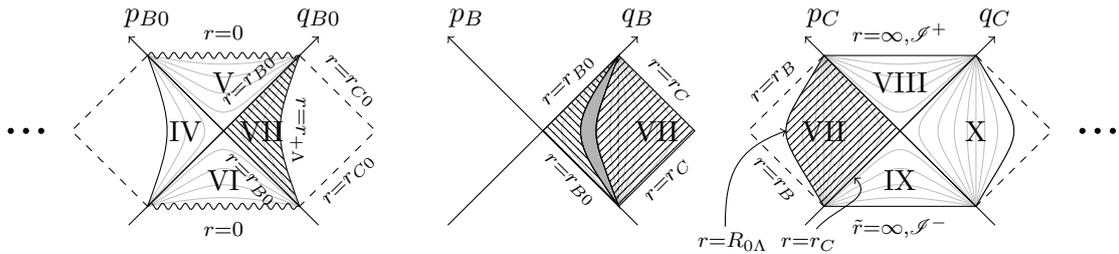
\begin{figure}[ht]
\begin{center}
\begin{tikzpicture}
\pgfdeclarelayer{bg}
\pgfsetlayers{bg,main}

\node (n1)  at (-4.7,0)    {IV};
\node (n2)  at (-4.2,0.65)   {V};
\node (n3)  at (-4.2,-0.65)  {VI};
\node (n4)  at (-3.7,0)    {VII};
\node (n5)  at (1.5,0)     {VII};
\node (n6)  at (3.7,0)       {VII};
\node (n7)  at (4.7,0.65)    {VIII};
\node (n8)  at (4.7,-0.65)   {IX};
\node (n9)  at (5.7,0)       {X};
\node (pbn) at (-5.2,1.5)  {$p_{B0}$};
\node (qbn) at (-2.95,1.5) {$q_{B0}$};
\node (pb)  at (-1,1.5)    {$p_B$};
\node (qb)  at (1.25,1.5)  {$q_B$};
\node (pc)  at (3.7,1.5)   {$p_C$};
\node (qc)  at (5.95,1.5)  {$q_C$};
\node (r0l) at (2.5,-1.47) {${\scriptstyle r=R_{0\Lambda}}$};
\node (rc2) at (3.5,-1.5)  {${\scriptstyle r=r_C}$};

\path % left chart
(n1) +(-0.5,1)      coordinate (n1lefttop)
     +(-0.75,1.25)  coordinate (n1lefttoptop)
     +(-0.5,-1)     coordinate (n1leftbot)
     +(-0.75,-1.25) coordinate (n1leftbotbot)
     +(-1.5,0)      coordinate (n1left)
     +(2.5,0)       coordinate (n1right)
     +(0.5,0)       coordinate (n1center)
     +(1.5,1)       coordinate (n1righttop)
     +(1.75,1.25)   coordinate (n1righttoptop)
     +(1.5,-1)      coordinate (n1rightbot)
     +(1.75,-1.25)  coordinate (n1rightbotbot)
;

\path % middle chart
(n5) +(-2.5,1)      coordinate (n5lefttop)
     +(-2.75,1.25)  coordinate (n5lefttoptop)
     +(-2.5,-1)     coordinate (n5leftbot)
     +(-2.75,-1.25) coordinate (n5leftbotbot)
     +(-1.5,0)      coordinate (n5center)
     +(-0.5,1)       coordinate (n5righttop)
     +(-0.25,1.25)   coordinate (n5righttoptop)
     +(-0.5,-1)      coordinate (n5rightbot)
     +(-0.25,-1.25)  coordinate (n5rightbotbot)
     +(0.5,0)       coordinate (n5right)
;

\path % right chart
(n6) +(-1,0)        coordinate (n6left)
     +(-0,1)        coordinate (n6lefttop)
     +(-0.25,1.25)  coordinate (n6lefttoptop)
     +(-0,-1)       coordinate (n6leftbot)
     +(-0.25,-1.25) coordinate (n6leftbotbot)
     +(1,0)         coordinate (n6center)
     +(2,1)         coordinate (n6righttop)
     +(2.25,1.25)   coordinate (n6righttoptop)
     +(2,-1)        coordinate (n6rightbot)
     +(2.25,-1.25)  coordinate (n6rightbotbot)
     +(3,0)         coordinate (n6right)
;

\begin{scope}[on background layer]
% Lines of constant r
\foreach \a in {2,4,6,8}{
\pgfmathparse{1+(\a-1)/7}
\edef\w{\pgfmathresult}
\draw[lightgray] (n5righttop)  to[out =-90+\a*5, in=90-\a*5, looseness=\w] (n5rightbot);
\draw[lightgray] (n6lefttop)  to[out =-90+\a*5, in=90-\a*5, looseness=\w] (n6leftbot);
\draw[lightgray] (n5righttop)    to[out=-90-\a*5, in=90+\a*5, looseness=\w] (n5rightbot);
\draw[lightgray] (n6righttop)    to[out=-90-\a*5, in=90+\a*5, looseness=\w] (n6rightbot);
\draw[lightgray] (n1lefttop)  to[out=-\a*5, in=-180+\a*5, looseness=\w] (n1righttop);
\draw[lightgray] (n6lefttop)  to[out=-\a*5, in=-180+\a*5, looseness=\w] (n6righttop);
\draw[lightgray] (n1leftbot)  to[out=\a*5, in=180-\a*5, looseness=\w] (n1rightbot);
\draw[lightgray] (n6leftbot)  to[out=\a*5, in=180-\a*5, looseness=\w] (n6rightbot);
}
\foreach \a in {2,4,6}{
\pgfmathparse{1+(\a-1)/7}
\edef\w{\pgfmathresult}
\draw[lightgray] (n6righttop)  to[out =-90+\a*5, in=90-\a*5, looseness=\w] (n6rightbot);
\draw[lightgray] (n6lefttop)    to[out=-90-\a*5, in=90+\a*5, looseness=\w] (n6leftbot);
}
\foreach \a in {6,8}{
\pgfmathparse{1+(\a-1)/7}
\edef\w{\pgfmathresult}
\draw[lightgray] (n1lefttop)  to[out =-90+\a*5, in=90-\a*5, looseness=\w] (n1leftbot);
\draw[lightgray] (n1righttop)    to[out=-90-\a*5, in=90+\a*5, looseness=\w] (n1rightbot);
}
\draw[lightgray] (n5righttop)  to[out=-90, in=90, looseness=0] (n5rightbot);
\draw[lightgray] (n6righttop)    to[out=-90, in=90, looseness=0] (n6rightbot);
\draw[lightgray] (n6lefttop)  to[out=-90, in=90, looseness=0] (n6leftbot);
\end{scope}

% matter shell
\draw[-, fill=black!30]
      (n5righttop) to[out=-125, in=125, looseness=1.5] (n5rightbot) to[out=110, in=-110, looseness=1.5] (n5righttop) -- cycle;

% axes
\draw[->] (n1leftbotbot) -- (n1righttoptop);
\draw[->] (n1rightbotbot) -- (n1lefttoptop);
\draw[->] (n5leftbotbot) -- (n5righttoptop);
\draw[->] (n5rightbotbot) -- (n5lefttoptop);
\draw[->] (n6leftbotbot) -- (n6righttoptop);
\draw[->] (n6rightbotbot) -- (n6lefttoptop);

\draw[decorate,decoration={snake, segment length=1.5mm, amplitude=0.4mm}] (n1lefttop) -- (n1righttop)
      node[midway, above, inner sep=2mm] {${\scriptstyle r=0}$};
\draw[decorate,decoration={snake, segment length=1.5mm, amplitude=0.4mm}] (n1leftbot) -- (n1rightbot)
      node[midway, below, inner sep=2mm] {${\scriptstyle r=0}$};
\draw (n6lefttop) -- node[midway, above, sloped] {${\scriptstyle r=\infty, \mathscr I^+}$} (n6righttop);
\draw (n6leftbot) -- node[midway, below, sloped] {${\scriptstyle \tilde r=\infty, \mathscr I^-}$} (n6rightbot);

% dashed lines
\draw[dashed] (n1lefttop) -- (n1left) -- (n1leftbot);
\draw[dashed] (n1righttop) -- node[midway, above, sloped] {${\scriptstyle r=r_{C0}}$} (n1right) -- node[midway, below, sloped] {${\scriptstyle r=r_{C0}}$} (n1rightbot);
\draw[dashed] (n6lefttop) -- node[midway, above, sloped] {${\scriptstyle r=r_{B}}$} (n6left) -- node[midway, below, sloped] {${\scriptstyle r=r_{B}}$} (n6leftbot);
\draw[dashed] (n6righttop) -- (n6right) -- (n6rightbot);

% shaded areas
\draw[-, pattern=north west lines] (n1righttop) to[out=-110, in=110, looseness=1.286] node[midway, above, sloped] {${\scriptstyle r=r_{+\Lambda}}$} (n1rightbot) -- node[midway, below, sloped] {${\scriptstyle r=r_{B0}}$} (n1center) -- node[midway, above, sloped] {${\scriptstyle r=r_{B0}}$} (n1righttop) -- cycle;
\draw (n1lefttop) to[out=-70, in=70, looseness=1.286] (n1leftbot) -- (n1center) -- (n1lefttop) -- cycle;
\draw[-, pattern=north west lines] (n5center) -- node[midway, above, sloped] {${\scriptstyle r=r_{B0}}$} (n5righttop) to[out=-125, in=125, looseness=1.5] (n5rightbot)  -- node[midway, below, sloped] {${\scriptstyle r=r_{B0}}$} (n5center) -- cycle;
\draw[-, pattern=north east lines] (n5righttop) to[out=-110, in=110, looseness=1.5] (n5rightbot) -- node[midway, below, sloped] {${\scriptstyle r=r_{C}}$} (n5right) -- node[midway, above, sloped] {${\scriptstyle r=r_{C}}$} (n5righttop) -- cycle;
\draw[-, pattern=north east lines] (n6lefttop) -- (n6center) -- (n6leftbot) to[out=120, in=-120, looseness=1.71] (n6lefttop) -- cycle;
\draw (n6righttop) -- (n6center) -- (n6rightbot) to[out=60, in=-60, looseness=1.71] (n6righttop) -- cycle;

% dots
\path[fill=black,draw=black] (-7,0) circle (0.2ex);
\path[fill=black,draw=black] (-6.8,0) circle (0.2ex);
\path[fill=black,draw=black] (-6.6,0) circle (0.2ex);
\path[fill=black,draw=black] (7.1,0) circle (0.2ex);
\path[fill=black,draw=black] (7.3,0) circle (0.2ex);
\path[fill=black,draw=black] (7.5,0) circle (0.2ex);

% labels
\path (rc2)+(0.1,0.2) coordinate (rc21);
\path (rc2)+(0.6,0.88) coordinate (rc22);
\draw[->] (rc21) to[out=49, in=-45, looseness=1.2] (rc22);
\path (r0l)+(-0.1,0.2) coordinate (r0l1);
\path (r0l)+(0.7,1.47) coordinate (r0l2);
\draw[->] (r0l1) to[out=90, in=-160, looseness=1.2] (r0l2);

\end{tikzpicture}
\caption{Construction of the spacetime shown in Figure \ref{penrose_periodic}. The middle part shows a Schwarzschild-deSitter spacetime with an immersed matter shell for $r_{B\Lambda}=r_{B0} < r < r_C$. The left and the right part show the adjacent vacuum region containing several coordinate singularities. On regions that are shaded in equal directions two coordinates are defined and one can change between them. All coordinates $p$, $q$ take values in $\left[-\frac \pi 2,\frac \pi 2\right]$.\label{constr_p}}
\end{center}
\end{figure}
We distinguish between the zeros of $1-\frac{r^2\Lambda}{3}-\frac{2m(r)}{r}$ when $m(r)\equiv M_0$ and $m(r) \equiv M$ and call them $r_{B0}$, $r_{C0}$ or $r_B$, $r_C$, respectively. Note that $r_{B0}=r_{B\Lambda}$. Consider the metric on the region $r_{B0} < r < r_C$ being part of region VII in Figure \ref{penrose_periodic} or the middle part of Figure \ref{constr_p}. The metric shall be extended to the left (regions IV, $\mathrm{V}_\mathrm{B}$, $\mathrm{VI}_\mathrm{W}$) and to the right (regions $\mathrm{VIII}_\mathrm{C}$, $\mathrm{IX}_\mathrm{C}$, X) as a vacuum solution until the next matter shell appears. So the coordinate transformations have to be chosen with respect to the radii $r_{B}$ and $r_C$ belonging to the current mass parameter in the respective spacetime region. Three coordinate charts are needed to extend the metric beyond the black hole and the cosmological horizon. First we compactify the region $r_{B\Lambda}=r_{B0} < r < r_C$ using the coordinates
\begin{equation}
U_B = \sqrt{\frac{(r-r_{B0})(r-r_n)^{\beta-1}}{(r_C-r)^\beta}}e^{\frac{t}{2\delta_{B0}}},\quad V_B=-\sqrt{\frac{(r-r_{B0})(r-r_n)^{\beta-1}}{(r_C-r)^\beta}}e^{-\frac{t}{2\delta_{B0}}}.
\end{equation}
where $\delta_{B0} = \frac{r_{B0}}{1-\Lambda r_{B0}^2} > 0$ and $\beta = \frac{r_C}{(\Lambda r_C^2-1)\delta_{B0}} > 1$. These coordinates give rise to $p_B=\arctan(U_B)$ and $q_B=\arctan(V_B)$. This region is depicted in the middle part of Figure \ref{constr_p}.
The spacetimes characterized by Figure \ref{penrose_periodic} show two types of connected vacuum regions. The first type is characterized by $r\leq r_{+\Lambda}$ (inside the matter shell) and the second one by $r\geq R_{0\Lambda}$ (beyond the matter shell).
To extend the metric to the region inside the matter shell (and the black hole) one uses the coordinates
\begin{equation} \label{coords_b0}
U_{B0} = \sqrt{\frac{(r-r_{B0})(r-r_n)^{\beta_0-1}}{(r_{C0}-r)^\beta_0}}e^{\frac{t}{2\delta_{B0}}},\quad V_{B0}=-\sqrt{\frac{(r-r_{B0})(r-r_n)^{\beta_0-1}}{(r_{C0}-r)^\beta_0}}e^{-\frac{t}{2\delta_{B0}}},
\end{equation}
where $\delta_{B0} = \frac{r_{B0}}{1-\Lambda r_{B0}^2} > 0$ and $\beta_0 = \frac{r_{C0}}{(\Lambda r_{C0}^2-1)\delta_{B0}} > 1$, and the corresponding compactification $p_{B0}=\arctan(U_{B0})$, $q_{B0}=\arctan(V_{B0})$. These coordinates are valid for $0<r<r_{+\Lambda}$. The black hole horizon can be crossed using the usual arguments of the extension of the Schwarzschild-deSitter metric as for example done in \cite{GiHa77,bf86,christina}. This is illustrated in the left part of Figure \ref{constr_p}. The region beyond the matter shell (and the cosmological horizon) can be reached via the coordinates
\begin{equation} \label{coords_c3}
U_{C} = -\sqrt{\frac{(r_C-r)(r-r_n)^{\gamma-1}}{(r-r_B)^\gamma}}e^{-\frac{t}{2\delta_{C}}},\quad V_{C}=\sqrt{\frac{(r_C-r)(r-r_n)^{\gamma-1}}{(r-r_B)^\gamma}}e^{-\frac{t}{2\delta_{C}}},
\end{equation}
with $\delta_C = \frac{r_C}{\Lambda r_C^2-1} > 0$ and $\gamma = \frac{r_B}{(1-\Lambda r_B^2) \delta_C},$ $0<\gamma <1$. These coordinates extend the metric to the area $R_{0\Lambda} < r < \infty$, shown in the right part of Figure \ref{constr_p}. \par
On the connected vacuum regions the metric is given by only one expression even though vacuum extends onto several regions of $\mathscr M_2$, e.g.~regions VII, $\mathrm{VIII}_\mathrm{C}$, $\mathrm{IX}_\mathrm{C}$ and X. This implies that the coordinates $U_{B0}$, $V_{B0}$ or $U_C$, $V_C$ have to be given by the same expressions (\ref{coords_b0}) or (\ref{coords_c3}), respectively (modulo sign, cf.~\cite{GiHa77,bf86,christina}) which in turn implies that the mass parameter has to stay the same on these connected vacuum regions. For the vacuum region with $r\geq R_{0\Lambda}$ this implies $M_0^A+M_\varrho^{A_2} = M_0^B + M_\varrho^{B_1}$ (notation of Figure \ref{penrose_periodic}). On the region characterized by $r\leq r_{+\Lambda}$ this is always granted because the mass is entirely given by the black hole mass $M_0$. Finally the shift constants $C>0$ of the vacuum metric have to coincide in this region (IV and VII in Figure \ref{penrose_periodic}). They are determined by the matter shells surrounding the black hole and are equal in particular if these shells have the same shape which implies $M_\varrho^{A_1}=M_\varrho^{A_2}$.
\end{proof}

% -----------------------------------------------------------------------------------------
%                   APPENDIX
% -----------------------------------------------------------------------------------------

\appendix

\section{Proof that T acts as a contraction} \label{apptcontr}
In order to show that the operator $T$, defined in \eqref{eqdeft} acts as a contraction on the set $M$, defined in \eqref{eqdefm}, one has to check
\begin{enumerate}
\item $u \equiv y_0\in M$,
\item $u \in M \Rightarrow T u \in M$, and
\item $\exists a \in \; (0,1) \;\forall u,v\in M:\;\Ab{Tu-T v}_{\infty,\delta} \leq a \Ab{u-v}_{\infty,\delta}$, where $\Ab{.}_{\infty,\delta}=\sup_{r\in[0,\delta]}(.)$.
\end{enumerate}
\textrm{(i)}: Consider $u\equiv y_0$. Only the second critical condition
\begin{equation} \label{uismu0cond}
\frac{r^2\Lambda}{3}+\frac{\kappa}{r}\int_0^rs^2 G_\phi(s,u(s))\mathrm ds \leq c
\end{equation}
is relevant. We calculate
$$
\frac{r^2\Lambda}{3}+\frac{\kappa}{r}\int_0^rs^2 G_\phi(s,u(s))\mathrm ds \leq \frac{r^2\Lambda}{3} + \frac{\kappa r^2}{3} G_\phi(\delta,y_0) \leq \frac{\Lambda + \kappa G_\phi(r,y_0)}{3} \delta^2 \leq c
$$
for $\delta$ small enough. \\
\textrm{(ii)}: We have to guarantee that $y_0-1\leq (Tu)(r) \leq y_0+1$ and
$$
\frac{r^2\Lambda}{3}+\frac{\kappa}{r} \int_0^r s^2 G_\phi(s,T u(s))\mathrm ds \leq c.
$$
By choosing $\delta$ sufficiently small, one can achieve the domain of integration in $T$ to become arbitrarily small and these properties follow. \\
\textrm{(iii)}: We calculate
\begin{eqnarray*}
&&\Ab{Tu -T v}_{\infty,\delta} \\
&& = \Bigg\| \int_0^r \Bigg[\frac{\kappa/2}{1-\frac{s^2\Lambda}{3}-\frac{\kappa}{s}\int_0^s\sigma^2 G_\phi(\sigma,u(\sigma))\mathrm d\sigma} \Bigg(s(H_\phi(s,u(s))-H_\phi(s,v(s))) \\
&&\hspace{7cm}+ \frac{1}{s^2} \int_0^s \sigma^2 (G_\phi(\sigma,u(\sigma)) - G_\phi(\sigma,v(\sigma)))\mathrm d\sigma \Bigg) \\
&&\quad\quad + \Bigg(s H_\phi(s,v(s)) - \frac{2s\Lambda}{3\kappa} + \frac{1}{s^2}\int_0^s\sigma^2 G_\phi(\sigma,v(\sigma))\mathrm d\sigma\Bigg) \\
&&\quad\quad \times \Bigg(\frac{\kappa/2}{1-\frac{s^2\Lambda}{3}-\frac{\kappa}{s}\int_0^s\sigma^2 G_\phi(\sigma,u(\sigma))\mathrm d\sigma} - \frac{\kappa/2}{1-\frac{s^2\Lambda}{3}-\frac{\kappa}{s}\int_0^s\sigma^2 G_\phi(\sigma,v(\sigma))\mathrm d\sigma}\Bigg)\Bigg]\mathrm ds \Bigg\|_{\infty,\delta}.
\end{eqnarray*}
Since $G_\phi(r,u)$, $H_\phi(r,u)$, $\partial_u G_\phi(r,u)$, and $\partial_u H_\phi(r,u)$ are strictly in $u$ increasing functions, we have
\begin{align*}
&\sup_{u\in[y_0-1,y_0+1]} H_\phi(r,u) = H_\phi(r,y_0+1) =: H_\mathrm{sup}(r), \\
&\sup_{u\in[y_0-1,y_0+1]} G_\phi(r,u) = G_\phi(r,y_0+1) =: G_\mathrm{sup}(r), \\
&\sup_{u\in[y_0-1,y_0+1]} |\partial_u H_\phi(r,u)| = |\partial_u H_\phi(r,y_0+1)| =: G_\mathrm{sup}'(r),\\
&\sup_{u\in[y_0-1,y_0+1]} |\partial_u G_\phi(r,u)| = |\partial_u G_\phi(r,y_0+1)| =: H_\mathrm{sup}'(r).
\end{align*}
We can estimate the first summand in the following way:
\begin{eqnarray*}
&&\int_0^r \frac{\kappa/2}{1-\frac{s^2\Lambda}{3}-\frac{\kappa}{s}\int_0^s\sigma^2 G_\phi(\sigma,u(\sigma))\mathrm d\sigma}\\
&& \quad \times \Big(s( H_\phi(s,u(s))-H_\phi(s,v(s))) + \frac{1}{s^2} \int_0^s\sigma^2 (G_\phi(\sigma,u(\sigma)) - G_\phi(\sigma,v(\sigma)))\mathrm d\sigma \Big) \mathrm ds \\
&&\leq \frac{\kappa}{2(1-c)} \frac{\delta^2}{2} \left(H_\mathrm{sup}'(\delta) + \frac 1 3 G_\mathrm{sup}'(\delta)\right) \|u-v\|_{\infty,\delta}.
\end{eqnarray*}
Next, we consider the second summand:
\begin{eqnarray*}
&&\int_0^r \Big(s H_\phi(s,v(s)) - \frac{2s\Lambda}{3\kappa} + \frac{1}{s^2}\int_0^s\sigma^2 G_\phi(\sigma,v(\sigma))\mathrm d\sigma\Big) \\
&&\quad \times \Bigg(\frac{\kappa/2}{1-\frac{s^2\Lambda}{3}-\frac{\kappa}{s}\int_0^s\sigma^2 G_\phi(\sigma,u(\sigma))\mathrm d\sigma} - \frac{\kappa/2}{1-\frac{s^2\Lambda}{3}-\frac{\kappa}{s}\int_0^s\sigma^2 G_\phi(\sigma,v(\sigma))\mathrm d\sigma}\Bigg)\mathrm ds \\
&&\leq \int_0^r s \left(H_\mathrm{sup}(r) - \frac{2\Lambda}{3\kappa} + \frac{1}{3} G_\mathrm{sup}(r)\right) \frac{\kappa^2 s^2}{6(1-2c+c^2)} \mathrm ds \; G_\mathrm{sup}'(r)\; \|u-v\|_{\infty,\delta} \\
&&\leq \frac{\kappa^2 \delta^4}{24(1-2c+c^2)} \left(H_\mathrm{sup}(\delta) - \frac{2\Lambda}{3\kappa} + \frac{1}{3} G_\mathrm{sup}(\delta)\right) G_\mathrm{sup}'(\delta)\; \|u-v\|_{\infty,\delta}.
\end{eqnarray*}
So we get in total
\begin{multline*}
\|Tu - Tv\|_{\infty,\delta}\leq \Bigg(\frac{\kappa}{4(1-c)}\left(H_\mathrm{sup}'(\delta) + \frac 1 3 G_\mathrm{sup}'(\delta)\right)\delta^2 \\
+ \frac{\kappa^2}{24(1-2c +c^2)} \left(H_\mathrm{sup}(\delta) - \frac{2\Lambda}{3\kappa}+\frac 1 3 G_\mathrm{sup}(\delta)\right) G_\mathrm{sup}'(\delta) \delta^4 \Bigg) \|u-v\|_{\infty,\delta}.
\end{multline*}
If one actually wants to calculate $\delta$ one can make use of the estimate
\begin{align}
G_\phi(r,u) &= c_\ell r^{2\ell} \int_{\sqrt{1+L_0/r^2}}^\infty \phi\left(1-\varepsilon e^{-y}\right) \varepsilon^2 \left(\varepsilon^2-\left(1+\frac{L_0}{r^2}\right)\right)^{\ell+\frac{1}{2}}\mathrm d\varepsilon \nonumber \\
&\leq c_\ell r^{2\ell} \int_{1}^\infty \phi\left(1-\varepsilon e^{-y}\right) \varepsilon^2 \left(\varepsilon^2-1\right)^{\ell+\frac{1}{2}}\mathrm d\varepsilon
\end{align}
and the analogue one for $H_\phi$ to obtain a polynomial in $\delta$.

\section{Estimate of $|\varrho_\Lambda(r)-\varrho(r)|+|p_\Lambda(r)-p(r)|$} \label{apdetest}

The following calculation is valid for $r\in[0,\tilde r^*]$ where we can take for granted $1-\frac{2m(r)}{r} \geq \frac{1}{9}$ (Buchdahl inequality, cf.~\cite{and08}), $1- \frac{r^2\Lambda}{3} - \frac{2m_\Lambda(r)}{r} \geq \frac{1}{18}$ and $|y_\Lambda(r)-y(r)| \leq |y(R_0+\Delta R)|$, $\Delta R > 0$ where $R_0$ is defined to be the (first) zero of the background solution $y$. Since
\begin{multline} \label{apeest}
|\varrho_\Lambda(r)-\varrho(r)|+|p_\Lambda(r)-p(r)| \\\leq \left(\sup_{u\in[y_\Lambda(r),y(r)]} |\partial_u G_\phi(r,u)|+ \sup_{u\in[y_\Lambda(r),y(r)]} |\partial_u H_\phi(r,u)|\right)|y_\Lambda(r)-y(r)|
\end{multline}
 we calculate
\begin{eqnarray*}
&&|y_\Lambda(r) - y(r)| \leq \int_0^r |y'(s) - y'_\Lambda(s)| \mathrm ds \\
&& \leq \int_0^r\Bigg[ \underbrace{\frac{4\pi}{1-\frac{s^2\Lambda}{3}-\frac{2 m_\Lambda(s)}{s}}}_{\leq 72\pi}  \\
&& \quad \times \Bigg( \left|-\frac{s\Lambda}{12\pi}\right| + s|H_\phi(s,y_\Lambda(s))-H_\phi(s,y(s))| \\
&& \quad\quad\quad+ \underbrace{\frac{1}{s^2}\int_0^s \sigma^2 |G_\phi(\sigma,y_\Lambda(\sigma))-G_\phi(\sigma,y(\sigma))|\mathrm d\sigma}_{I_1} \Bigg) \\
&&\quad+ \left(s H_\phi(s,y(s)) + \frac{1}{s^2}\int_0^s \sigma^2 G_\phi(\sigma,y(\sigma))\mathrm d\sigma\right)\underbrace{\left(\frac{4\pi}{1-\frac{s^2\Lambda}{3}-\frac{2m_\Lambda(s)}{s}}-\frac{4\pi}{1-\frac{2m(s)}{s}}\right)}_{I_2} \Bigg]\mathrm ds.
\end{eqnarray*}
We estimate $I_1$ and $I_2$ separately:
\begin{eqnarray*}
I_1 &=& \int_0^r \frac{1}{s^2}\int_0^s \sigma^2 |G_\phi(\sigma,y_\Lambda(\sigma))- G_\phi(\sigma,y(\sigma))| \mathrm d\sigma \mathrm ds \\
&\leq& \int_0^r\int_0^r  |G_\phi(\sigma,y_\Lambda(\sigma))- G_\phi(\sigma,y(\sigma))| \mathrm d\sigma \mathrm ds \\
&\leq& r \int_0^r |G_\phi(\sigma,y_\Lambda(\sigma))- G_\phi(\sigma,y(\sigma))| \mathrm d\sigma,
\end{eqnarray*}
\begin{eqnarray*}
I_2 &=& \frac{4\pi}{1-\frac{s^2\Lambda}{3}-\frac{2m_\Lambda(s)}{s}}-\frac{4\pi}{1-\frac{2m(s)}{s}} \\
&\leq& 4\pi \cdot 18 \cdot 9 \cdot \left(\frac{s^2\Lambda}{3} + \frac{8\pi}{s}\int_0^s \sigma^2 |G_\phi(\sigma,y_\Lambda(\sigma))- G_\phi(\sigma,y(\sigma))|\mathrm d\sigma\right) \\
&\leq& 648\pi \left(\frac{s^2\Lambda}{3} + 8\pi s\int_0^s|G_\phi(\sigma,y_\Lambda(\sigma))-G_\phi(\sigma,y(\sigma))|\mathrm d\sigma\right).
\end{eqnarray*}
So using that $y$ is decreasing we have
\begin{eqnarray*}
&&|y_\Lambda(r) - y(r)| \\
&&\leq \Lambda \int_0^r \left(6s + 216\pi s^3 \left(H_\phi(r,y_0)+\frac 1 3 G_\phi(r,y_0)\right)\right) \mathrm ds \\
&& \qquad + 72\pi r\int_0^r|H_\phi(s,y_\Lambda(s))-H_\phi(s,y(s))|\mathrm ds \\
&& \qquad + \left(72\pi r + 5184\pi^2\frac{r^3}{3}\left(H_\phi(r,y_0)+\frac 1 3 G_\phi(r,y_0)\right)\right) \int_0^r |G_\phi(s,y_\Lambda(s))-G_\phi(s,y(s))| \mathrm ds \\
&&\leq \Lambda \left(3r^2 + 54\pi r^4 \left(H_\phi(r,y_0)+\frac 1 3 G_\phi(r,y_0)\right)\right) \\
&& \quad + \left(72\pi r + 1728\pi^2r^3\left(H_\phi(r,y_0)+\frac 1 3 G_\phi(r,y_0)\right)\right)\\
&&\quad\times  \int_0^r (|H_\phi(s,y_\Lambda(s))-H_\phi(s,y(s))|+|G_\phi(s,y_\Lambda(s))-G_\phi(s,y(s))|) \mathrm ds \\
&& \leq \Lambda C_1(r) + C_2(r) \int_0^r \left(\left|p_\Lambda(s)-p(s)\right|+\left|\varrho_\Lambda(s)-\varrho(s)\right|\right)\mathrm ds
\end{eqnarray*}
The derivatives with respect to $y$ of $G_\phi(r,y)$ and $H_\phi(r,y)$ are strictly increasing both in $r$ and $y$. And since $|y_\Lambda(r)-y(r)|\leq |y(R_0+\Delta R)|$ we can write
\begin{multline*}
 \left(\sup_{u\in[y_\Lambda(r),y(r)]} |\partial_u G_\phi(r,u)|+ \sup_{u\in[y_\Lambda(r),y(r)]} |\partial_u H_\phi(r,u)|\right) \\ \leq \left|\partial_u G_\phi(\tilde r^*,u)|_{y_0+|y(R_0+\Delta R)|}\right|+\left|\partial_u H_\phi(\tilde r^*,u)|_{y_0+|y(R_0+\Delta R)|}\right| =: C_3.
\end{multline*}
So we have obtained that equation \eqref{apeest} is of the form
\begin{equation}
\left|p_\Lambda(s)-p(s)\right|+\left|\varrho_\Lambda(s)-\varrho(s)\right| \leq C_4(r) \Lambda + C_5(r) \int_0^r \left(\left|p_\Lambda(s)-p(s)\right|+\left|\varrho_\Lambda(s)-\varrho(s)\right|\right)\mathrm ds\nonumber
\end{equation}
Note that $C_4(r)$ is strictly increasing. Gr\"onwall's inequality yields
\begin{equation} \label{defcgh}
(|\varrho_\Lambda(r)-\varrho(r)|+|p_\Lambda(r)-p(r)|) \leq C_4(r) e^{\int_0^r C_5(r) \mathrm ds} = C_4(r) \Lambda e^{rC_5(r)} =: C_{gh}(r) \Lambda.
\end{equation}
Note that $C_{gh}(r)$ is increasing when $r$ is increasing.

\vspace{0.2cm}
\textsc{H\r{a}kan Andr\'easson, Chalmers University of Technology and University of Gothenburg}\\
\texttt{hand@chalmers.de}\\
\vspace{-0.2cm}
\\
\textsc{David Fajman, University of Vienna}\\
\texttt{David.Fajman@univie.ac.at}\\
\vspace{-0.2cm}
\\
\textsc{Maximilian Thaller, University of Vienna}\\
\texttt{Maximilian.Thaller@arcor.de}

\end{document}